\documentclass[10pt]{article}

\usepackage[leqno]{amsmath}
\usepackage{amssymb}
\usepackage{amsthm}
\usepackage{enumerate}
\usepackage{verbatim}
\usepackage{mathrsfs}
\usepackage{dsfont}
\usepackage{graphicx}
\usepackage{tikz}
\usepackage[pdf]{pstricks}
\usepackage{float}
\usepackage{caption}
\usepackage{subcaption}

\newcommand{\tu}{\textup}
\newcommand{\ds}{\displaystyle}
\newcommand{\ts}{\textstyle}

\usepackage{soul}

\newcommand{\mr}{\mathrm}
\newcommand{\mi}{\mathit}

\newcommand{\mc}{\mathcal}
\newcommand{\mb}{\mathbf}
\newcommand{\mbb}{\mathbb}

\newcommand{\scr}{\mathscr}

\newcommand{\uhr}{\upharpoonright}

\newcommand{\tsfrac}[2]{{\ts\frac{#1}{#2}}}

\newcommand{\sq}[1]{\ensuremath{\langle#1\rangle}}

\newcommand{\eps}{\varepsilon}
\newcommand{\wt}{\widetilde}

\newcommand{\BAR}[1]{\overline{#1}}

\let\Pr\undefined

\DeclareMathOperator*{\Pr}{\mathds{P}}

\usepackage{mathtools}
\newcommand{\defeq}{\vcentcolon=}

\newcommand{\fieldfont}[1]{\mathbb{#1}}
\newcommand{\N}{\fieldfont{N}}
\newcommand{\R}{\fieldfont{R}}
\newcommand{\Z}{\fieldfont{Z}}

\newtheoremstyle{theorem-style}
  {}
  {}
  {\slshape}
  {}
  {\bf}
  {.}
  {.5em}
  {}

\newtheorem{thm}{Theorem}[section]

\newtheorem{la}[thm]{Lemma}

\newtheorem{main-la}[thm]{Main Lemma}

\newtheorem{cor}[thm]{Corollary}

\newtheorem{obs}[thm]{Observation}

\theoremstyle{definition}
\newtheorem{df}[thm]{Definition}

\newtheorem{rmk}[thm]{Remark}
\newtheorem{ex}[thm]{Example}
\newtheorem{notn}[thm]{Notation}

\usepackage{fullpage}
\usepackage{hyperref} 
\usepackage{float}

\def\[#1\]{\begin{align*}#1\end{align*}}

\newcommand{\td}{\mathsf{td}}
\newcommand{\tw}{\mathsf{tw}}
\newcommand{\bw}{\mathsf{bw}}
\newcommand{\lab}{\mr{label}}
\newcommand{\Fib}{\mr{Fib}}

\newcommand{\SUB}{\mr{SUB}}

\newcommand{\Gr}[1]{G_{#1}}
\newcommand{\A}{\scr A}
\newcommand{\B}{\scr B}
\newcommand{\C}{\scr C}

\newcommand{\proj}[2]{\mr{proj}_{#1}(#2)}
\newcommand{\rest}{\mr{rest}}

\renewcommand{\P}[1]{\mathsf{Pathset}_{#1}}
\newcommand{\Rel}[1]{\mathsf{Relation}_{#1}}

\newcommand{\RHO}[3]{\rest_{#1}(#2|#3)}

\newcommand{\GG}[1]{G_{#1}}

\newcommand{\ACzero}{\cc{AC}^{\,\cc 0}}
\newcommand{\NCone}{\cc{NC}^{\,\cc 1}}

\newcommand{\cc}{\textsl}
\newcommand{\gad}{\rotatebox[origin=c]{180}{$\dag$}}
\newcommand{\smallgad}{\rotatebox[origin=c]{180}{$\scriptstyle\dag$}}

\newtheorem{problem}{Problem}
\usepackage[font={small,sl},labelfont={sf,up},width=6in]{caption}

\usepackage{bm}
\renewcommand{\mb}{\bm}

\begin{document}

\title{Tree-depth and the Formula Complexity of Subgraph Isomorphism}
\author{Deepanshu Kush \\ University of Toronto \and Benjamin Rossman \\ Duke University}
\date{\today}
\maketitle{}

\begin{abstract}
For a fixed ``pattern'' graph $G$, the {\em colored $G$-subgraph isomorphism problem} (denoted $\SUB(G)$) asks, given an $n$-vertex graph $H$ and a coloring $V(H) \to V(G)$, whether $H$ contains a properly colored copy of $G$.
The complexity of this problem 
is tied to parameterized versions of $\cc{P}$ ${=}?$ $\cc{NP}$ and $\cc{L}$ ${=}?$ $\cc{NL}$, among other questions.
An overarching goal is to understand the complexity of $\SUB(G)$, under different computational models, in terms of natural invariants of the pattern graph $G$. 

In this paper, we establish a close relationship between the {\em formula complexity} of $\SUB(G)$ and an invariant known as {\em tree-depth} (denoted $\td(G)$). 
$\SUB(G)$ is known to be solvable by
monotone $\ACzero$ formulas 
of size $O(n^{\td(G)})$. 
Our main result is an $n^{\wt\Omega(\td(G)^{1/3})}$ lower bound
for formulas that are monotone {\em or} have sub-logarithmic depth.
This complements a lower bound of Li, Razborov and Rossman \cite{li2017ac} relating tree-width and $\ACzero$ circuit size. 
As a corollary, it implies a stronger homomorphism preservation theorem for first-order logic on finite structures \cite{RossmanHPT}.

The technical core of this result is an 
$n^{\Omega(k)}$
lower bound in the special case where $G$ is a complete binary tree of height $k$, which we establish using the {\em pathset framework} introduced in \cite{rossman2018formulas}.  (The lower bound for general patterns follows via a recent excluded-minor characterization of tree-depth \cite{czerwinski2019improved,KR18}.)
Additional results of this paper
extend the pathset framework and 
improve upon both, the best known upper and lower bounds on the average-case formula size of $\SUB(G)$ when $G$ is a path.
\end{abstract}

\tableofcontents{}

\section{Introduction}

Let $G$ be a fixed ``pattern'' graph. In the \textsc{colored $G$-subgraph isomorphism problem}, denoted $\SUB(G)$, we are given an $n$-vertex ``host'' graph $H$ and a vertex-coloring $c : V(H) \to V(G)$ as input and required to determine whether or not $H$ contains a properly colored copy of $G$ (i.e.,\ a subgraph $G' \subseteq H$ such that the restriction of $c$ to $V(G')$ constitutes an isomorphism from $G'$ to $G$). This general problem includes, as special cases, several important problems in parameterized complexity.  In particular, $\SUB(G)$ is equivalent (up to $\ACzero$ reductions) to the \textsc{$k$-clique} and \textsc{distance-$k$ connectivity} problems when $G$ is a clique or path of order $k$.  

For any fixed pattern graph $G$, the problem $\SUB(G)$ is solvable by brute-force search in polynomial time $O(n^{|V(G)|})$.  Understanding the fine-grained complexity of $\SUB(G)$ 
--- in this context, we mean the exponent of $n$ in the complexity of $\SUB(G)$ under various computational models --- for general patterns $G$ is an important challenge that is tied to major open questions including $\cc{P} \mathrel{{=}{?}} \cc{NP}$, $\cc{L} \mathrel{{=}{?}} \cc{NL}$, $\NCone \mathrel{{=}{?}} \cc{L}$, and their parameterized versions ($\cc{FPT} \mathrel{{=}{?}} \cc{W}[1]$, etc.)  An overarching goal is to bound the fine-grained complexity of $\SUB(G)$ in terms of natural invariants of the graph $G$.

Two key invariants arising in this connection are tree-width ($\tw$) and tree-depth ($\td$).
The {\em tree-depth} of $G$ is the minimum height of a rooted forest whose ancestor-descendant closure contains $G$ as a subgraph. 
This invariant has a number of equivalent characterizations and plays a major role in structural graph theory and parameterized complexity \cite{nevsetvril2006tree}.
{\em Tree-width} is even more widely studied in graph theory and parameterized complexity \cite{diestel2005graph}.  It is defined in terms of a different notion of tree decomposition and provides a lower bound on tree-depth ($\tw+1 \le \td$).

These two invariants provide well-known upper bounds on the {\em circuit size} and {\em formula size} of $\SUB(G)$.  To state this precisely, we regard $\SUB(G)$ as a sequence of boolean functions $\{0,1\}^{|E(G)|{\cdot}n^2} \to \{0,1\}$ where the input encodes a host graph $H$ with vertex set $V(G) \times \{1,\dots,n\}$ under the vertex-coloring that maps $(v,i)$ to $v$.  (Restricting attention to this class of inputs is without loss of generality.)  Throughout this paper, we consider circuits and formulas in the unbounded fan-in basis $\{\mr{AND}_\infty,\mr{OR}_\infty,\mr{NOT}\}$; we measure {\em size} of both circuits and formulas by the number of gates.  A circuit or formula is {\em monotone} if it contains no $\mr{NOT}$ gates.  We use $\ACzero$ as an adjective that means ``depth $O(1)$'' in reference to upper bounds and ``depth $o(\log n)$'' in reference to lower bounds on formula size.\footnote{Here and elsewhere, asymptotic notation hides constants that may depend on $G$. In other contexts, e.g.\ $\Omega(\td(G))$, hidden constants are absolute.}

\begin{thm}[Folklore upper bounds]\label{thm:folklore}
For all pattern graphs $G$, $\SUB(G)$ is solvable by monotone $\ACzero$ circuits (respectively, formulas) of size $O(n^{\tw(G)+1})$ (respectively, $O(n^{\td(G)})$).
\end{thm}

It is conjectured that $\SUB(G)$ requires circuit size $n^{\Omega(\tw(G))}$ for all graphs $G$; if true this would imply $\mi{FPT} \ne \mi{W}[1]$ and $\mi{P} \ne \mi{NP}$ in a very strong way.
As evidence for this conjecture, Marx \cite{marx2010can} proved a conditional $n^{\Omega(\tw(G)/\log\tw(G))}$ lower bound assuming the Exponential Time Hypothesis.
Providing further evidence, Li, Razborov and Rossman \cite{li2017ac} established an unconditional $n^{\Omega(\tw(G)/\log\tw(G))}$ lower bound for $\ACzero$ circuits, via a technique that extends to (unbounded depth) monotone circuits.  
This result is best stated in terms of a certain graph invariant $\kappa(G)$ introduced in \cite{li2017ac}:

\begin{thm}[Lower bound on the restricted circuit size of $\SUB(G)$ \cite{li2017ac}]\label{thm:LRR}
For all pattern graphs $G$, the circuit size of $\SUB(G)$ --- in both the $\ACzero$ and monotone settings --- is at least 
$n^{\kappa(G)-o(1)}$ 
where $\kappa(G)$ is a 
graph invariant satisfying 
$\Omega(\tw(G)/\log \tw(G)) \le \kappa(G) \le \tw(G)+1$.\footnote{\normalfont 
It is actually shown that $\kappa(G)$ is at most the {\em branch-width} of $G$, an invariant related to tree-width by $\frac23(\tw+1) \le \bw \le \tw+1$.
The relationship between $\kappa(G)$ and $\tw(G)$ was further investigated by Rosenthal \cite{rosenthal2019}, who identified the separating example $\kappa(Q) = \Theta(\tw(Q)/\sqrt{\log\tw(Q)})$ for hypercubes $Q$.}
\end{thm}

Shifting our focus from circuits to formulas, it is natural to conjecture that $\SUB(G)$ requires formula size $n^{\Omega(\td(G))}$. 
This statement generalizes the prominent conjecture that \textsc{distance-$k$ connectivity} requires formula size $n^{\Omega(\log k)}$, which as a consequence implies $\NCone \ne \cc{NL}$.  (There is also an average-case version of this conjecture which implies $\NCone \ne \cc{L}$, as we explain shortly.)

In this paper, we carry out the final step in the proof of an 
analogous result to Theorem \ref{thm:LRR} that lower bounds the restricted formula size of $\SUB(G)$ in terms of an invariant $\tau(G)$ that is polynomially related to tree-depth:  

\begin{thm}[Lower bound on the restricted formula size of $\SUB(G)$]\label{thm:main-lb}
For all patterns graphs $G$, the formula size of $\SUB(G)$ --- in both the $\ACzero$ and monotone settings --- is at least $n^{\tau(G)-o(1)}$ 
where $\tau(G)$ is a graph invariant satisfying $\wt\Omega(\td(G)^{1/3}) \le \tau(G) \le \td(G)$.
\end{thm}

The invariant $\tau(G)$ was introduced in \cite{RossmanICM}, where  it was also shown that $n^{\tau(G)-o(1)}$ is a lower bound on the formula size of $\SUB(G)$ in the $\ACzero$ and monotone settings. 
The results of \cite{RossmanICM} generalized lower bounds for $\SUB(P_k)$ from papers 
\cite{rossman2015correlation,rossman2018formulas}, which showed that $\tau(P_k) = \Omega(\log k)$ (where $P_k$ is the path graph of length $k$).
As we will explain shortly, this lower bound for $\tau(P_k)$ implies that $\tau(G) = \Omega(\log \td(G))$ for all graphs $G$.  
The contribution of the present paper lies in improving this logarithmic lower bound to a polynomial one by showing $\tau(G) = \wt\Omega(\td(G)^{1/3})$.

\begin{rmk}
It is helpful to keep in mind the related inequalities:
\[
  \text{circuit size} \le \text{formula size},\qquad  \tw+1 \le \td,\qquad  \kappa \le\tau.
\] 
It is further known that $\td(G) \le (\tw(G) + 1)\log|V(G)|$ \cite{nevsetvril2006tree}.  A nearly maximal separation between invariants $\td$ and $\tw$ is witness by {\em bounded-degree trees} $T$, which have tree-width $1$ but tree-depth $\Omega(\log|V(T)|)$.  
This class includes paths and complete binary trees, the two families of pattern graphs studied in this paper.
\end{rmk}

For trees $T$, we point out that $\SUB(T)$ is computable by monotone $\ACzero$ circuits of size $c(T) \cdot n^2$ for a constant $c(T)$ depending on $T$.  (This follows from Theorem \ref{thm:folklore}, since all trees have tree-width $1$.)
Although formulas are a weaker model than circuits, establishing formula lower bounds for $\SUB(T)$ of the form $n^{\Omega(\log|V(T)|)}$, as we do in this paper, is a subtle task which requires techniques that distinguish formulas from circuits.  Accordingly, Theorem \ref{thm:main-lb} involves greater machinery than Theorem \ref{thm:LRR}.  The invariant $\tau(G)$ is also significantly harder to define and analyze compared to $\kappa(G)$.

\subsection{Minor-monotonicity}\label{sec:minormon}

Recall that a graph $F$ is a {\em minor} of $G$ if $F$ can be obtained from $G$ by a sequence of edge deletions and contractions (i.e.,\ remove an edge and identify its two endpoint).
A graph invariant $p$ 
is said to be {\em minor-monotone} if $p(F) \le p(G)$ whenever $F$ is a minor of $G$.  
As observed in \cite{li2017ac}, the complexity of $\SUB(G)$ (under any reasonable class of circuits) is minor-monotone in the following sense:

\begin{la}\label{la:mon-proj}
If $F$ is a minor of $G$, then there is a reduction from $\SUB(F)$ to $\SUB(G)$ via a monotone projection.\footnote{\normalfont 
That is, for every $n$, there is a reduction from $\SUB(F)$ to $\SUB(G)$, viewed as boolean functions $\{0,1\}^{|E(F)|{\cdot}n} \to \{0,1\}$ and $\{0,1\}^{|E(G)|{\cdot}n} \to \{0,1\}$, via a monotone projection that maps each variable of $\SUB(G)$ to a variable of $\SUB(F)$ or a constant $0$ or $1$.
}
\end{la}

In the quest to characterize the complexity of $\SUB(G)$ in terms of invariants of $G$, it makes sense to focus on minor-monotone ones.  Indeed, invariants $\tw$, $\td$, $\kappa$, $\tau$ are all minor-monotone. 
This feature is useful in bounding the complexity of $\SUB(G)$. For example, we can combine the result of \cite{chuzhoy2019towards} that every graph with tree-width at least $k^9\,\mr{polylog}\,k$ contains a $(k \times k)$-grid minor,
with the lower bound $\kappa(\text{($k \times k$)-grid graph}) = \Omega(k)$ from \cite{amano2010k}, in order to conclude that $\kappa(G) = \wt\Omega(\tw(G)^{1/9})$ for all graphs $G$.
(Notation $\wt O(\cdot)$ and $\wt\Omega(\cdot)$ suppresses poly-logarithmic factors.)
The stronger $\kappa(G) = \Omega(\tw(G) / \log \tw(G))$ bound of Theorem \ref{thm:LRR} is obtained by a more nuanced analysis of the invariant $\kappa$.

In a similar manner, we can combine the fact that every graph $G$ contains a path of length $\td(G)$ \cite{nevsetvril2006tree}, with the lower bound $\tau(P_k) = \Omega(\log k)$ \cite{rossman2018formulas}, in order to conclude that $\tau(G) = \Omega(\log\td(G))$ for all graphs $G$.
With the goal of improving this lower bound to $\Omega(\mr{poly}\,\td(G))$ (that is, $\Omega(\td(G)^\eps)$ for some constant $\eps > 0$), Kawarabayashi and Rossman \cite{KR18} established a polynomial excluded-minor characterization of tree-depth, which was subsequently sharpened by Czerwi{\'n}ski, Nadara and Pilipczuk \cite{czerwinski2019improved}.

\begin{thm}[Excluded-minor characterization of tree-depth \cite{KR18,czerwinski2019improved}]\label{thm:excluded-minor}
Every graph $G$ with tree-depth $\Omega(k^3)$ satisfies at least one of the following:
\begin{enumerate}[\normalfont\quad (i)]
\item
$G$ has tree-width $\ge k$,
\item
$G$ contains a path of length $2^k$,
\item
$G$ contains a $T_k$-minor, where $T_k$ is the complete binary tree of height $k$.
\end{enumerate}
\end{thm}

Theorem \ref{thm:excluded-minor} reduces the task of proving $\tau(G) = \Omega(\mr{poly}\,\td(G))$ to the task of proving $\tau(T_k) = \Omega(\mr{poly}\,k)$.  It is this final step that we tackle in this paper.\footnote{Theorem \ref{thm:Tk} delivers on a promise in papers \cite{KR18,RossmanHPT,RossmanICM}, which cite $\tau(T_k) = \Omega(\mr{poly}\,k)$ as an unpublished result of upcoming work.
Let us mention that, after finding many devils in the details of an earlier sketch of an $\Omega(\sqrt k)$ bound by the second author, we worked out an entirely different approach in this paper, which moreover gives a {\em linear} lower bound (which is tight up to a constant since $\tau(T_k) \le \td(T_k) = k$).}

\begin{thm}[Main result of this paper]
\label{thm:Tk}
$\tau(T_k) = \Omega(k)$. 
\end{thm}

This lower bound is proved in Section \ref{sec:Tk} using a certain potential function (described in Section \ref{sec:prelims}), which further reduces our task to a combinatorial problem concerning {\em join-trees} over $T_k$, that is, rooted binary trees whose leaves are labeled by edges of $T_k$.  This is the same combinatorial framework as the $\tau(P_k) = \Omega(\log k)$ lower bound of \cite{rossman2018formulas}; however, the task of analyzing join-trees over $T_k$ turned out to be significantly harder compared with $P_k$. 

Theorems \ref{thm:excluded-minor} and \ref{thm:Tk} combine to prove Theorem \ref{thm:main-lb} (the bound $\tau(G) = \wt\Omega(\td(G)^{1/3})$) as follows.  For a graph $G$ with tree-depth $\Omega(k^3)$, we can see that $\tau(G) = \Omega(k/\log k)$ by considering the three cases given by Theorem \ref{thm:excluded-minor}:
\begin{enumerate}[\normalfont\quad (i)]
\item
If $G$ has tree-width $\ge k$, then
$
  \tau(G) \ge \kappa(G) = \Omega(k/\log k)
$
by Theorem \ref{thm:LRR}.
\item
If $G$ contains a path of length $2^k$, then
$
  \tau(G) \ge \tau(P_{2^k}) = \Omega(k)
$
by the lower bound of \cite{rossman2018formulas}.
\item
If $G$ contains a $T_k$-minor, then 
$
\tau(G) \ge \tau(T_k) = \Omega(k)
$
by Theorem \ref{thm:Tk}.
\end{enumerate}

\subsection{Corollary in finite model theory}

Theorem \ref{thm:main-lb} has a striking consequence in finite model theory, observed in the paper \cite{RossmanHPT}.  

\begin{cor}[Polynomial-rank homomorphism preservation theorem over finite structures]\label{cor:HPT}
Every first-order sentence of quantifier-rank $r$ that is preserved under homomorphisms of finite structures is logically equivalent on finite structures to an existential-positive first-order sentence of quantifier-rank $\wt O(r^3)$.
\end{cor}

The polynomial upper bound of Corollary \ref{cor:HPT} improves an earlier {\em non-elementary} upper bound of \cite{rossman2008homomorphism}.  
This surprising connection between circuit complexity and finite model theory was in fact the original motivation behind Theorems \ref{thm:main-lb} and \ref{thm:excluded-minor}, as well as the present paper.

\subsection{Improved bounds for average-case $\SUB(P_k)$}

Additional results of this paper improve both the average-case upper and lower bounds for $\SUB(P_k)$ \cite{rossman2015correlation}.  
Here {\em average-case} refers to the $p$-biased product distribution on $\{0,1\}^{kn^2}$ where $p = n^{-(k+1)/k}$.  
This input distribution corresponds to a random graph $\mb X$, comprised of $k+1$ layers of $n$ vertices, where every pair of vertices in adjacent layers is connected by an edge independently with probability $p$.  
For this choice of $p$, the probability that $\mb X$ contains a path of length $k$ containing one vertex from each layer is bounded away from $0$ and $1$.

\begin{thm}[\cite{rossman2018formulas}]\label{thm:previous}
$\SUB(P_k)$ is solvable on $\mb X$ with probability $1 - o(1)$ by monotone $\ACzero$ formulas of size $n^{\frac12\lceil \log_2(k) \rceil + o(1)}$.  
On the other hand, $\ACzero$ formulas solving $\SUB(P_k)$ on $\mb X$ with probability $\ge 0.9$ require size $n^{\tau(P_k) - o(1)}$ where $\tau(P_k) \ge \frac12\log_{\sqrt{13}+1}(k)$ \tu($\ge 0.22\log_2(k)$\tu).
\end{thm}

A similar average-case lower bound for (unbounded depth) monotone formulas was subsequently shown in \cite{rossman2015correlation}.  Precisely speaking, that paper gives an $n^{\frac12\tau(P_k) - o(1)}$ lower bound under $\mb X$, as well as an $n^{\tau(P_k) - o(1)}$ lower bound under the distribution that, half of the time, is $\mb X$ and, the other half, is a uniform random path of length $k$ with no additional edges.

\subsubsection{Upper bound}\label{sec:ub-sketch}

The {\em average-case} upper bound of Theorem \ref{thm:previous} can be recast, in stronger terms, as a {\em worst-case randomized} upper bound for the problem of multiplying $k$ $(n\times n)$-permutation matrices $Q_1,\dots,Q_k$. 
This problem is solvable by deterministic (non-randomized) $\ACzero$ formulas of size $n^{\log_2(k) + O(1)}$ via the classic ``recursive doubling'' procedure: recursively compute matrix products $L \defeq Q_1\cdots Q_{\lceil k/2 \rceil}$ and $R \defeq Q_{\lceil k/2 \rceil+1}\cdots Q_k$ and then obtain $Q_1\cdots Q_k = LR$ by a single matrix multiplication.

Randomization lets us achieve quadratically smaller formula size $n^{\frac12\log_2(k) + O(1)}$.  The idea is as follows.  Generate $m \defeq \wt O(\sqrt n)$ independent random sets $\mb I_1,\dots,\mb I_m \subseteq [n]$, each of size $\sqrt n$.  Rather than compute all entries of the permutation matrix $L$ using $n^2$ subformulas, we will encode the information in $L$ more efficiently using $(2\log n + 1)m^2 = \wt O(n)$ subformulas (note that $\log(n!) = O(n\log n)$ bits are required to encode a permutation matrix). 
For each $(r,s) \in [m]^2$, we recursively construct
\begin{itemize}
\item
one subformula that indicates\footnote{When describing the behavior of randomized formulas in this subsection (using verbs like ``indicate'', ``output'', etc.), we leave implicit that the description holds {\em correctly with high probability} for any input.} whether or not there exists a {\em unique} pair $(a,b) \in \mb I_r \times \mb I_s$ such that $L_{a,b} = 1$, and
\item
$2\log n$ additional subformulas that give the binary representation of $a$ and $b$ whenever such $(a,b)$ uniquely exists.
\end{itemize}
Similarly, with respect to the permutation matrix $R$, for each $(s,t) \in [m]^2$, we have $2\log n + 1$ recursively constructed subformulas that indicate whether there exists a {\em unique} pair $(b,c) \in \mb I_s \times \mb I_t$ such that $R_{b,c} = 1$, and if so, give the binary representation of $b$ and $c$.
Using these subformulas for subproblems $L$ and $R$, we construct the corresponding formulas for $Q_1\cdots Q_k$ which, for each $(r,t) \in [m]^2$, indicate whether there exists a {\em unique} pair $(a,c) \in \mb I_r \times \mb I_t$ such that $R_{a,c} = 1$, and if so, give the binary representation of $a$ and $c$.
These formulas check, for each $s \in [m]$, whether the $(r,s)$- and $(s,t)$-subformulas of the $L$- and $R$-subproblems output $(a,b)$ and $(b',c)$, respectively such that $b = b'$.  These formulas are therefore larger than the subformulas for subproblems $L$ and $R$ by a factor $\wt O(m)$.  This implies an upper bound $\wt O(m)^{\lceil\log_2(k)\rceil} = n^{\frac12\log_2(k)+O(1)}$ on size and $O(\log k)$ on depth of the resulting randomized $\ACzero$ formulas.

A similar construction solves $\SUB(P_k)$ in the average-case. This yields an upper bound on $\frac12\log k + O(1)$ on the parameter $\tau(P_k)$, which we initially guessed might be optimal. However, in the course of trying prove a matching lower bound, we were surprised to discover a better upper bound!

\begin{thm}\label{thm:Pkupper}
There exist randomized $\ACzero$ formulas of size $n^{\frac13\log_\varphi(k)+O(1)}$ $(\le n^{0.49\log_2(k) + O(1)})$, where $\varphi = (\sqrt 5 + 1)/2$ is the golden ratio, which compute the product of $k$ permutation matrices.
\end{thm}

The algorithm generalizes the randomized ``recursively doubling'' method outlined above.   Here we give a brief sketch (full details are given in Section \ref{sec:Pkupper}).  Let $k = \Fib(\ell)$ where $\ell \ge 3$ (i.e.,\ the $\ell^{\text{th}}$ Fibonacci number, which satisfies $\Fib(\ell)=\Fib(\ell-1)+\Fib(\ell-2)$).  We will represent information about the product $Q_1\cdots Q_k$ by constructing formulas that enumerate all triples $(a,c,d) \in [n]^3$ such that 
\begin{equation}\label{eq:Q}
  (Q_1\cdots Q_{\Fib(\ell-1)})_{a,c} = (Q_{\Fib(\ell-1)+1}\cdots Q_k)_{c,d} = 1.
\end{equation}
This is accomplished by generating $m \defeq \wt O(n^{1/3})$ independent random sets $\mb I_1,\dots,\mb I_m$, each of size $n^{2/3}$, and recording the {\em unique} triples $(a,c,d) \in \mb I_r \times \mb I_t \times \mb I_u$ for which (\ref{eq:Q}) holds.

The recursive construction breaks into a ``left'' subproblem on $(Q_1,\dots,Q_{\Fib(\ell-1)})$ and a ``right'' subproblem on $(Q_{\Fib(\ell-2)+1},\dots,Q_k)$.\footnote{The ``right'' subproblem on $(Q_{\Fib(\ell-2)+1},\dots,Q_k)$ can also be viewed as a ``left'' subproblem on $(P_1,\dots,P_{\Fib(\ell-1)})$ where $P_i$ is the transpose of $Q_{k-i+1}$.}  (In contrast to the ``recursive doubling'' method, here the ``left'' and ``right'' subproblems involve overlapping subsequences of permutation matrices.)
In the ``left'' subproblem: for each $(r,s,t) \in [m]^3$, we have
\begin{itemize}
\item
$3\log n+1$ subformulas that indicate whether there exists a {\em unique} triple $(a,b,c) \in \mb I_r \times \mb I_s \times \mb I_t$ such that $(Q_1\cdots Q_{\Fib(\ell-2)})_{a,b} = (Q_{\Fib(\ell-2)+1}\cdots Q_{\Fib(\ell-1)})_{b,c} = 1$, and if so, give the binary representation of $a,b,c$.
\end{itemize}
In the ``right'' subproblem: for each $(r,s,t) \in [m]^3$, we have
\begin{itemize}
\item
$3\log n+1$ subformulas that indicate whether there exists a {\em unique} triple $(b,c,d) \in \mb I_s \times \mb I_t \times \mb I_u$ such that $(Q_{\Fib(\ell-2)+1}\cdots Q_{\Fib(\ell-1)})_{b,c} = (Q_{\Fib(\ell-1)+1}\cdots Q_k)_{c,d} = 1$, and if so, give the binary representation of $b,c,d$.
\end{itemize}
The subformulas in the ``left'' and ``right'' subproblems may be combined to produce the analogous (left-handed) formulas for the original input $(Q_1,\dots,Q_k)$: for each $(r,t,u) \in [m]^3$, we construct
\begin{itemize}
\item
$3\log n+1$ subformulas that indicate whether there exists a {\em unique} triple $(a,c,d) \in \mb I_r \times \mb I_t \times \mb I_u$ such that (\ref{eq:Q}) holds,
and if so, give the binary representation of $a,c,d$.
\end{itemize}
These formulas check, for each $s \in [m]$, whether the $(r,s,t)$- and $(s,t,u)$-subformulas in the ``left'' and ``right'' subproblems output triples $(a,b,c)$ and $(b',c',d)$, respectively, such that $b=b'$ and $c=c'$. 
These formulas are therefore larger than the subformulas in the ``left'' and ``right'' subproblems by a factor $\wt O(m)$.  Taking $k = \Fib(3) = 2$ as our base case with formula size $n^{O(1)}$, this gives an upper bound $\wt O(m)^{\ell-3} \cdot n^{O(1)} = n^{\frac13\log_\varphi(k)+O(1)}$ for all $k = \Fib(\ell)$ (which extends as well to non-Fibonacci numbers $k$).\bigskip

In Section \ref{sec:Pkupper}, we introduce a broad class of randomized algorithms (based on a simplification of the pathset complexity measure) that generalize both the ``recursive doubling'' and ``Fibonacci overlapping'' algorithms outlined above.  We also discuss reasons, including experimental data, which suggest that $n^{\frac13\log_\varphi(k)+O(1)}$ might in fact be the asymptotically {\em tight} bound on the randomized formula size of multiplying $k$ permutations.

\subsubsection{Lower bound}

The final result of this paper improves the 
$\tau(P_k) \ge \frac12\log_{\sqrt{13}+1}(k)$ \tu($\ge 0.22\log_2(k)$\tu) lower bound of Theorem~\ref{thm:previous}.

\begin{thm}\label{thm:Pklower}
$\tau(P_k) \ge \log_{\sqrt 5 + 5}(k) - 1$ \tu($\ge 0.35\log_2(k) - 1$\tu)
\end{thm}

More significant than the quantitative improvement we obtain in Theorem \ref{thm:Pklower} is the fact that our proof further develops {\em pathset framework} by introducing a new potential function that gives stronger lower bounds on $\tau(G)$.  This development and the proof of Theorem \ref{thm:Pklower} are presented in detail in Sections \ref{sec:Phi}, \ref{sec:pathset} and \ref{sec:Pk}.  

Since $\frac13\log_\varphi(k) = \log_{\sqrt 5 + 2}(k)$, our upper and lower bounds are off by exactly $3$ in the base of the logarithm.  It would be very interesting to completely close this gap.

\subsection{Related work}

There have been several papers, including \cite{cygan2016tight,krauthgamer2017conditional,marx2010can}, which give conditional lower bounds (under ETH and other assumptions) on the {\em circuit size} of $\SUB(G)$ and its uncolored variant.
We are not aware of any conditional hardness results for the {\em formula size} of $\SUB(G)$.
It would be interesting to show that $\SUB(G)$ requires (unrestricted) formula size $n^{\Omega(\td(G))}$ under a natural assumption.

\section{Preliminaries}\label{sec:prelims}

For a natural number $n$, $[n]$ denotes the set $\{1,\dots,n\}$.
For simplicity of presentation, we occasionally omit floors and ceilings, e.g.,\ treating quantities like $\sqrt n$ as natural numbers).  This is always without loss of parameters in our results.
When no base is indicated, $\log(\cdot)$ denotes the base-2 logarithm.

\subsection{Graphs}

In this paper, {\em graphs} are simple graphs, i.e., pairs $G = (V(G),E(G))$ where $V(G)$ is a set and $E(G)$ is a subset of $\binom{V(G)}{2}$ (the set of unordered pairs $\{v,w\}$ where $v,w$ are distinct elements of $V(G)$).  Unless explicitly stated otherwise, graphs are assumed to be locally finite (i.e.,\ every vertex has finite degree) and without isolated vertices (i.e.,\ $V(G) = \bigcup_{e \in E(G)} e$).  For a vertex $v \in V(G)$, $\deg_G(v)$ or simply $\deg(v)$ denotes the degree of $v$ in $G$.

We regard $G$ as a fixed (possibly infinite) ``pattern'' graph. $F$ shall consistently denote a finite subgraph of $G$. We write $\subseteq$ for the subgraph relation and $\subset$ (or sometimes $\subsetneqq$) for the proper subgraph relation.  If $F$ is a subgraph of $G$, then $G \setminus F$ denotes the graph with edge set $E(G) \setminus E(F)$ (and no isolated vertices).

Two important graphs in this paper are paths and complete binary trees.  $P_k$ denotes the path graph of length $k$ (with $k+1$ vertices and $k$ edges).  $T_k$ denotes the complete binary tree of height $k$ (with $2^{k+1}-1$ vertices and $2^{k+1}-2$ edges).
We also consider infinite versions of these graphs.  $P_\infty$ is the path graph with vertex set $\Z$ and edge set $\{(i,i+1) : i \in \Z\}$.  $T_\infty$ is the union $\bigcup_{k = 1}^\infty T_k$ under the nesting $T_1 \subset T_2 \subset T_3 \subset \cdots$ where $\mr{Leaves}(T_1) \subset \mr{Leaves}(T_2) \subset \mr{Leaves}(T_3) \subset \cdots$.  Thus, $T_\infty$ is an infinite, rootless, layered binary tree, with leaves in layer $0$, their parents in level $1$, etc.

We use terms {\em graph invariant} and {\em graph parameter} interchangeably in reference to real-valued functions on graphs that are invariant under isomorphism.

\subsection{Threshold weightings}

We describe a family of edge-weightings on graphs $G$, which in the case of finite graphs correspond to product distributions that are balanced with respect to the problem $\SUB(G)$.  (Definitions in this section are adapted from \cite{li2017ac}.)

\begin{df}
For any graph $G$ and function $\theta : E(G) \to \R$, we denote by $\Delta_\theta : \{$finite subgraphs of $G\} \to \R$ the function
\[
  \Delta_\theta(F) \defeq |V(F)| - \sum_{e \in E(F)} \theta(e).
\]
\end{df}

\begin{df}
A {\em threshold weighting} for a graph $G$ is a function $\theta : E(G) \to [0,2]$ such that $\Delta_\theta(F) \ge 0$ for all finite subgraphs $F \subseteq G$; if $G$ is finite, we additionally require that $\Delta_\theta(G) = 0$.

We refer to the pair $(G,\theta)$ as a {\em threshold-weighted graph}.  When $\theta$ is fixed, we will at times simply write $\Delta(\cdot)$ instead of $\Delta_\theta(\cdot)$.
\end{df}


\begin{df}
A {\em Markov chain} on a graph $G$ is a matrix $[0,1]^{V(G) \times V(G)}$ that satisfies 
\begin{itemize}
\item
$\sum_{w \in V(G)} M_{v,w} = 1$ for all $v \in V(G)$ and 
\item
$M_{v,w} > 0 \,\Longrightarrow\, \{v,w\} \in E(G)$ for all $v,w \in V(G)$.
\end{itemize}
\end{df}

\begin{la}\label{la:markov}
Every Markov chain $M$ on $G$ induces a threshold weighting $\theta$ on $G$ defined by 
\[
  \theta(\{v,w\}) \defeq M_{v,w} + M_{w,v}.
\]
This threshold weighting satisfies
\[
  \Delta_\theta(F) = \sum_{v \in V(F)}\, \sum_{w \in V(G) \,:\, \{v,w\} \notin E(F)} M_{v,w}.
\]
\end{la}

We remark that this lemma has a converse (shown in \cite{rosenthal2019}): {\em Every threshold weighting on $G$ is induced by a (not necessarily unique) Markov chain on $G$.}  Lemma \ref{la:markov} also gives us a way to define {\em threshold weightings} when $G$ is an infinite graph; this will be useful later on.

\begin{ex}\label{ex:theta}
Let $M$ be the transition matrix of the uniform random walk on $T_k$ where $k \ge 2$. That is,
\[
  M_{v,w} \defeq \begin{cases}
  1/\mr{deg}(v) &\text{if } \{v,w\} \in E(T_k),\\
  0 &\text{otherwise.}
  \end{cases}
\]
For the associated threshold weighting $\theta : E(T_k) \to [0,2]$, we have 
\[
  \theta(e) = \begin{cases}
    4/3 
    &\text{if $e$ contains a leaf},\\
    5/6 
    &\text{if $e$ contains the root},\\
    2/3 
    &\text{otherwise.}
  \end{cases}
\]
A key property of this $\theta$ that we will use later on (Lemma \ref{la:partial}) is that 
\[
  \Delta_\theta(F) \ge \frac{|V(F) \cap V(T_k \setminus F)|}{3}
\]
(that is, $\Delta_\theta(F)$ is at least one-third the size of the boundary of $F$)
for all graphs $F \subseteq T_k$.  This is a straightforward consequence of Lemma \ref{la:markov}, which is also true in the infinite tree $T_\infty$.
\end{ex}

\begin{ex}\label{ex:theta2}
Let $P_k$ be the path of length $k$ (with $k+1$ vertices and $k$ edges).  The constant function $\theta \equiv 1 + \frac{1}{k}$ is a threshold weighting for $P_k$.  (This is different from the threshold function induced by the uniform random walk on $P_k$, which has value $1/2$ on the two outer edges of $P_k$ and value $1$ on the inner edges.)

This example again makes sense for $k = \infty$.  The constant function $E(P_\infty) \mapsto \{1\}$ is a threshold weighting for $P_\infty$.  This threshold function has the nice property that 
\[
  \Delta(F) = |V(F)| - |E(F)| = \#\{\text{connected components of }F\}
\] 
for all finite subgraphs $F \subset P_\infty$.
\end{ex}

\begin{df}
\label{df:X}
Let $G$ be a finite graph, let $\theta$ be a threshold weigting on $G$, and let $n$ be a positive integer. We denote by $\mb X_{\theta,n}$ be the random $V(G)$-colored graph (i.e.,\ 
input distribution to $\SUB(G)$) with vertex set $V(G) \times [n]$, vertex-coloring $(v,i) \mapsto v$,
and random edge relation given by
\[
  \Pr[\, \{(v,i),(w,j)\} \text{ is an edge of }\mb X_{\theta,n} \,] = 
  \begin{cases}
  1/n^{\theta(\{v,w\})} &\text{if }\{v,w\} \in E(G),\\
  0 &\text{otherwise,}
  \end{cases}
\]
independently for all $\{(v,i),(w,j)\} \in \binom{V(G) \times [n]}{2}$. 
\end{df}

\begin{la}[\cite{li2017ac}]
The probability that $\mb X_{\theta,n}$ is a YES-instance of $\SUB(G)$ is bounded away from $0$ and $1$.
\end{la}

The lower bounds of Theorem \ref{thm:LRR} and \ref{thm:main-lb} are in fact average-case lower bounds for $\SUB(G)$ under $\mb X_{\theta,n}$ for arbitrary threshold weightings $\theta$.
Parameters $\kappa(G)$ and $\tau(G)$ are obtained by taking the optimal choice of threshold weighting $\theta$, as we describe in the next subsection.

\subsection{Join-trees and parameters $\kappa(G)$ and $\tau(G)$}\label{sec:join-trees}

Parameters $\kappa(G)$ and $\tau(G)$ are defined in terms of a notion called {\em join-trees} for subgraphs of $G$.  A join-tree is simply a ``formula'' computing a subgraph of $G$, starting from individual edges, where union ($\cup$) is the only operation.

\begin{df}
A {\em join-tree over $G$} is a finite rooted binary tree $A$ together with a labeling $\lab_A : \mr{Leaves}(A) \to E(G) \cup \{\bot\}$ (which may also be viewed as a partial function $\mr{Leaves} \rightharpoonup E(G)$).  We reserve symbols $A,B,C,D,E$ for join-trees.  ($F$ will always denote a subgraph of $G$.)

The {\em graph} of $A$, denoted $\Gr{A}$, is the subgraph of $G$ with edge set $E(G) \cap \mr{Range}(\lab_A)$.  (Note that $\Gr{A}$ is always finite.)
As a matter of notation, we write $E(A)$ for $E(\Gr{A})$ and $V(A)$ for $V(\Gr{A})$.  We also write $\Delta_\theta(A)$ for $\Delta_\theta(\Gr{A})$ where $\theta$ is a threshold weighting on $G$.  

We write $\sq{}$ for the single-node join-tree labeled by $\bot$.  For $e \in E(G)$, we write $\sq{e}$ for the single-node join-tree labeled by $e$.  For join-trees $B$ and $C$, we write $\sq{B,C}$ for the join-tree consisting a root with $B$ and $C$ as children (with the inherited labels, i.e.,\ $\lab_{\sq{B,C}} = \lab_B \cup \lab_C$).  Note that $\Gr{\sq{B,C}} = \Gr{B} \cup \Gr{C}$.

Every join-tree $A$ is clearly either $\sq{}$, or $\sq{e}$ where $e \in E(G)$, or $\sq{B,C}$ where $B,C$ are join-trees.
In the first two cases, we say that $A$ is {\em atomic}; in the third case, we say that $A$ is {\em non-atomic}.

We say that $B$ is a {\em child} of $A$ if $A \in \{\sq{B,C},\sq{C,B}\}$ for some $C$.  We say that $D$ is a {\em sub-join-tree} of $A$ (denoted $D \preceq A$) if $D = \sq{}$ or $D = A$ or $D$ is a sub-join-tree of a child of $A$.  We say that $D$ is a {\em proper sub-join-tree} (denoted $D \prec A$) if $D \preceq A$ and $D \ne A$.
\end{df}

We are now able to define the invariant $\kappa(G)$ in Theorem \ref{thm:LRR}, which lower bounds the restricted circuit size of $\SUB(G)$.  (In fact, $\kappa(G)$ also provides a nearly tight upper bound on the average-case $\ACzero$ circuit size of $\SUB(G)$ \cite{rosenthal2019}.)

\begin{df}[The invariant $\kappa(G)$ of Theorem \ref{thm:LRR}]\label{df:kappa}
For finite graphs $G$, let
\[
  \kappa(G) \defeq 
  \max_{\text{threshold weightings $\theta$ for $G$}}\  
  \min_{\text{join-trees $A$ with graph $G$}}\ \max_{B \preceq A}\ \Delta_\theta(B).
\]
\end{df}

The invariant of $\tau(G)$ of Theorem \ref{thm:main-lb} is significantly more complicated to define.  We postpone the definition to Section \ref{sec:pathset} and, in the meantime, focus on a simpler ``potential function'' on join-trees, denoted $\Phi_\theta(A)$, which we use to lower bound $\tau(G)$.
In order to state the definition of $\Phi_\theta(A)$, we require the following operation $\ominus$ (``restriction away from'') on graphs and join-trees.

\begin{df}[The operation $\ominus$ on graphs and join-trees]
For $F \subseteq G$ and a subset $S \subseteq V(G)$, we denote by $F \ominus S$ the graph consisting of the connected components of $F$ that are vertex-disjoint from $S$.

For a join-tree $A$, we denote $A \ominus S$ the join-tree with the same rooted tree structure as $A$ and leaf labeling function
\[
  \lab_{A \ominus S}(l)
  &\defeq
  \begin{cases}
    \lab_A(l) &\text{if } \lab_A(l) \in E(\Gr{A} \ominus S),\\
    \bot &\text{otherwise.}
  \end{cases}
\]
That is, $A \ominus S$ deletes all labels except for edges in $\Gr{A} \ominus S$. Note that $\Gr{A \ominus S} = \Gr{A} \ominus S$.

As a matter of notation, if $B$ is another join-tree, we write $A \ominus B$ for $A \ominus V(B)$ and $A \ominus (S \cup B)$ for $A \ominus (S \cup V(B))$.  
\end{df}

\begin{figure}[H]
\centering
\includegraphics[scale=.6]{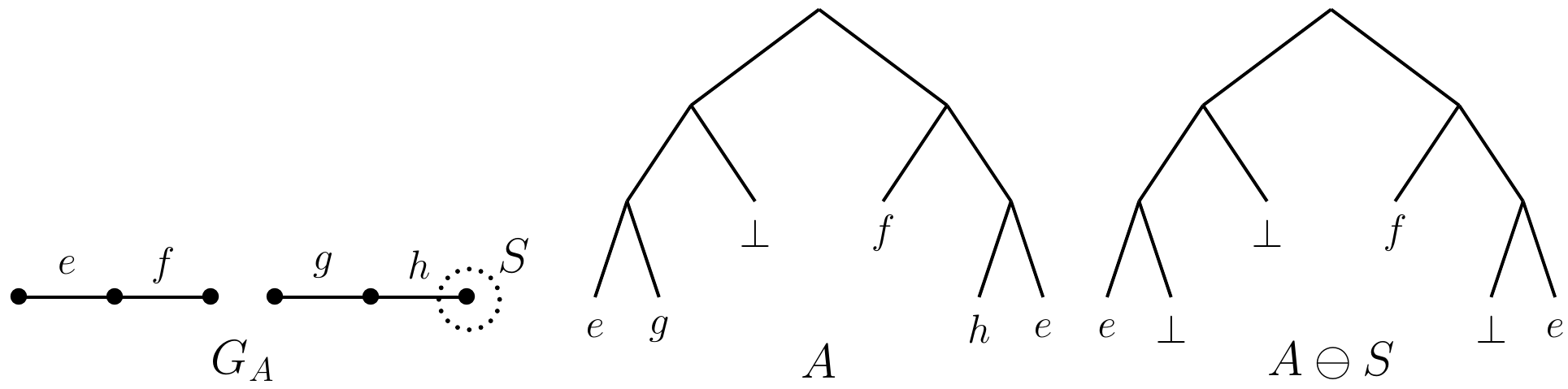}
\caption{\small An example where $A$ is a join-tree whose graph $\Gr{A}$ consists of two paths of length $2$ with edges $e,f,g,h$.  $S$ is the set containing just the external endpoint of $h$.  The join-tree $A \ominus S$ is depicted to the right.}
\end{figure}

\begin{df}[The potential function $\Phi_\theta$ on join-trees]\label{df:Phi}
Fix a threshold weighting $\theta$ on a graph $G$. The potential function $\Phi_\theta : \{$join-trees over $G\} \to \R_{\ge 0}$ is the unique pointwise minimum function satisfying the following inequalities for all join-trees $A,B,C,D$:
\begin{align}
\tag{$\dag$}
  \Phi_\theta(A) &\ge \Phi_\theta(D) + \Delta_\theta(C \ominus D) + \Delta_\theta(A \ominus (C \cup D))
  &&\text{if $A \in \{\sq{B,C},\sq{C,B}\}$ and $D \preceq B$},\\
\tag{$\ddag$}
  \Phi_\theta(A) &\ge \frac12\Big(\Phi_\theta(D) + \Phi_\theta(E \ominus D) + \Delta_\theta(A) + \Delta_\theta(A \ominus (D \cup E))\Big)
  &&\text{if $D,E \prec A$}.
\end{align}
Alternatively, $\Phi_\theta(A)$ has the following recursive characterization:
\begin{itemize}
\item
If $A$ is an atomic join-tree, then
\[
  \Phi_\theta(A) \defeq \Delta_\theta(A) =
  \begin{cases}
  0 &\text{if } A = \sq{},\\
  2 - \theta(e) &\text{if $A = \sq{e}$ where $e \in E(G)$}.
  \end{cases}
\]
(Obs: In the case $A = \sq{e}$, the constraint $\Phi_\theta(A) \ge \Delta_\theta(A)$ is forced by ($\ddag$) where $B = C = \sq{}$.)
\item
If $A = \sq{B,C}$, then
\[
  \Phi_\theta(A) \defeq 
  \max
  \left\{
  \begin{aligned}
  &\max_{\  D \preceq B \ }\ \Phi_\theta(D) + \Delta_\theta(C \ominus D) + \Delta_\theta(A \ominus (C \cup D)),\\
  &\max_{\  D \preceq C\ }\ \Phi_\theta(D) + \Delta_\theta(B \ominus D) + \Delta_\theta(A \ominus (B \cup D)),\\
  &\max_{D,E \prec A}\, 
  \smash{\frac12}\Big(\Phi_\theta(D) + \Phi_\theta(E \ominus D) + \Delta_\theta(A) + \Delta_\theta(A \ominus (D \cup E))\Big)
  \end{aligned}
  \right\}.
\]
\end{itemize}
That is, at least one of inequalities ($\dag$) or ($\ddag$) is tight for each join-tree $A$.
\end{df}

This definition, although opaque at first, will be clarified later (in Sections \ref{sec:Phi} and \ref{sec:pathset}).  The key property of $\Phi_\theta(A)$ is that it provides a lower bound the invariant $\tau(G)$, which in turn provides a lower bound on the restricted formula complexity of $\SUB(G)$.

\begin{thm}[\cite{RossmanICM}]\label{thm:tau}
The invariant $\tau(G)$ of Theorem \ref{thm:main-lb} satisfies
\[
  \tau(G) \ge 
  \max_{\textup{threshold weightings $\theta$ for $G$}}\  
  \min_{\textup{join-trees $A$ with graph $G$}}\ \Phi_\theta(A).
\]
\end{thm}

The definition of $\tau(G)$ and proof of Theorem \ref{thm:tau} are postponed to Section \ref{sec:pathset}.  First, in Section \ref{sec:Tk}, we will present our combinatorial main lemma, which gives a lower bound on $\Phi_\theta(A)$ for all join-trees with graph $T_k$ under the threshold weighting $\theta$ of Example \ref{ex:theta}.

\subsection{Observations about $\Phi_\theta$}

Note that inequality ($\dag$) implies $\Phi_\theta(A) \ge \Phi_\theta(D)$ for all $D \preceq A$ (since $\Delta_\theta(\cdot)$ is nonnegative).  Also note that inequality ($\ddag$) implies $\Phi_\theta(A) \ge \Delta_\theta(A)$ in the special case $B = C = \sq{}$ (since $\Phi_\theta(\sq{}) = 0$ and $A \ominus ($the empty graph$) = A$).
Combining these observations, we see that $\Phi_\theta(A) \ge \Delta_\theta(D)$ for all $D \preceq A$.
It follows that $\tau(G) \ge \kappa(G)$ for all graphs $G$, which makes sense in light of the fact that $\kappa(G)$ bounds circuit size and $\tau(G)$ bounds formula size.

Next, observe that $\Phi_\theta(A)$ always equals either $\Phi_\theta(D) + (\text{some $\Delta_\theta(\cdot)$-terms})$ or $\frac12(\Phi_\theta(D) + \Phi_\theta(E \ominus D)) + (\textit{some $\Delta_\theta(\cdot)$-terms})$ where $D$ and $E$ are proper sub-join-trees of $A$.  This can be expanded out until we get a nonnegative linear combination of $\Delta_\theta(\cdot)$-terms.  Looking closely, we see that
\[
  \Phi_\theta(A) = \sum_{F \subseteq G} c_F \cdot \Delta_\theta(F)
\]
where coefficients $c_F$ (which depend on both $\theta$ and $A$) are nonnegative dyadic rational numbers coming from the tight instances of inequalities ($\dag$) and ($\ddag$).  We may further observe, for any $v \in V(G)$, that
\[
  \sum_{F \subseteq G \,:\, v \in V(F)} c_F \le 1.
\]
This is easily shown by induction using the fact that graphs $F_1$ and $F_2 \ominus F_1$ and $F_3 \ominus (F_1 \cup F_2)$ are pairwise disjoint for any $F_1,F_2,F_3 \subseteq G$.

One consequence of this observation is the following lemma, which we will use in Sections \ref{sec:Tk} and \ref{sec:Pk}.

\begin{la}\label{la:infinite}
Suppose $(G,\theta)$ and $(G^\ast,\theta^\ast)$ are threshold-weighted graphs such that $G \subseteq G^\ast$ and $\theta^\ast(e) \le \theta(e)$ for all $e \in E(G)$. Then for any join-tree $A$ with graph $G$, we have 
\[
  \Phi_{\theta}(A) \ge \Phi_{\theta^\ast}(A) - \sum_{e \in E(G)} \Big(\theta(e) - \theta^\ast(e)\Big).
\]
\end{la}

\begin{proof}
Let $\{c_F\}_{F \subseteq G}$ be nonnegative dyadic rationals --- arising from the tight instances of inequalities ($\dag$) and ($\ddag$) in the recursive definition of $\Phi_{\theta^\ast}(A)$   --- such that $\Phi_{\theta^\ast}(A) = \sum_{F \subseteq G} c_F \cdot \Delta_{\theta^\ast}(F)$.  We may apply inequalities ($\dag$) and ($\ddag$) in the exact same way to get the bound $\Phi_\theta(A) \ge \sum_{F \subseteq G} c_F \cdot \Delta_\theta(F)$. We now have
\[
  \Phi_{\theta^\ast}(A) - \Phi_\theta(A) 
  &\le 
  \sum_{F \subseteq G} c_F  \Big(\Delta_{\theta^\ast}(F) - \Delta_\theta(F)\Big)\\
  &= 
  \sum_{F \subseteq G} c_F  
  \sum_{e \in E(F)} \Big(\theta(e) - \theta^\ast(e)\Big)\\  
  &=
  \sum_{e \in E(G)} \Big(\theta(e) - \theta^\ast(e)\Big)
  \sum_{F \subseteq G \,:\, e \in E(F)} c_F\\
  &\le
  \sum_{e \in E(G)} \Big(\theta(e) - \theta^\ast(e)\Big),
\]
using the fact that $\theta(e) - \theta^\ast(e) \ge 0$ and $\sum_{F \subseteq G \,:\, v \in V(F)} c_F \le 1$ for all $v \in V(G)$.
\end{proof}

\subsection{Lower bounds on $\Phi_\theta$}

Having introduced the potential function $\Phi_\theta$ and described its connection to $\tau$ in Theorem \ref{thm:tau}, we conclude this section by briefly explaining how it is used derive lower bounds $\tau(P_k)$ and $\tau(T_k)$. The main combinatorial lemma behind the lower bound of Theorem \ref{thm:previous} is the following:

\begin{la}[\cite{rossman2018formulas}]\label{la:PhiPk}
Let $\theta$ be the constant $1 + \frac{1}{k}$ threshold weighting on $P_k$.  For every join-tree $A$ with graph $P_k$, we have $\Phi_\theta(A) \ge \frac{1}{2}\log_{\sqrt{13}+1}(k)$.  (Therefore, $\tau(G) \ge \frac{1}{2}\log_{\sqrt{13}+1}(k)$.)
\end{la}

The proof is included in Appendix \ref{sec:appendix}, for the sake of comparison with our two lower bounds below.  We remark that this proof makes crucial use of both ($\dag$) and ($\ddag$); it was shown in \cite{rossman2018formulas} that no lower bound better than $\Phi_\theta(A) = \Omega(1)$ is provable using ($\dag$) alone or ($\ddag$) alone.

Our lower bound $\tau(T_k) = \Omega(k)$ (Theorem \ref{thm:Tk}) is an immediate consequence of the following:

\begin{la}\label{la:PhiTk}
Let $\theta$ be the threshold weighting arising from the uniform random walk on $T_k$ (Example \ref{ex:theta}).  For every join-tree $A$ with graph $T_k$, we have $\Phi_\theta(A) \ge k/30 - 1/5$.
\end{la}

Our proof, given in the next section, is purely graph-theoretic. Interestingly, the argument essentially uses only inequality ($\ddag$); we do not require ($\dag$), other than in the weak form $\Phi_\theta(A) \ge \Phi_\theta(D)$ for all $D \prec A$.  

It is worth mentioning that the choice of threshold weighting is important in Lemma \ref{la:PhiTk}.  A different, perhaps more obvious, threshold weighting is the constant function with value $\frac{|V(T_k)|}{|E(T_k)|}$ ($= \frac{2^{k+1}-1}{2^{k+1}-2}$).  With respect to this threshold weighting, no lower bound better than $\Omega(1)$ is possible.

Finally, our improved lower bound $\tau(P_k) \ge \log_{\sqrt{5}+5}(k)$ (Theorem \ref{thm:Pklower}) is obtained via the following lemma.  This result involves a 2-parameter extension of $\Phi_\theta(A)$ denoted $\Phi_\theta(A|S)$ (where $S \subseteq V(G)$), which we introduce in Section \ref{sec:Phi}.

\begin{la}\label{la:PhiPk2}
Let $\theta$ be the constant $1 + \frac{1}{k}$ threshold weighting on $P_k$.  For every join-tree $A$ with graph $P_k$, we have $\Phi_\theta(A|\emptyset) \ge \log_{\sqrt{5}+5}(k) - 1$.
\end{la}

This lemma is proved in Section \ref{sec:Pk}, after we show how $\Phi_\theta(\cdot|\cdot)$ provides a lower bound on $\tau(\cdot)$ in Section~\ref{sec:pathset}.

\section{Lower bound $\tau(T_k) = \Omega(k)$}\label{sec:Tk}

We fix the infinite pattern graph $T_\infty$ with the threshold weighting $\theta$ induced by the uniform random walk.
Recall that $T_\infty = \bigcup_{k=1}^\infty T_k$ under a nesting $T_1 \subset T_2 \subset T_3 \subset \cdots$ with $\mr{Leaves}(T_1) \subset \mr{Leaves}(T_2) \subset \mr{Leaves}(T_3) \subset \cdots$.   
$F,G,H$ will represent finite subgraphs of $T_\infty$, and $A,B,C,D,E$ will be join-trees over $T_\infty$. (In particular, note that $G$ no longer denotes the ambient pattern graph.)  

We next recall the definition of $\theta$ from Example \ref{ex:theta}.  Let $M \in [0,1]^{V(T_\infty) \times V(T_\infty)}$ be the transition matrix of the uniform random walk on $T_\infty$, that is,
\[
  M_{v,w} = \begin{cases}
    1 &\text{if $\{v,w\} \in E(T_\infty)$ and $v$ is a leaf,}\\
    1/3 &\text{if $\{v,w\} \in E(T_\infty)$ and $v$ is a non-leaf,}\\
    0 &\text{if $\{v,w\} \notin E(T_\infty)$.}
  \end{cases}
\]
This induces the threshold weighting $\theta : E(T_\infty) \to [0,2]$ given by
\[
  \theta(\{v,w\}) 
  \defeq 
  M_{v,w} + M_{w,v}
  =
  \begin{cases}
    2/3 &\text{if $v$ or $w$ is a leaf of $T_\infty$},\\
    4/3 &\text{otherwise.}
  \end{cases}
\]
Since $\theta$ is fixed, we will suppress it when writing $\Delta(F)$ and $\Phi(A)$.

\begin{df}
For all $k \ge 0$, let 
\[
  V_k \defeq \{v \in V(T_k) : v\text{ has distance $k$ from a leaf}\}.
\]
Thus, $V_0$ is the set of leaves in $T_\infty$, $V_1$ is the set of parents of leaves, etc. Note that $V(T_\infty) = \bigcup_{k=0}^\infty V_k$. We shall refer to the $V_k$ as the various \emph{levels} of $T_\infty$.

For $k \ge 1$ and $x \in V_k$, let $T_x \subset T_\infty$ be the complete binary tree of height $k$ rooted at $x$ (in the case $k=0$, we regard $T_x$ as a single isolated vertex).
We denote by $T_x^+$ the graph obtained from $T_x$ by including an extra edge between $x$ and its parent.  Note that
\[
  |V(T_x)| = 2^{k+1}-1,\qquad |V(T_x^+)| = 2^{k+1},\qquad
  |E(T_x)| = 2^{k+1}-2,\qquad |V(T_x^+)| = 2^{k+1}-1.
\]
For $j \in \{0,\dots,k\}$, let $V_j(T_x) \defeq V_j \cap V(T_x)$.
\end{df}

\begin{obs}\normalfont\label{obs:levelsize}
If $x\in V_k$, then for $j \in \{0,\dots,k\}$, $|V_j(T_x)| = 2^{k-j}$.
\end{obs}

We next define two useful parameters of finite subgraphs of $T_\infty$.

\begin{df}[Max-complete height]
For a finite subgraph $F$ of $T_\infty$, define the {\em max-complete height} $\lambda(F)$ to be the maximum $k \in \N$ for which there exists $x \in V_k$ with $T_x \subseteq G$;
$\lambda(F)$ is defined to be zero when no such $x$ exists (in particular, this happens when $V(F)\cap V_0 = \emptyset$). 
\end{df}

\begin{obs}\normalfont
For any $x\in V_k$, $\lambda(T_x) = \lambda(T_x^+) = k$.
\end{obs}

\begin{df}[Boundary size]
Let $\partial(F)$ denote the size of {\em boundary} of $F$ in $T_\infty$:
\[
  \partial(F) \defeq |V(F) \cap V(T_\infty \setminus F)|.
\]
\end{df}

\begin{obs}\normalfont
For any $x\in V_k$, we have $\partial(T_x) = \partial(T_x^+) = 1$, as the boundaries in the respective graphs are simply the singletons $\{x\}$ and $\{\mr{parent}(x)\}$. Another example is as follows: if $x\in V_k$ for some $k\geq 2$ and $F$ is the subgraph of $T_x$ induced by the set of vertices $V(T_x)\setminus V_0$, then $\partial(F) = 2^{k-1} + 1$ as all vertices in $V_1(T_x)$ (along with $x$) lie in the boundary of $F$.
\end{obs}

\begin{df}[Grounded and ungrounded subgraphs of $T_\infty$]
Let $F,H$ be finite subgraphs of $T_\infty$. We say that $F$ is {\em grounded} if it is connected and $V(F) \cap V_0 \ne \emptyset$ (that is, $F$ is a tree, at least one of whose leaves is also a leaf of $T_\infty$).  
We say that $H$ is {\em ungrounded} if it is non-empty and connected and $V(H) \cap V_0 = \emptyset$ (that is, $H$ is a non-empty tree, none of whose leaves is a leaf of $T_\infty$).
\end{df}

\begin{figure}[H]\label{fig:introductory}
\centering
\includegraphics[scale=.55]{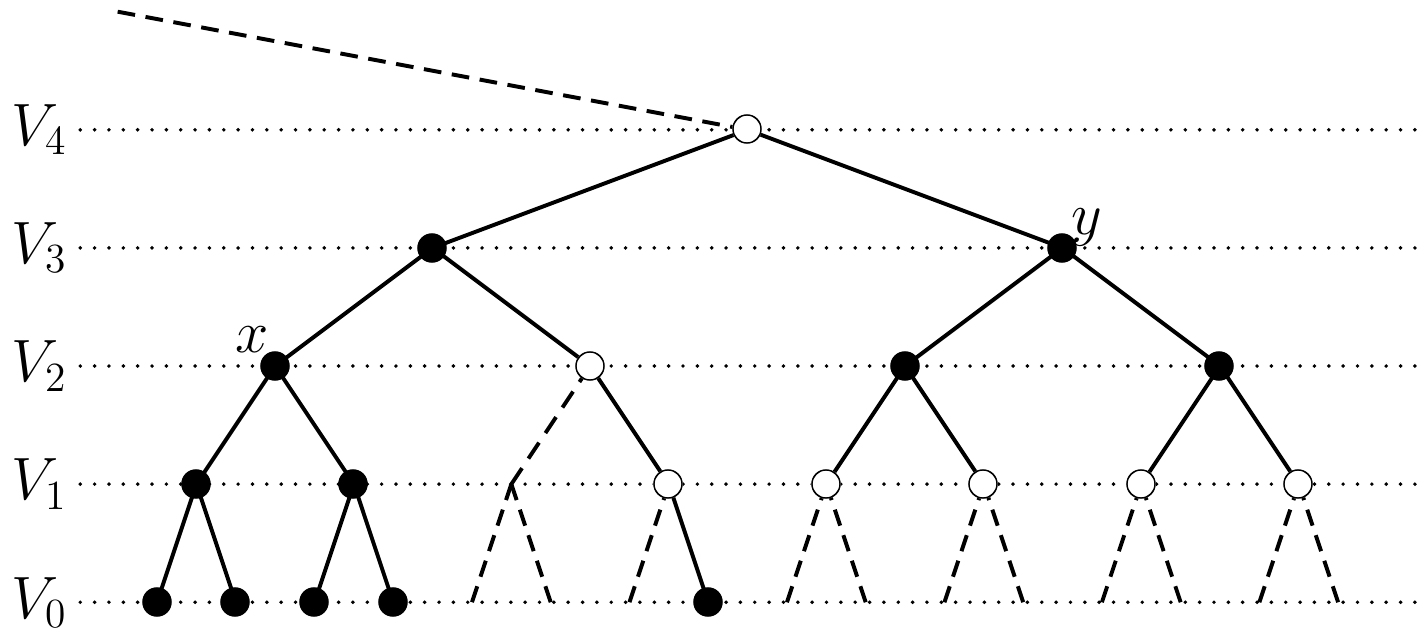}
\caption{\small Example of a grounded graph $F \subset T_\infty$. The dotted lines indicate the various levels that $V(F)$ intersects and dashed lines indicate (some of the) edges in $T_\infty \setminus F$. The nodes colored white lie in the boundary of $F$, and therefore $\partial(F) = 7$ and $\Delta(F) = 11/3$. The max-complete height $\lambda(F) = 2$, since $T_x \subset F$ for $x \in V_2$ and $x$ is the highest such node.  The subgraph $H = F \cap T_y$ is ungrounded with $\lambda(H) = 0$. 
}
\end{figure}

We shall think of the function $\partial(F)$ as essentially a proxy for $\Delta(F)$, as it has the advantage of having a simple combinatorial definition.  This is justified by the following:

\begin{la}\label{la:partial}
For every $F \subset T_\infty$, we have $\Delta(F) \ge \partial(F)/3$. 
\end{la}

It also holds that $\Delta(F) \le 2\partial(F)/3$ for $F$ without isolated vertices (or $\Delta(F) \le \partial(F)$ if we allow isolated vertices), but we will not need this upper bound.

\begin{proof}
Since $\theta$ is the threshold weighting induced by $M$, Lemma \ref{la:markov} tells us
\[
  \Delta(F) 
  =
  \sum_{v \in V(F)}\ \sum_{w \in V(T_\infty) \,:\, \{v,w\} \in E(T_\infty \setminus F)} M_{v,w}.
\]
Note that $v \in V(F)$ contributes to this sum if, and only if, it belongs to the boundary of $F$ (i.e.,\ $v \in V(F) \cap V(T_\infty \setminus F)$). Since $M_{v,w} \ge 1/3$ whenever $\{v,w\} \in E(F)$, the claim follows.
\end{proof}

We are now ready to state the main theorem of the section. 

\begin{thm}\label{thm:philb}
Let $\eps = 1/30$ and $\delta = 2/5$.
Then for every join-tree $A$,
\[
  \Phi(A) \ge \eps \lambda(A) + \delta \Delta(A).
\]
\end{thm}

Theorem \ref{thm:philb} directly implies Lemma \ref{la:PhiTk}, which in turn yields the lower bound $\tau(T_k) = \Omega(k)$ of Theorem \ref{thm:Tk}.  
To see why, let $\theta'$ be the threshold weighting on $T_k$ coming from the uniform random walk (Example \ref{ex:theta}).  
Viewing $T_k$ a subgraph of $T_\infty$, note that  $\sum_{e \in E(T_k)} (\theta'(e) - \theta(e)) = 2(\frac56 - \frac23) = \frac13$.  
For any join-tree $A$ with graph $T_k$, 
Lemma \ref{la:infinite} and Theorem \ref{thm:philb} now imply
\[
  \Phi_{\theta'}(A) 
  \ge 
  \Phi(A) - \frac13 
  \ge 
  \frac{1}{30}\lambda(T_k) + \frac{2}{5}\Delta(T_k) - \frac13  
  =
  \frac{1}{30}k + \frac{2}{5}\cdot\frac13 - \frac13
  =
  \frac{1}{30}k - \frac{1}{5}. 
\]

We proceed with a few definitions and lemmas needed for the proof of Theorem \ref{thm:philb} in Section \ref{sec:philb}.  In order to present the main arguments first, proofs of three auxiliary lemmas (\ref{lem:boundarylarge}, \ref{la0}, \ref{la1}) are postponed to Section \ref{sec:lemmas}.

\begin{la}\label{lem:boundarylarge}
Let $H$ be a non-empty finite subgraph of $T_\infty$, all of whose components are ungrounded. Then $\partial(H) \geq \frac12 (|E(H)|+3)$.
\end{la}

We make a note of its following corollary here.

\begin{cor}\label{la:no-branch}
Suppose $H$ be a finite connected subgraph of $T_\infty$ and $y \in V(H)$ such that $E(H) \cap E(T_y) \neq \emptyset$ and $H$ does \underline{not} contain any path from $y$ to a leaf of $T_y$. Then $\partial(H) \ge \frac12(|E(H)\cap E(T_y)|+1)$.
\end{cor}

\begin{proof}
Let $F$ be the graph with edge set $E(F) \coloneqq E(H)\cap E(T_y)$. Note that $F$ is non-empty, connected and ungrounded.
Observe that $\partial(H) \geq \partial(F) - 1$ because all vertices in the boundary of $H$, with the only possible exception of $y$, also lie in the boundary of $F$. Hence by Lemma~\ref{lem:boundarylarge}, $\partial(H) \geq \frac12(|E(H)\cap E(T_y)|+1)$.
\end{proof}

The second auxiliary lemma gives a useful inequality relating $\partial(G)$, $\lambda(G)$ and $|E(G)|$.

\begin{la}\label{la0}
For every finite subgraph $G$ of $T_\infty$, we have
$\lambda(G) + \partial(G) \ge \log(|E(G)|+1)$.
\end{la}

(This is tight when $G = T_x^+$ for some $x \in V_k$, in which case $\lambda(G) = k$ and $\partial(G) = 1$ and $|E(G)| = 2^{k+1}-1$.)
The third auxiliary lemma shows that subgraphs of $T_k$ that contain at most half the edges of $T_k$ and have boundary size $j$ ($\le k/2$) have empty intersection with a large complete subtree of $T_k$ of height $k - j$.

\begin{la}\label{la1}
Let $x \in V_k$ and suppose $G \subseteq T_x$ such that $|E(G \cap T_x)| \le 2^k-1$ and $\partial(G) \le k/6$. 
Then there exists a vertex $z \in V_{k - \partial(G)}(T_x)$ such that $E(G) \cap E(T^+_z) = \emptyset$.
\end{la}

We now state and prove the main lemma used in the proof of Theorem \ref{thm:philb}.

\begin{la}\label{la2}
For any integers $1 \le t \le \ell$, let $z \in V_\ell$ and suppose $A$ is a join-tree such that $T_z \subseteq \Gr{A}$.  Then one the following conditions holds:
\begin{enumerate}[\normalfont\quad (i)]
\item
There exists $D \preceq A$ such that $\partial(D) \ge t$ and $\lambda(D)+\partial(D) \ge \ell$.
\item
There exists $C \prec A$ with $\lambda((C\cap T_z) \ominus \{z\}) + \partial((C\cap T_z) \ominus \{z\}) \ge \ell - t$.
\item
There exists $E \prec A$ such that $\partial(E) \ge \ell - t$.
\end{enumerate}
\end{la}

\begin{proof}
Descend in the join-tree $A$ until reaching a $B \preceq A$ such that $B$ contains a path $P$ from $z$ to a leaf of $T_z$, but no $B' \prec B$ contains a path from $z$ to a leaf of $T_z$.

Let $j \in \{1,\dots,\ell\}$ be maximal such that there exists $y \in V_j(T_z)\cap P$ such that $T_y \subseteq \Gr B$.  We claim that $\partial(B) \ge \ell-j$. To see this, note that for every vertex $v \neq y$ on the path from $z$ to $y$ (a subpath of $P$), if $c(v)$ denotes the child of $v$ that is {not} on the path $P$, then $G$ does {not} contain $T^+_{c(v)}$ (because if it did then $\lambda(G) > k$). As a result, it must be the case that for every vertex $v\neq y$ on the path from $z$ to $y$, some vertex in $V(G_B) \cap V(T^+_{c(v)})$ lies in the boundary of $G$ and so, $\partial(G) \geq \ell - j$. 

As an illustration of this argument, consider the graph $G_B$ in Figure \ref{fig:j-small}. Then for every vertex $v\in \{v_2,\ldots,v_5\}$, either $v$ itself is in the boundary (like $v_2$ and $v_5$) or some vertex in $T_{c(v)}$ is in the boundary (for $v_3$, it is $c(v_3)$ for example and for $v_4$, it is either child of $c(v_4)$).

\begin{figure}
\centering
\begin{subfigure}{0.45\textwidth}
    \centering
    \includegraphics[scale=0.5]{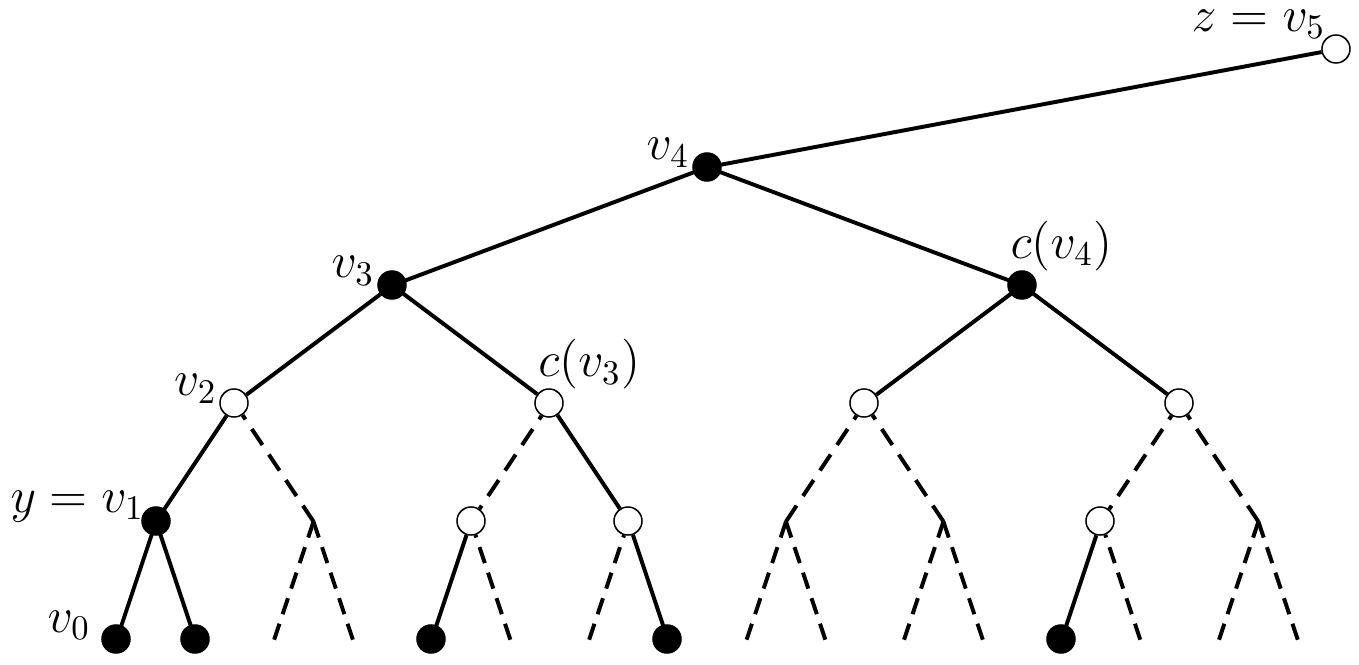} 
    \caption{\normalfont\footnotesize Here $j = 1$ is ``small''. For each $i \in \{2,\dots,5\}$, either $v_i$ or a node in its right subtree lies in the boundary of $B$.} 
    \vspace{2ex}
    \label{fig:j-small}
  \end{subfigure}\quad
  \begin{subfigure}{0.45\textwidth}
    \centering
    \includegraphics[scale=0.5]{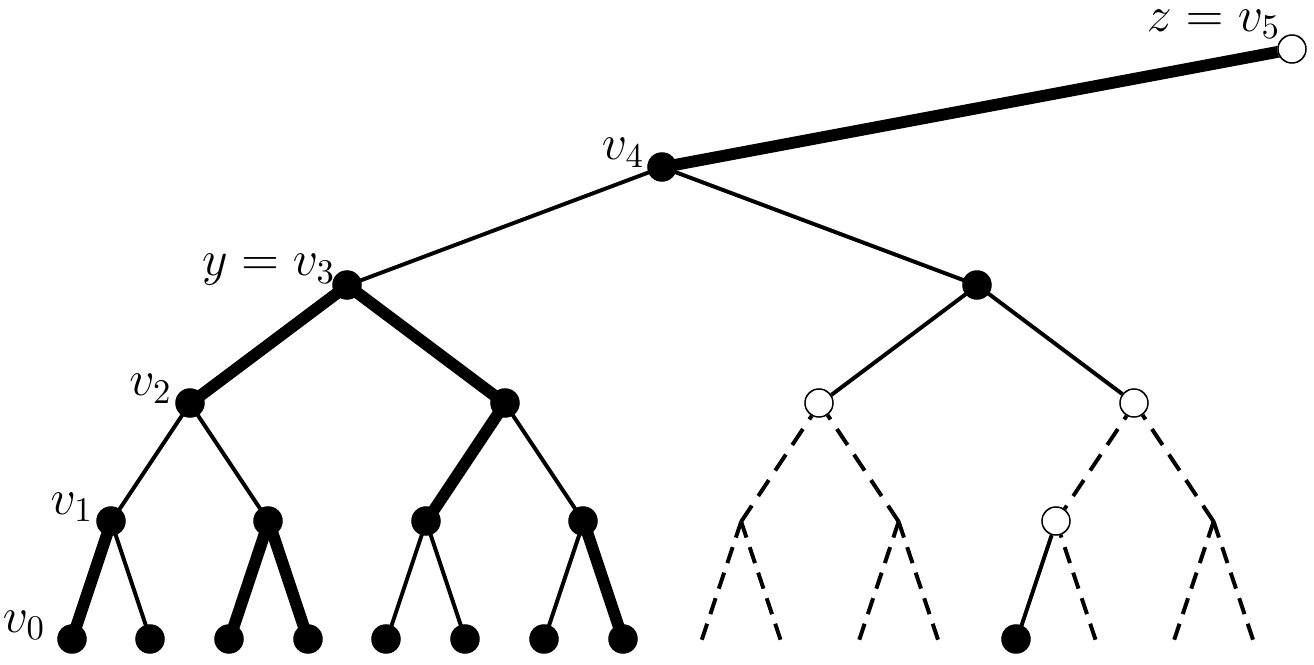} 
    \caption{\normalfont\footnotesize $C$ doesn't contain the path between $y$ and $z$. So we can directly use Lemma~\ref{la0} on $(C\cap T_z)\ominus \{z\}$.} 
    \vspace{2ex}
    \label{fig:zy_not_connected}
  \end{subfigure}\\
  \begin{subfigure}{0.45\linewidth}
    \centering
    \includegraphics[scale=0.5]{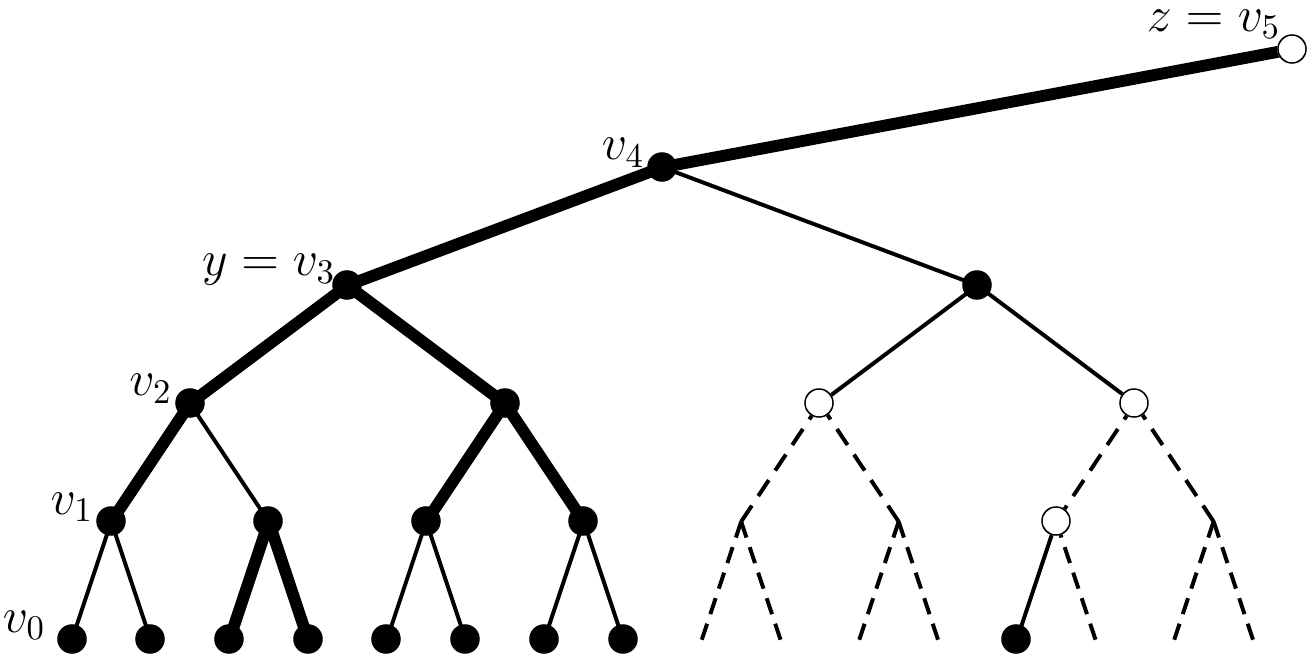} 
    \caption{\normalfont\footnotesize $C$ contains the path from $z$ to $y$ and $H\cap T_y$ is ``large''. But also $H$ is ungrounded and thus has large boundary.} 
    \vspace{1ex}
    \label{fig:C_y-large}
  \end{subfigure}\quad 
  \begin{subfigure}{0.45\linewidth}
    \centering
    \includegraphics[scale=0.5]{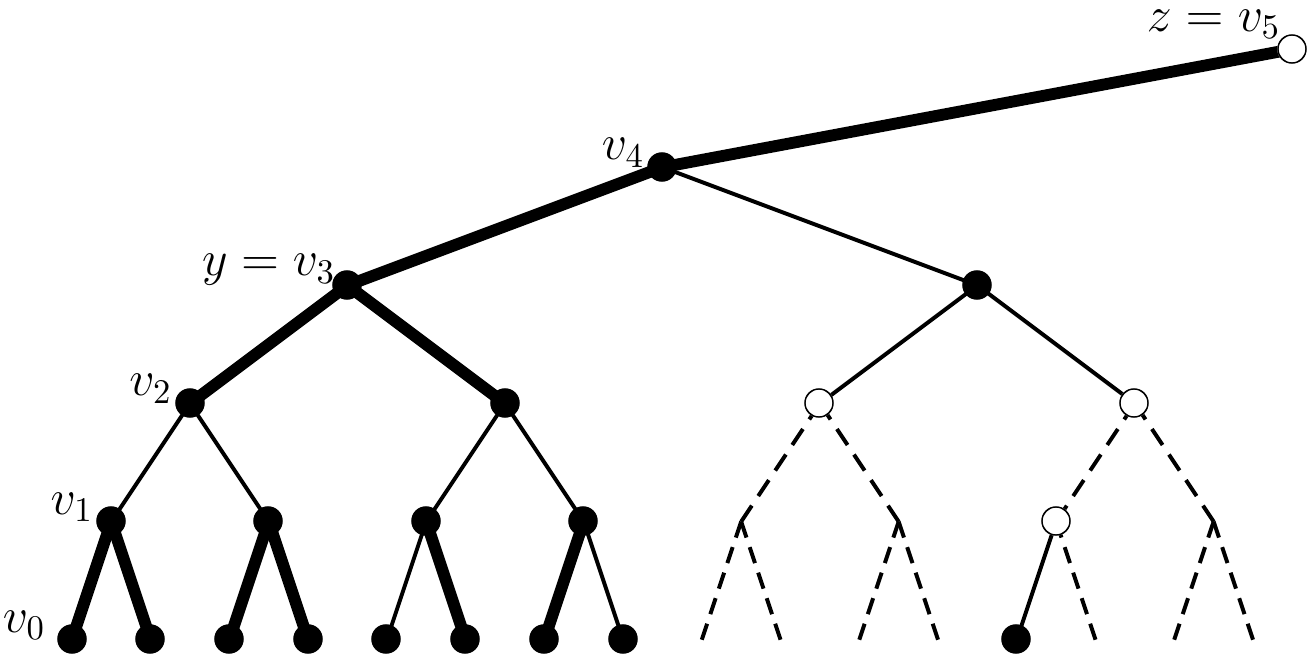} 
    \caption{\normalfont\footnotesize $H\cap T_y$ is ``small''. But then again, $(C\cap T_z)\ominus \{z\}$ is large and so Lemma~\ref{la0} is applicable.} 
    \vspace{1ex}
    \label{fig:C_y-small}
  \end{subfigure}
  \caption{\small An illustration with $\ell = 5$, $t = 3$.  In each example, solid lines denote the edges of $B$, thickly shaded lines denote the subgraph $C$, and dashed lines denote edges of $T_z \setminus B$. 
  Nodes shaded white lie in the boundary of $B$ (thought not necessarily the boundary of $C$).  $P$ is the path $v_0v_1v_2v_3v_4v_5$.}
\end{figure}

Consider the case that $j \le \ell-t$ (again see Figure~\ref{fig:j-small}).  Letting $D \defeq B$, we have $\partial(D) \ge \ell-j \ge t$ and $\lambda(D) \ge j$, so condition (i) is satisfied. We shall therefore proceed under the assumption that 
\[
  j \ge \ell - t + 1.
\]

Since $T_y \subseteq \Gr{B}$, at least one child $C$ of $B$ satisfies 
\[
  |E(C) \cap E(T_y)| \ge \frac12|E(T_y)| = 2^j-1.
\]
Fix one such $C$.

Consider the case that $C$ does \underline{not} contain the path between $z$ and $y$ (see Figure~\ref{fig:zy_not_connected}). Then $C\cap T_y \subseteq C \ominus \{z\}$, so by Lemma \ref{la0},
\[
  \lambda((C\cap T_z) \ominus \{z\}) 
  + \partial((C\cap T_z) \ominus \{z\}) 
  &\ge 
  \log(|E(C \cap T_y)|+1) 
  \ge
  \log(2^j)
  \ge
  \ell-t + 1.
\]
In this case, we satisfy condition (ii).  We shall therefore proceed under the assumption that $C$ contains the path between $z$ and $y$. 

Note that $C$ does \underline{not} contain a path from $y$ to any leaf of $T_y$ (since otherwise $C$ would contain a path from $z$ to a leaf of $T_z$, contradicting the way we choose $B \preceq A$).
Let $H$ be the connected component of $\Gr C$ that contains $y$ (and hence also contains the path between $z$ and $y$).  

We now consider two final cases, depending on the size of $|E(H) \cap E(T_y)|$.
First, assume $|E(H) \cap E(T_y)| \ge 2(\ell-t)$ (see Figure~\ref{fig:C_y-large}). 
In this case, Corollary~\ref{la:no-branch} implies
\[
  \partial(C) \ge \partial(H) \ge \frac12\Big(|E(H)\cap E(T_y)|+1\Big) \ge \ell-t.
\]
We satisfy condition (iii) setting $E \defeq C$.

Finally, assume $|E(H) \cap E(T_y)| \le 2(\ell-t) - 1$ (see Figure~\ref{fig:C_y-small}). We have
\[
  |E((C\cap T_z) \ominus \{z\})| 
  &\ge |E(C) \cap E(T_y)| - |E(H) \cap E(T_y)|\\
  &\ge (2^j - 1) - (2(\ell-t) - 1)\\
  &\ge 2^{\ell-t+1} - 2(\ell-t).
\]
Lemma \ref{la0} now implies
\[
  \lambda((C\cap T_z) \ominus \{z\}) 
  + \partial((C\cap T_z) \ominus \{z\}) 
  &\ge 
  \log(|E((C\cap T_z) \ominus \{z\})| + 1) \\
  &\ge
  \log(2^{\ell-t+1} - 2(\ell-t) + 1) \\
  &> 
  \ell-t 
\]
since $2^{x+1} - 2x + 1 > 2^x$ for all $x \ge 0$. 
Therefore, condition (ii) is again satisfied in this final case. 
\end{proof}

\subsection{Proof of Theorem \ref{thm:philb}}\label{sec:philb}

We now prove Theorem \ref{thm:philb}: the lower bound $\Phi(A) \ge \eps\lambda(A) + \delta\Delta(A)$ where $\eps = 1/30$ and $\delta = 2/5$.

We argue by a structural induction on join-trees $A$.  First, consider the case that $\Gr{A}$ is empty, then $\Phi(A) = 0$ and $\eps\lambda(A) + \delta\Delta(A) = 0$. We shall assume that $\Gr{A}$ is non-empty.

Next consider the base case where $A$ is the atomic join-tree $\sq{e}$ for an edge $e \in E(T_\infty)$.  In this case, we have $\lambda(A) \le 1$ and 
\[
  \Phi(A) = \Delta(A) = 
  \begin{cases}
    2/3 &\text{if $e$ contains a leaf,}\\
    4/3 &\text{otherwise.}
  \end{cases}
\]
Therefore, $\eps\lambda(A) + \delta\Delta(A) \le (1/30) + (2/5)(4/3) = 17/30$.
We are done, since $\Phi(A) \ge 2/3 > \eps\lambda(A) + \delta\Delta(A)$.

From now on, let $A$ be a non-atomic join-tree with whose graph is non-empty.  Let
\[
  k \defeq \lambda(A).
\]
Our goal is thus to prove $\Phi(A) \ge \eps k + \delta\Delta(A)$, which we do by analyzing numerous cases.

Consider first the case that $k = 0$. In this case, we clearly have $\Phi(A) \ge \eps k + \delta \Delta(A)$ (since $\Phi \ge \Delta \ge 0$ and $\delta < 1$).
So shall proceed on the assumption that $k \ge 1$. 

Since $\Phi \ge \Delta$, we are done if $\Delta(A) \ge \eps k + \delta\Delta(A)$. So we shall proceed on the additional assumption that
\begin{equation}\label{eq:DeltaA}
  \Delta(A) \le \frac{\eps}{1-\delta} k
  = \frac{1}{18} k.
\end{equation}

\paragraph{Fixing $x \in T_k$ with $T_k \subseteq \Gr{A}$:}
By definition of $\lambda(A)$, there exists a vertex $x \in T_k$ such that $T_k \subseteq \Gr{A}$.  Let us fix any such $x$.

\paragraph{Fixing $B \preceq A$ with $2^{k-1} \le |E(B) \cap E(T_x)| \le 2^k-1$:}  
We next fix a sub-join-tree $B \preceq A$ satisfying $2^{k-1} \le |E(B) \cap E(T_x)| \le 2^k-1$. 
To see that such $B$ exists, first note that $|E(A) \cap E(T_x)| = |E(T_x)| = 2^{k+1}-2$. Consider a walk down $A$ which at each step descends to a child $C$ which maximizes $|E(C) \cap E(T_x)|$.  This quantity shrinks by a factor $\ge 1/2$ at each step, eventually reaching size $1$.  Therefore, at some stage, we reach a sub-join-tree $B$ such that the intersection size $|E(B) \cap E(T_x)|$ is between $2^{k-1}$ and $2^k-1$.\bigskip

Observe that $\Phi(A) \ge \Phi(B) \ge \Delta(B)$ (by ($\dag$) and the fact that $\Phi \ge \Delta$ for all join-trees).  Therefore, we are done if $\Delta(B) \ge \eps k + \delta\Delta(A)$. So we shall proceed under the additional assumption that
\begin{equation}\label{eq:DeltaB}
  \Delta(B) \le \eps k + \delta\Delta(A)
  =
  \frac{1}{30}k + \frac{2}{5}\Delta(A)
  \le
  \Big(\frac{1}{30} + \frac{2}{5}\cdot\frac{1}{18}\Big)k
  =
  \frac{1}{18}k.
\end{equation}

Since $|E(B)| \ge |E(B) \cap E(T_x)| \ge 2^{k-1}$, Lemma \ref{la0} tells us that $\lambda(B) + \partial(B) \ge k-1$.  We make note of the fact that this implies 
\begin{align}\label{eq:PhiB}
  \Phi(B) 
  &\ge
  \eps\lambda(B) + \delta\Delta(B) &&\text{(induction hypothesis)}\\
\notag  &\ge
  \eps \Big(k - \partial(B) - 1\Big) + \delta \Delta(B)\\
\notag  &\ge
  \eps k + (\delta - 3\eps) \Delta(B) - \eps &&\text{(using $-\partial \ge -3\Delta$)}.
\end{align}

\paragraph{Fixing $z \in V_{k-\partial(B)}(T_x)$ with $E(B) \cap E(T_z^+) = \emptyset$:}
Note that $\partial(B) \le 3\Delta(B) \le k/6$ by (\ref{eq:DeltaB}).  Since 
$|E(B) \cap E(T_x)| \le 2^k - 1$, the hypotheses of Lemma \ref{la1} are satisfied with respect to the vertex $x$ and the graph $\Gr{B} \cap T_x$.
Therefore, we may fix a vertex $z \in V_{k-\partial(B)}(T_x)$ such that $E(B) \cap E(T_z^+) = \emptyset$.\bigskip

We next introduce a parameter
\[
  t \defeq 6\Delta(A) + 3\Delta(B).
\] 
Note that $t$ is an integer, since $3\Delta$ is integral. 
Our choice of parameters moreover ensure that $1 \le t \le k/2$, since $\Delta(A),\Delta(B) \le k/18$ by (\ref{eq:DeltaA}),(\ref{eq:DeltaB}) and $\Delta(A),\Delta(B) \ge 1/3$ for all nonempty graphs.  
Since $k - \partial(B) \ge k - 3\Delta(B) \ge 5k/6$, it follows that $t < k - \partial(B)$.
\bigskip

Since $z \in V_{k-\partial(B)}$ and $T_z \subseteq T_x \subseteq \Gr{A}$, Lemma \ref{la2} (with respect to $z$ and $1 \le t < \ell \defeq k - \partial(B)$) tells us that one of the following conditions holds:
\begin{enumerate}[\normalfont\quad (i)]
\item
There exists $D \preceq A$ such that $\partial(D) \ge t$ and $\lambda(D) + \partial(D) \ge k - \partial(B)$. 
\item
There exists $C \prec A$ with $\lambda((C\cap T_z) \ominus \{z\}) + \partial((C\cap T_z) \ominus \{z\}) \ge k-\partial(B)-t$.
\item
There exists $E \prec A$ with $\partial(E) \ge k-\partial(B)-t$.
\end{enumerate}
We will show that $\Phi(A) \ge \eps k + \delta \Delta(A)$ in each of these three cases.

\paragraph{Case (i):}
Suppose $D = A$. Then it follows that
$
\partial (A) \geq t > 6 \Delta(A)
$,
but this contradicts Lemma~\ref{la:partial} by which $\Delta(A) \geq \partial(A)/3$, as boundary of a non-empty graph is always non-empty. Thus, $D\prec A$ and we have
\[
  \Phi(A) 
  \stackrel{(\dag)}{\ge} 
  \Phi(D) 
  \vphantom{\big|}&\ge 
  \eps \lambda(D) + \delta \Delta(D)
  &&\text{(induction hypothesis)}\\
  &\ge 
  \eps \lambda(D) +  \frac{\delta}{3} \partial(D)
  &&\text{(since $\Delta \ge \partial/3$)}\\
  &\ge
  \eps \Big(k-\partial(B)\Big) + \Big(\frac{\delta}{3} - \eps\Big) \partial(D)
  &&\text{(since $\lambda(D) + \partial(D) \ge k - \partial(B)$ by Case (i))}\\
  &\ge
  \eps \Big(k-\partial(B)\Big) + \Big(\frac{\delta}{3} - \eps\Big) t
  &&\text{(since $\delta/3 > \eps$ and $\partial(D) \ge t$ by Case (i))}\\
  &=
  \eps \Big(k-\partial(B)\Big) + \Big(\frac{\delta}{3} - \eps\Big)\Big(6\Delta(A) + 3\Delta(B)\Big)\\
  &\ge
  \eps \Big(k - 6\Delta(A) - 6\Delta(B)\Big) + \delta\Big(2\Delta(A) + \Delta(B) \Big)
  &&\text{(since $-\partial \ge -3\Delta$)}\\
  &= \eps k + \delta \Delta(A) +   \Big( \Delta(A) + \Delta(B) \Big) (\delta - 6\eps)\\
  &\ge
  \eps k + \delta \Delta(A)
  &&\text{(since $\delta > 6\eps$)}.
\]

\paragraph{Case (ii):} 
Recall that $z \in V_{k-\partial(B)}(T_x)$ was chosen such that $E(B) \cap E(T_z^+) = \emptyset$. 
It follows that the graph of $(C\cap T_z) \ominus \{z\}$ is a union of connected components of $C \ominus S$. 
Therefore,
$\lambda(C \ominus B) \ge \lambda((C\cap T_z) \ominus \{z\})$ and $\Delta(C \ominus B) \ge \Delta((C\cap T_z) \ominus \{z\})$.
Case (ii) now implies
\[
  \lambda(C \ominus B) +  \partial(C \ominus B)
  \ge 
  k - \partial(B) - t
\]
It follows that
\[
  \Phi(C \ominus B)
  &\ge
  \eps\lambda(C \ominus B) + \delta\Delta(C \ominus B)
  &&\text{(induction hypothesis)}\\
  \vphantom{\Big|}&\ge
  \eps\Big(k - \partial(B) - \partial(C \ominus B) - t \Big) + \delta\Delta(C \ominus B)
  &&\text{(by the above inequality)}\\
  &\ge
  \eps\Big(k  - 6\Delta(A) - 6\Delta(B) - 3\Delta(C \ominus B) \Big) + \delta\Delta(C \ominus B)
  &&\text{(using $-\partial \ge -3\Delta$)}\\
  &=
  \eps k - 6\eps \Big(\Delta(A) + \Delta(B)\Big) + (\delta - 3\eps)\Delta(C \ominus B)\\
  &\ge
  \eps k - 6\eps \Big(\Delta(A) + \Delta(B)\Big) 
  &&\text{(since $\delta > 3\eps$)}.
\]
By inequality ($\ddag$) and the induction hypothesis, we have
\[
  \Phi(A)
  &\stackrel{(\ddag)}{\ge}
  \frac12\Phi(B) + \frac12\Phi(C \ominus B) + \frac12\Delta(A)\\
  &\ge
  \frac12\Big(\eps k + (\delta - 3\eps) \Delta(B) - \eps\Big)
  + \frac12\Big(\eps k -  6\eps \Big(\Delta(A) + \Delta(B)\Big) \Big)
  + \frac12\Delta(A)
  &&\text{(by (\ref{eq:PhiB}) and the above)}\\
  &=
  \eps k + \frac{1-6\eps}{2}\Delta(A) + \frac{\delta - 9\eps}{2}\Delta(B) - \frac{\eps}{2}\\  
  \vphantom{\Big|}&=
  \eps k + \delta\Delta(A) + \frac{1}{20}\Delta(B) - \frac{1}{60}\\
  \vphantom{\Big|}&\ge
  \eps k + \delta \Delta(A).
\]
The final step above uses $\Delta(B) \ge \partial(B)/3 \ge 1/3$ since $\Gr{B}$ is a nonempty subgraph of $T_\infty$.

\paragraph{Case (iii):} We have
\[
  \Phi(A) \stackrel{(\dag)}\ge \Phi(E) \ge \Delta(E) \ge \frac13\partial(E)
  &\ge \frac13\Big(k - \partial(B) - t\Big)
  &&\text{(by the inequality of Case (iii))}\\
  &\ge \frac13\Big(k - 6\Delta(A) - 6\Delta(B)\Big)
  &&\text{(using $-\partial \ge -3\Delta$)}\\
  &\vphantom{\Big|}= \eps k + \delta\Delta(A) + 
  \Big(\frac13-\eps\Big)k - (2+\delta)\Delta(A) - 2\Delta(B)\hspace{-3in}\\
  &= \eps k + \delta\Delta(A) + 
  \Big(\frac{3}{10}k - \frac{12}{5}\Delta(A) - 2\Delta(B)\Big).\hspace{-3in}
\]
Recalling that $\Delta(A),\Delta(B) \le k/18$ by (\ref{eq:DeltaA}), (\ref{eq:DeltaB}), we have
\[
  \frac{3}{10}k - \frac{12}{5}\Delta(A) - 2\Delta(B)
  &\ge
  \Big(\frac{3}{10} - \frac{12}{5}\cdot\frac{1}{18} - 2\cdot\frac{1}{18}\Big)k
  =
  \frac{5}{90}k > 0.
\]
This establishes that $\Phi(A) \ge \eps k + \delta \Delta(A)$ in the final case, which concludes the proof of the theorem.\qed

\subsection{Proofs of Lemmas \ref{lem:boundarylarge}, \ref{la0}, \ref{la1}}\label{sec:lemmas}

\paragraph{Proof of Lemma \ref{lem:boundarylarge}.}
{\em 
Let $H$ be a non-empty finite subgraph of $T_\infty$, all of whose components are ungrounded. Then $\partial(H) \geq \frac12 (|E(H)|+3)$.}\medskip

As $\partial(\cdot)$ and $|E(\cdot)|$ are additive over disjoint components, it suffices to prove the lemma in the case where $H$ is connected. Let $y$ be the unique highest vertex in $H$ (i.e.,\ belonging to $V_k$ for the maximal $k$), which we view as the ``root'' of $H$.

We now argue by induction on the size of $E(H)$.
If $|E(H)| = 1$, then both $y$ and one of its children lie in the boundary of $H$ and hence, we are done. If $|E(H)| = 2$, then either $H$ is the graph induced by $y$ and its two children, or $H$ is a path of length $2$ emanating from $y$. In either case, as $H$ is ungrounded, all three vertices are in the boundary of $H$ and the claim follows. 

So assume that $|E(H)| > 2$. We consider two cases:\medskip

\noindent\underline{Case 1}: $H$ contains a vertex $v$ such that both its children $v_1$ and $v_2$ are leaves of $H$. Let $H'$ be the subgraph of $H$ induced by $V(H)\setminus \{v_1,v_2\}$. By the induction hypothesis, $\partial(H') \geq \frac12 (|E(H')|+3)$. But $\partial(H) = \partial(H') + 1$ because $v_1,v_2$ are in the boundary of $H$ while $v$ is not. Therefore,
\[
\partial(H) = \partial(H') + 1 \geq \frac12 (|E(H')|+5) = \frac12 (|E(H)|+3).
\]

\noindent\underline{Case 2}: There exists a leaf $u$ of $H$ such that
$v \defeq \mr{parent}(u) \in V(H)$ and $w \defeq \mr{parent}(v)\in V(H)$ and $\deg_H(v) = 2$. Define $H'$ to be the subgraph of $H$ induced by $V(H)\setminus \{u,v\}$. By the induction hypothesis, $\partial(H') \geq \frac12 (|E(H')|+3)$. But $\partial(H) \geq \partial(H') + 1$ because $u,v$ are in the boundary of $H$ while $w$ \emph{may} or \emph{may not} be. In either case, we again have,
\[
\partial(H) \geq \partial(H') + 1 \geq \frac12 (|E(H')|+5) = \frac12 (|E(H)|+3),
\]
which completes the proof.\qed
\vspace{1ex}

\paragraph{Proof of Lemma \ref{la0}.}
{\em For all finite subgraphs $G$ of $T_\infty$, 
$\lambda(G) + \partial(G) \ge \log(|E(G)|+1)$.
}\medskip

We will prove the following equivalent inequality
\begin{equation}\label{eq:stronger}
  |E(G)| + 1\le 2^{\partial(G)+\lambda(G)}.
\end{equation}
We claim that it suffices to establish (\ref{eq:stronger}) for connected graphs $G$. To see why, assume (\ref{eq:stronger}) holds for two vertex-disjoint graphs $G$ and $H$.  Then we have
\[
  |E(G\cup H)| + 1 = (|E(G)| + 1) + (|E(H)| + 1) - 1 
  &\le 
  2^{\partial(G)+\lambda(G)}
  +
  2^{\partial(H)+\lambda(H)}-1\\
  &\le
  2^{\max\{\lambda(G),\lambda(H)\}}\cdot(2^{\partial(G)} + 2^{\partial(H)})-1\\
  &\leq
  2^{\max\{\lambda(G),\lambda(H)\}+\partial(G) + \partial(H)}-1\\
  &=
  2^{\partial(G \cup H)+\lambda(G \cup H)}-1< 2^{\partial(G \cup H)+\lambda(G \cup H)}
\]
proving (\ref{eq:stronger}) for the graph $G \cup H$.

We now prove (\ref{eq:stronger}) assuming $G$ is connected.  
If $G$ is ungrounded, then we have $\lambda(G) = 0$ and $\partial(G) \ge \frac12(|E(G)|+3)$ by Lemma \ref{lem:boundarylarge}, and so
\[
  |E(G)| + 1\le 2\partial(G)-2 < 2^{\partial(G)} = 2^{\partial(G)+\lambda(G)}.
\]

So assume now that $G$ is grounded. 
Let $y \in V_{m}$ be the unique vertex in $G$ of maximum height. Let $k \defeq \lambda(G)$ (note that $0 \le k \le {m}$) and fix a choice of $x \in V_k$ such that $T_x \subseteq G$.  
If $y = x$, then $G = T_x$ and therefore $\lambda(G) = k$ and $\partial(G) = 1$ and $|E(G)| = 2^{k+1} - 2$, so the inequality follows. 
So assume that $y \neq x$. 

As $G$ is connected, it contains the path $P$ from $y$ to $x$. Consider the case when $G$ contains only one child $y'$ of $y$. (See Figure \ref{fig:bdry1} for an example.) Then 
$G\subseteq T^+_{y'}$ and therefore $|E(G)| \leq |E(T^+_{y'})| = 2^{m} - 1$. Further, we claim that $\partial(G) \geq {m} - k$. To see this, note that for every vertex $v \neq x$ on the path $P$, if $c(v)$ denotes the child of $v$ that is {not} on the path $P$, then $G$ does {not} contain $T^+_{c(v)}$ (because otherwise, $\lambda(G) > k$). As a result, it must be the case that for every $v\neq x$ on the path $P$, some vertex in $V(G) \cap V(T^+_{c(v)})$ lies in the boundary of $G$ and so, $\partial(G) \geq {m} - k$. 
Therefore, the desired inequality follows as 
$|E(G)| + 1\leq 2^{m} = 2^{(m-k) + k}  \leq 2^{\partial(G)+\lambda(G)}.$
\begin{figure}[H]
\centering
\begin{subfigure}{0.45\linewidth}
\includegraphics[scale=0.5]{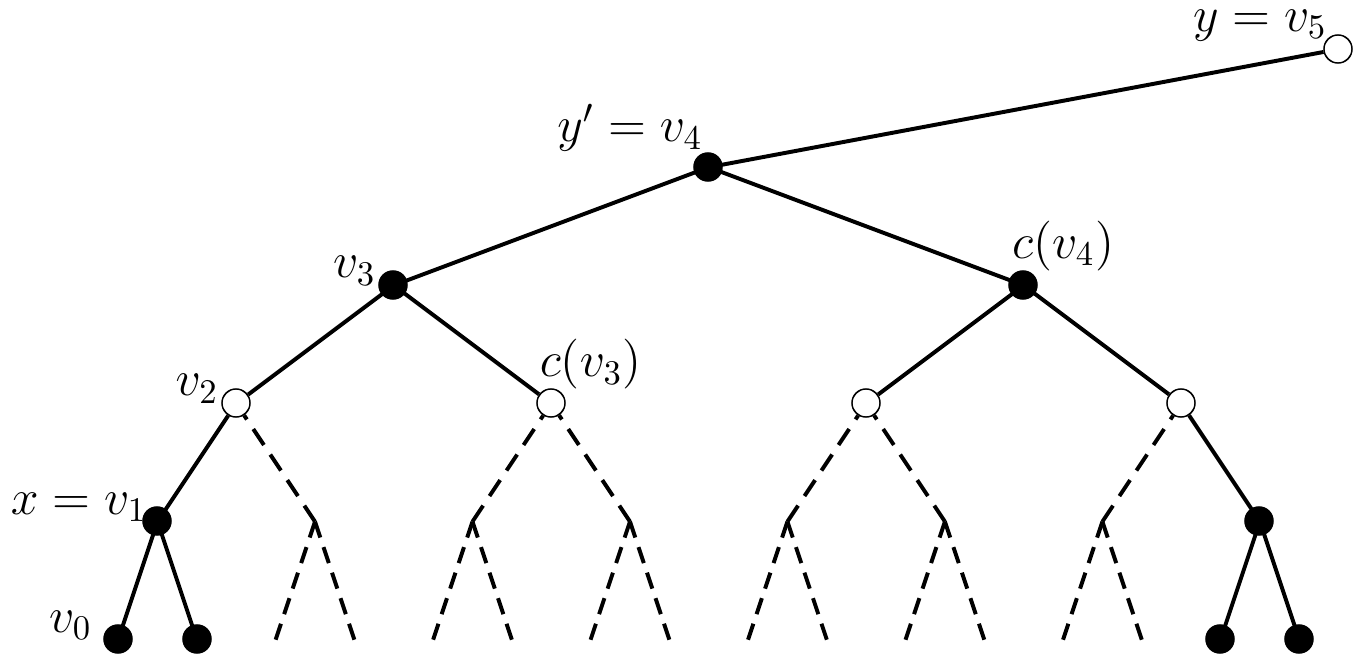}
\caption{An example where ${m} = 5$ and $y$ has only one child in $G$. 
Every vertex on the path from $y$ to $\mr{parent}(x)$ ``contributes'' to the boundary of $G$. 
(Here each $v_i$ belongs to $V_i$.)}
\label{fig:bdry1}
\end{subfigure}\quad\ 
\begin{subfigure}{0.45\linewidth}
\centering
\includegraphics[scale=0.5]{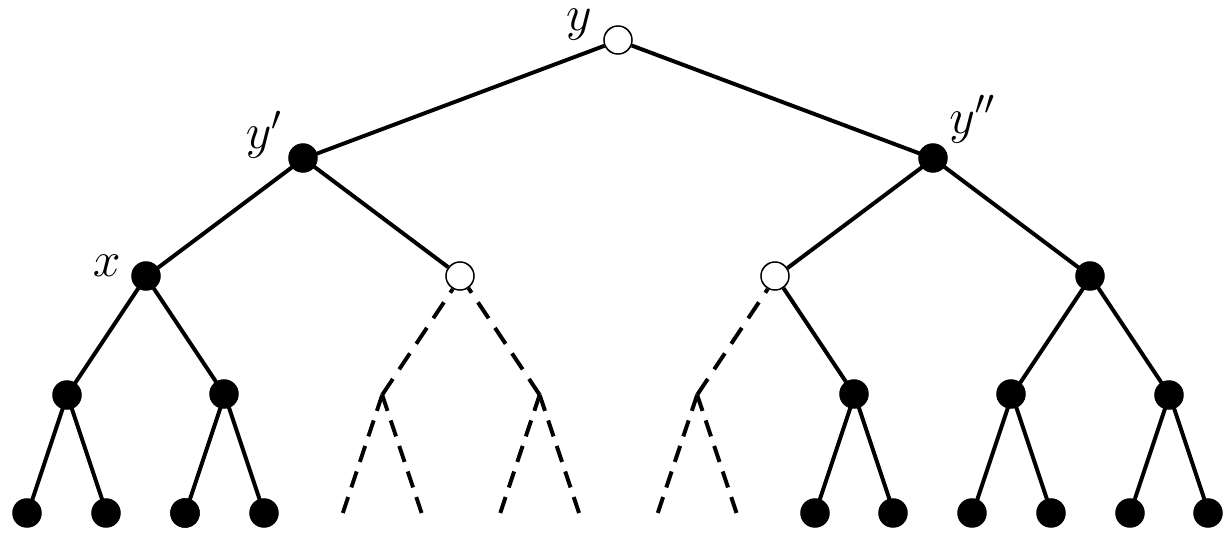}
\caption{An example where ${m} = 4$ and $G$ contains both children of $y$. 
Every vertex on the path from $y$ to $\mr{parent}(x)$ ``contributes'' to the boundary of $G$. Furthermore, $T_{y''} \not\subseteq G$, thereby providing an additional vertex in the boundary. }
\label{fig:bdry2}
\end{subfigure}
\caption{Examples of two cases that arise when the graph $G$ is grounded.}
\end{figure}

Next, suppose that $G$ contains both children of $y$. Let $y'$ be the child of $y$ on the path to $x$, and let $y''$ be its sibling.
(See Figure \ref{fig:bdry2} for an example.) Note that $G$ cannot contain the complete binary tree $T_{y''}$ (since otherwise $\lambda(G) > k$ contradicting our choice of $x$).
Therefore, there is at least one vertex in $T_{y''}$ that lies in the boundary of $G$. As a result, $\partial(G) \geq {m} - k + 1$ (the ${m}-k$ vertices identified by the previous argument, along with the additional boundary element in $V(G) \cap V(T_{y''})$).
Also as $G\subseteq T_y$, we have $|E(G)| \leq 2^{{m}+1} - 2$ and thus, $|E(G)| + 1< 2^{m+1} = 2^{(m-k+1) + k}  \leq 2^{\partial(G)+\lambda(G)}.$
\qed

\paragraph{Proof of Lemma \ref{la1}.} 
{\em Let $x \in V_k$ and suppose $G \subseteq T_x$ such that $|E(G)| \le 2^k-1$ and $\partial(G) \le k/6$. Then there exists a vertex $z \in V_{k - \partial(G)}(T_x)$ such that $E(G) \cap E(T^+_z) = \emptyset$.}\medskip

The claim is easy to establish when $G$ is empty or of the form $T_y$ or $T_y^+$ where $y \in V(T_\infty)$ (that is, in the cases where $\partial(G) \le 1$).  So we shall assume that $\partial(G) \ge 2$.

Note that $T_x$ contains $2^{\partial(G)}$ vertices with height $k-\partial(G)$, that is, $|V_{k-\partial(G)}(T_x)| = 2^{\partial(G)}$ (Observation \ref{obs:levelsize}).
Let 
\[
  Y &\defeq \{y \in V_{k-\partial(G)}(T_x) : T^+_y \subseteq G\},\\
  Z &\defeq \{z \in V_{k-\partial(G)}(T_x) : \emptyset \subsetneqq G \cap T^+_z \subsetneqq G\}.
\]
It suffices to show that $|Y| + |Z| < 2^{\partial(G)}$.

Since $\bigsqcup_{y \in Y} E(T^+_y) \subseteq E(G)$, it follows that
\[
  |Y|\cdot (2^{k-\partial(G)+1}-1) \le |E(G)| \le 2^k-1.
\]
We next observe that $\partial(G) \ge |Z| + 1$. This is because for each $z \in Z$, $G$ has at least one boundary element in the set $V(G \cap T_z)$; and we get an additional boundary element by consider any element $w \in V(G)$ of maximum height, noting that $w$ cannot lie in $V(T_z)$ for any $z \in Z$.
Therefore,
\[
  |Y| + |Z| &\le \frac{2^k-1}{2^{k-\partial(G)+1}-1} + \partial(G) - 1.
\]
Letting $b \defeq \partial(G)$, it suffices to show that
\[
  \frac{2^k-1}{2^{k-b+1}-1} + b - 1 < 2^b 
  \text{, or equivalently,\ }
  2^b + (b-1)(2^{k-b+1}-1) < 2^k + 1.
\]
This numerical inequality is simple to verify 
for all $2 \le b \le k/6$, so we are clearly done by the assumption that $2 \le \partial(G) \le k/6$.

\section{A better potential function}\label{sec:Phi} 

We again fix a graph $G$ and a threshold weighting $\theta$.  In this section, we define a potential function $\Phi_\theta(A|S)$ with two parameters: a join-tree $A$ and a set $S \subseteq V(G)$.
This improved potential function serves the same purpose of lower-bounding $\tau(G)$. 
In Section \ref{sec:Pk} we use $\Phi_\theta(A|S)$ to obtain a better lower bound on $\tau(P_k)$.

Let us begin by recalling the defining inequalities for $\Phi_\theta(A)$:
\begin{align}
\tag{$\dag$}
  \Phi_\theta(A) &\ge \Phi_\theta(D) + \Delta_\theta(C \ominus D) + \Delta_\theta(A \ominus (C \cup D))
  &&\text{if $A \in \{\sq{B,C},\sq{C,B}\}$ and $D \preceq B$},\\
\tag{$\ddag$}
  \Phi_\theta(A) &\ge \frac12\Big(\Phi_\theta(D) + \Phi_\theta(E \ominus D) + \Delta_\theta(A) + \Delta_\theta(A \ominus (D \cup E))\Big)
  &&\text{if $D,E \prec A$}.
\end{align}
A first observation toward improving $\Phi_\theta(A)$ is that we could have included an additional inequality in the definition of $\Phi_\theta(A)$ while maintaining Theorem \ref{thm:tau} (the lower bound on $\tau(G)$ in terms of $\Phi_\theta(A)$).  We call this inequality ($\gad$), since it is a variant of ($\dag$):
\begin{align}
\tag{$\gad$}
  \Phi_\theta(A) &\ge \Phi_\theta(D \ominus C) + \Delta_\theta(C) + \Delta_\theta(A \ominus (C \cup D))
  &&\text{if $A \in \{\sq{B,C},\sq{C,B}\}$ and $D \preceq B$}.
\end{align}
A second observation is that, in the recursive view of $\Phi_\theta(A)$, we ``shrink'' more than necessary by passing to $\Phi_\theta(D \ominus C)$ in ($\gad$) and $\Phi_\theta(E \ominus D)$ in ($\ddag$). Recall that $D \ominus C$ is a join-tree with graph $\Gr{D} \ominus V(C)$ formed by the connected components of $\Gr{D}$ that are vertex-disjoint from $V(C)$.  Rather than recursing on $D \ominus C$, we can instead simply recurse on $D$ while treating ``${\ominus}\,C$'' as an extra parameter.  These two observations lead to the definition of $\Phi_\theta(A|S)$ below.

\begin{notn}
The following alternative notation for $\Delta_\theta$ will be convenient in what follows.  For a graph $F \subseteq G$ and a set $S \subseteq V(G)$, we write $\Delta_\theta(F|S)$ for $\Delta_\theta(F \ominus S)$.  Similarly, for a join-tree $A$, we write $\Delta_\theta(A|S)$ for $\Delta_\theta(\Gr{A} \ominus S)$.  
\end{notn}

\begin{df}[The potential function $\Phi_\theta(A|S)$]
\label{df:Phi2}
Let $\Phi_\theta : \{$join-trees for subgraphs of $G\} \times \{$subsets of $V(G)\} \to \R_{\ge 0}$ be the unique pointwise minimum function --- written as $\Phi_\theta(A|S)$ rather than $\Phi_\theta(A,S)$ --- satisfying the following inequalities for all sets $S \subseteq V(G)$ and join-trees $A,B,C,D$:
\begin{align}
\tag{$\dag$}
  &&&&\Phi_\theta(A|S) &\ge \vphantom{\frac12}
  \Phi_\theta(B|S) + \Delta_\theta(C|S \cup B)
  &&\hspace{-0in}\text{if $A \in \{\sq{B,C},\sq{C,B}\}$},\\
\tag{$\gad$}
  &&&&\Phi_\theta(A|S) &\ge 
  \Delta_\theta(B|S) + \Phi_\theta(C|S \cup B)
  &&\hspace{-0in}\text{if $A \in \{\sq{B,C},\sq{C,B}\}$},\\
\tag{$\ddag$}
  &&&&\Phi_\theta(A|S) &\ge \frac12\Big(\Phi_\theta(D|S) + \Phi_\theta(A|S \cup D) + \Delta_\theta(A|S)\Big)
  &&\text{if $D \prec A$}.&&&&
\end{align}
Alternatively, $\Phi_\theta(A|S)$ has the following recursive characterization:

\begin{itemize}
\item
If $A$ is an atomic join-tree, then $\Phi_\theta(A|S) \defeq \Delta_\theta(A|S)$.
\item
If $A = \sq{B,C}$, then
\[
  \Phi_\theta(A|S) \defeq 
  \max
  \left\{
  \begin{aligned}
  &\vphantom{\Big|}\Phi_\theta(B|S) + \Delta_\theta(C|S \cup B),\ 
  \Phi_\theta(C|S) + \Delta_\theta(B|S \cup C),\:\\
  &\Delta_\theta(B|S) + \Phi_\theta(C|S \cup B),\ 
  \Delta_\theta(C|S) + \Phi_\theta(B|S \cup C),\:\\
  &\max_{D \prec A}\, \smash{\frac12}\Big(\Phi_\theta(D|S) + \Phi_\theta(A|S \cup D) + \Delta_\theta(A|S)\Big)
  \end{aligned}
  \right\}.
\]
(To avoid circularity, we take this $\max_{D \prec A}$ over proper sub-join-trees $D \prec A$ such that $V(D) \nsubseteq S$.)
\end{itemize}
That is, at least one among inequalities ($\dag$), ($\gad$), ($\ddag$) is tight for each join-tree $A$.  
\end{df}

\begin{rmk}
We could have defined $\Phi_\theta(A|S)$ in a stronger manner that ensures $\Phi_\theta(A \ominus S) \le \Phi_\theta(A|S)$ for all $A$ and $S$ (in order to claim that $\Phi_\theta(A|\emptyset)$ improves $\Phi_\theta(A)$) by using the following more general versions of ($\dag$), ($\gad$), ($\ddag$):
\begin{align}
\tag{$\dag'$}
  \Phi_\theta(A|S) &\ge \vphantom{\frac12}\Phi_\theta(D|S) + \Delta_\theta(C |S \cup D) + \Delta_\theta(A |S \cup C \cup D)
  &&\hspace{-.8in}\text{if $A \in \{\sq{B,C},\sq{C,B}\}$ and $D \preceq B$},\\
\tag{$\gad'$}
  \Phi_\theta(A|S) &\ge \Phi_\theta(D |S \cup C) + \Delta_\theta(C|S) + \Delta_\theta(A |S \cup C \cup D)
  &&\hspace{-.8in}\text{if $A \in \{\sq{B,C},\sq{C,B}\}$ and $D \preceq B$},\\
\tag{$\ddag'$}
  \Phi_\theta(A|S) &\ge \frac12\Big(\Phi_\theta(D|S) + \Phi_\theta(E |S \cup D) + \Delta_\theta(A|S) + \Delta_\theta(A |S \cup D \cup E)\Big)
  &&\text{if $D,E \prec A$}.
\end{align}
We chose the simpler Definition \ref{df:Phi2} since it is sufficient for our lower bound on $\tau(P_k)$ in Section \ref{sec:Pk}. Definition \ref{df:Phi2} also leads to a mildly simpler proof of Theorem \ref{thm:general-chiPhi} in Section \ref{sec:pathset}, compared with inequalities ($\dag')$, ($\gad'$), ($\ddag'$) above.
\end{rmk}

The following theorem shows that $\Phi_\theta(A|\emptyset)$ serves the same purpose as $\Phi_\theta(A)$ of lower-bounding the invariant $\tau(G)$.

\begin{thm}\label{thm:tau2}
The invariant $\tau(G)$ in Theorem \ref{thm:main-lb} satisfies
\[
  \tau(G) \ge 
  \max_{\textup{threshold weightings $\theta$ for $G$}}\  
  \min_{\textup{join-trees $A$ with graph $G$}}\ \Phi_\theta(A|\emptyset).
\]
\end{thm}

In the next section, we will finally state the definition of $\tau(G)$ and prove Theorem \ref{thm:tau2}.  We will then use this theorem to prove a lower bound on $\tau(P_k)$ using $\Phi_\theta(A|S)$ in Section \ref{sec:pathset}.  (The argument in Section \ref{sec:pathset} is purely graph-theoretic and does not require the material in Section \ref{sec:pathset} if Theorem \ref{thm:tau2} is taken for granted.)

\begin{rmk}
The authors first proved the lower bound for $\tau(T_k)$ using $\Phi_\theta(A)$ before considering the improved potential function $\Phi_\theta(A|S)$.  It is conceivable that the use of $\Phi_\theta(A|S)$ would simplify or improve the constant in our $k/30$ lower bound.  
On the other hand, a suitable induction hypothesis would have to account for the additional parameter $S$, so it is unclear whether a dramatic simplification can be achieved.
\end{rmk}

\section{The pathset framework}\label{sec:pathset}

In this section we present the {\em pathset framework}, state the definition of $\tau(G)$, and prove Theorem \ref{thm:tau2} (which bounds $\tau(G)$ in terms of the potential function $\Phi_\theta(A|S)$).  All definitions and results in this section are from papers \cite{rossman2015correlation,rossman2018formulas,RossmanICM} (with  minor modifications); a few straightforward lemmas are stated without proof.  The reader is referred to those papers for much more context, illustrative examples, and an explanation of how $\tau(G)$ provides a lower bound on the $\ACzero$ and monotone formulas size of $\SUB(G)$.

Throughout this section, we fix a threshold-weighted graph $(G,\theta)$ and an arbitrary positive integer $n$.  Let $F$ range over subgraphs of $G$, let $S,T$ range over subsets of $V(G)$, and let $A,B,C,D$ range over join-trees for subgraphs of $G$.

\begin{df}[Relations, density, join, projection, restriction]\
Let $V,W,T$ be arbitrary finite sets.
\begin{itemize}
\item
For a tuple $x \in [n]^V$ and subset $U \subset V$, let $x_U \in [n]^U$ denote the restriction of $x$ to coordinates in $U$.
\item
For a relation $\A \subseteq [n]^V$, the {\em density} of $\A$ is denoted
\[
  \mu(\A) \defeq |\A|/n^{|V|}.
\]
\item
For relations $\A \subseteq [n]^V$ and $\B \subseteq [n]^W$, the join $\A \bowtie \B \subseteq [n]^{V \cup W}$ is the relation defined by
\[
  \A \bowtie \B \defeq \{z \in [n]^{V \cup W} : x_V \in \A \text{ and } z_W \in \B\}.
\]
\item
For $\A \subseteq [n]^V$ and $U \subseteq V$, the {\em projection} $\proj{U}{\A} \subseteq [n]^U$ is defined by
\[
  \proj{U}{\A} \defeq \{x \in [n]^U : \exists y \in \A \text{ s.t.\ } y_U = x\}.
\]
\item
For $\A \subseteq [n]^V$ and $z \in [n]^T$, the {\em restriction} $\RHO{V \setminus T}{\A}{z} \subseteq [n]^{V \setminus T}$ is defined by
\[
  \RHO{V \setminus T}{\A}{z} \defeq \{x \in [n]^{V \setminus T} : \exists y \in \A \text{ s.t.\ } y_{V \setminus T} = x_{V \setminus T} \text{ and } y_{V \cap T} = z_{V \cap T}\}.
\]
\end{itemize}
\end{df}

The next two lemmas bound the density of relations in terms of projections and restrictions.

\begin{la}\label{la:split}
For every relation $\A \subseteq [n]^V$ and $U \subseteq V$, 
\[
  \mu(\A) &\le \mu(\proj{U}{\A}) \cdot \max_{z \in [n]^U} \mu(\RHO{V \setminus U}{\A}{z}).
\]
\end{la}

\begin{la}\label{la:split2}
For all relations $\A \subseteq [n]^V$ and $\B \subseteq [n]^W$,
\[
  \mu(\A \bowtie \B) &\le \mu(\A) \cdot \max_{z \in [n]^V} \mu(\RHO{W}{\B}{z}).
\]
\end{la}

For subgraph $F \subseteq G$ and sets $S \subseteq V(G)$, we will be interested in relations $\A \subseteq [n]^{V(G)\setminus S}$ called {\em $G|S$-pathsets} that satisfy certain density constraints.  These density constraints are related to subgraph counts in the random graph $\mb X_{\theta,n}$ (see \cite{RossmanICM}). 

\begin{df}[Pathsets]
Let $F \subseteq G$ and $S \subseteq V(G)$.
\begin{itemize}
\item
We write $[n]^{F|S}$ for $[n]^{V(F)\setminus S}$.
\item
We write $\Rel{F|S}$ for the set of relations $\A \subseteq [n]^{F|S}$.  (That is, $\Rel{F|S}$ is the power set of $[n]^{F|S}$.)
\item
For $\A \in \Rel{F|S}$ and $F' \subseteq F$, let
$
  \proj{F'|S}{\A} \defeq \proj{V(F') \setminus S}{\A}.
$
\item
For $\A \in \Rel{F|S}$ and $T \supseteq S$ and $z \in [n]^T$, let
$
  \RHO{F|T}{\A}{z}
  \defeq 
  \RHO{V(F)\setminus T}{\A}{z}.
$
\item
A relation $\A \in \Rel{F|S}$ is a {\em $F|S$-pathset} if it satisfies
\[
\mu(\RHO{F|T}{\A}{z}) \le (1/n)^{\Delta(F|T)}
\]
for all $T \supseteq S$ and $z \in [n]^T$. The set of all $F|S$-pathsets is denoted by $\P{F|S}$.
\end{itemize}
\end{df}

The next lemma is immediate from the definition of $\P{F|S}$.

\begin{la}[Pathsets are closed under restriction]
For all $\A \in \P{F|S}$ and $T \supseteq S$ and $z \in [n]^T$, we have $\RHO{F|T}{\A}{z} \in \P{F|T}$.
\end{la}

We next introduce, for each join-tree $A$ and set $S$, a complexity measure $\chi_{A|S}$ on relations $\A \subseteq [n]^{V(A)\setminus S}$.  Very roughly speaking, $\chi_{A|S}$ measures the cost of ``constructing'' $\A$ via operations $\cup$ and $\bowtie$, where all intermediate relations are subject to pathset constraints and the pattern of joins is given by $A$.

\begin{df}[Pathset complexity $\chi_{A|S}(\A)$]\label{df:chi}\ 
\begin{itemize}
\item
For a join-tree $A$, let
  $\Rel{A|S} \defeq \Rel{\GG{A}|S}$, 
  $\P{A|S} \defeq \P{\GG{A}|S}$,
etcetera.
\item
For an atomic join-tree $A$ and relation $\A \in \Rel{A|S}$, the {\em $A|S$-pathset complexity} of $\A$, denoted $\chi_{A|S}(\A)$, is the minimum number $m$ of pathsets $\A_1,\dots,\A_m \in \P{A|S}$ such that $\A \subseteq \bigcup_{i=1}^m \A_i$.
\item
For a non-atomic join-tree $A = \sq{B,C}$ and relation $\A \in \Rel{A|S}$, 
the {\em $A|S$-pathset complexity} of $\A$, denoted $\chi_{A|S}(\A)$, is 
the minimum value of $\sum_{i=1}^m \max\{\chi_{B|S}(\B_i),\, \chi_{C|S}(\C_i)\}$ over families $\{(\A_i,\B_i,\C_i)\}_{i\in[m]}$ 
satisfying
\begin{enumerate}[$\circ$]
\item
$(\A_i,\B_i,\C_i) \in \P{A|S} \times \P{B|S} \times \P{C|S}$ for all $i$,
\item
$\vphantom{\big|}\A_i \subseteq \B_i \bowtie \C_i$ for all $i$, and
\item
$\A \subseteq \bigcup_{i=1}^m \A_i$.
\end{enumerate}
We express this concisely as:
\[
  \chi_{A|S}(\A) \defeq \min_{\{(\A_i,\B_i,\C_i)\}_i}\, \ts\sum_i\, \max\{\chi_{B|S}(\B_i),\, \chi_{C|S}(\C_i)\}.
\]
We refer to any family $\{(\A_i,\B_i,\C_i)\}_i$ achieving this minimum as a {\em witnessing family} for $\chi_{A|S}(\A)$.
\end{itemize}
\end{df}

The next lemma lists three properties of $\chi_{A|S}$, which are easily deduced from the definition.

\begin{la}[Properties of $\chi_{A|S}$]\label{la:chi-props}\ 
\begin{itemize}
\item
$\chi_{A|S}$ is subadditive: $\chi_{A|S}(\A_1 \cup \A_2) \le \chi_{A|S}(\A_1) + \chi_{A|S}(\A_2)$,
\item
$\chi_{A|S}$ is monotone: $\A_1 \subseteq \A_2 \ \Longrightarrow\  \chi_{A|S}(\A_1) \le \chi_{A|S}(\A_2)$,
\item 
if $A = \sq{B,C}$ and $\B \in \P{B|S}$ and $\C \in \P{C|S}$, then $\chi_{A|S}(\B \bowtie \C) \le \max\{\chi_{B|S}(\B),\,\chi_{C|S}(\C)\}$.
\end{itemize}
\end{la}

The next two lemmas show that pathset complexity is non-increasing under restrictions, as well as under projections to sub-join-trees. 

\begin{la}[Projection Lemma]\label{la:proj}
For all $\A \in \Rel{A|S}$ and $B \preceq A$, we have 
\[
  \chi_{B|S}(\proj{B|S}{\A}) 
  &\le 
  \chi_{A|S}(\A).
\]
\end{la}

\begin{proof}
It suffices to prove the lemma in the case where $A = \sq{B,C}$ (since $\prec$ is the transitive closure of ``$B$ is a child of $A$''). 
Fix a witnessing family $\{(\A_i,\B_i,\C_i)\}_i$ for $\chi_{A|S}(\A)$. Note that $\proj{B|S}{\A} \subseteq \bigcup_i \B_i$, since $\A \subseteq \bigcup_i \A_i$ and $\A_i \subseteq \B_i \bowtie \C_i$. It follows that
\[
  \chi_{B|S}(\proj{B|S}{\A})
  &\le
  \chi_{B|S}(\ts\bigcup_i \B_i)\\
  &\le
  \ts\sum_i \chi_{B|S}(\B_i)\\
  &\le
  \ts\sum_i \max\{\chi_{B|S}(\B_i),\, \chi_{C|S}(\C_i)\}\\
  &=
  \chi_{A|S}(\A).\qedhere
\]
\end{proof}

\begin{la}[Restriction Lemma]\label{lem:res}\label{la:rest}
For all $\A \in \Rel{A|S}$ and $T \supseteq S$ and $z \in [n]^T$, we have
\[
  \chi_{A|T}(\RHO{A|T}{\A}{z}) 
  &\le 
  \chi_{A|S}(\A).
\]
\end{la}

\begin{proof}
We argue by induction on join-trees $A$. The lemma is trivial in the base case where $A$ is atomic. So assume $A = \sq{B,C}$. Fix a witnessing family $\{(\A_i,\B_i,\C_i)\}_i$ for $\chi_{A|S}(\A)$.
Observe that the family of restricted triples $\{(\RHO{A|T}{\A_i}{z},\RHO{B|T}{\B_i}{z},\RHO{C|T}{\C_i}{z})\}_i$ satisfies:
\begin{enumerate}[\hspace{\parindent}$\circ$]
\item
$(\RHO{A|T}{\A_i}{z},\RHO{B|T}{\B_i}{z},\RHO{C|T}{\C_i}{z}) \in \P{A|T}\times\P{B|T}\times\P{C|T}$ for all $i$,
\item
$\RHO{A|T}{\A_i}{z} \subseteq \RHO{B|T}{\B_i}{z} \bowtie \RHO{C|T}{\C_i}{z}$ for all $i$, 
\item
$\RHO{A|T}{\A}{z} \subseteq \bigcup_i \RHO{A|T}{\A_i}{z}$.
\end{enumerate}
Therefore,
\[
  \chi_{A|T}(\RHO{A|T}{\A}{z}) &\le 
  \ts\sum_i\max\{\chi_{B|T}(\RHO{B|T}{\B_i}{z}),\, 
  \chi_{C|T}(\RHO{C|T}{\C_i}{z})\}\\
  &\le
  \ts\sum_i\max\{\chi_{B|S}(\B_i),\, 
  \chi_{C|S}(\C_i)\}
  \qquad\text{(induction hypothesis)}\\
  &=
  \chi_{A|S}(\A).\qedhere
\]
\end{proof}

We now prove the key theorem bounding pathset complexity $\chi_{A|S}(\A)$ in terms of the density of $\A$ and the potential function $\Phi(A|S)$.

\begin{thm}\label{thm:general-chiPhi}
For every join-tree $A$ and set $S$ and relation $\A \in \Rel{A|S}$, we have
\[
  \mu(\A) \le (1/n)^{\Phi(A|S)} \cdot \chi_{A|S}(\A).
\]
\end{thm}

\begin{proof}
We argue by induction on join-trees $A$. The base case where $A$ is atomic is straightforward. Suppose $\chi_{A|S}(\A) = m$ where $\A_1,\dots,\A_m \in \P{A|S}$ with $\A \subseteq \bigcup_i\A_i$. We have $\mu(\A) \le \sum_i \mu(\A_i)$. For each $i$, we have $\mu(\A_i) \le (1/n)^{\Phi(A|S)}$ (by definition of $\A_i \in \P{A|S}$). Therefore, $\mu(\A) \le (1/n)^{\Phi(A|S)} \cdot m = (1/n)^{\Phi(A|S)} \cdot \chi_{A|S}(\A)$ as required.

So we assume $A = \sq{B,C}$ is non-atomic. Fix a witnessing family $\{(\A_i,\B_i,\C_i)\}_i$ for $\chi_A(\A)$. Since $\mu(\A) \le \sum_i \mu(\A_i)$, it suffices to show, for each $i$, that
\begin{align}\label{eq:to-show}
  \mu(\A_i) 
  &\le 
  (1/n)^{\Phi(A|S)} \cdot \max\{\chi_{B|S}(\B_i),\, \chi_{C|S}(\C_i)\}.
\end{align}
We establish (\ref{eq:to-show}) by considering the three different cases for $\Phi(A|S)$, according to which of the inequalities ($\dag$), ($\gad$) ($\ddag$) is tight.

\paragraph{Case ($\dag$): $\Phi(A|S) = \Phi(B|S) + \Delta(C|S \cup B)$ (or symmetrically $\Phi(A|S) = \Phi(C|S) + \Delta(B|S \cup C)$)}\ \\

Since $\A_i \subseteq \B_i \bowtie \C_i$ and $\B_i \in \P{B|S}$, we have $\proj{B|S}{\A_i} \subseteq \B_i$ and $\mu(\RHO{C|S \cup B}{\A_i}{z} \subseteq \C_i$ for any $z \in [n]^{B|S}$. We now have (\ref{eq:to-show}) as follows:
\[
  \mu(\A_i) 
  &\le 
  \mu(\proj{B|S}{\A_i})
  \max_{z \in [n]^{B|S}}
  \mu(\RHO{C|S \cup B}{\A_i}{z})
  &&\text{(Lemma \ref{la:split2})}\\
  &\le
  \mu(\B_i)
  \max_{z \in [n]^{B|S}}
  \mu(\RHO{C|S \cup B}{\C_i}{z})
  &&\text{(by above observations)}\\
  &\le 
  (1/n)^{\Delta(C|S \cup B)} \cdot
  \mu(\B_i)
  &&\text{(since $\C_i \in \P{C|S}$)}\\
  &\le 
  (1/n)^{\Delta(C|S \cup B)} \cdot (1/n)^{\Phi(B|S)} \cdot
  \chi_{B|S}(\B_i)
  &&\text{(induction hypothesis)}\\
  &\le 
  (1/n)^{\Phi(A|S)} \cdot
  \max\{\chi_{B|S}(\B_i),\, \chi_{C|S}(\C_i)\}.
\]

\paragraph{Case ($\protect\gad$): $\Phi(A|S) = \Delta(B|S) + \Phi(C|S \cup B)$ (or symmetrically $\Phi(A|S) = \Delta(C|S) + \Phi(B|S \cup C)$)}\ \\

We show (\ref{eq:to-show}) as follows:
\[
  \mu(\A_i) 
  &\le
  \mu(\B_i)
  \max_{z \in [n]^{B|S}}
  \mu(\RHO{C|S \cup B}{\C_i}{z})
  &&\text{(as in the previous case)}\\
  &\le
  (1/n)^{\Delta(B|S)}
  \max_{z \in [n]^{B|S}}
  \mu(\RHO{C|S \cup B}{\C_i}{z})
  &&\text{(since $\B_i \in \P{B|S}$)}\\
  &\le
  (1/n)^{\Delta(B|S)}
  \max_{z \in [n]^{B|S}}
  (1/n)^{\Phi(C|S\cup B)} \cdot \chi_{C|S \cup B}(\RHO{C|S \cup B}{\C_i}{z})
  &&\text{(induction hypothesis)}\\
  &\le 
  (1/n)^{\Delta(B|S)} \cdot (1/n)^{\Phi(C|S\cup B)} \cdot \chi_{C|S}(\C_i)
  &&\text{(Restriction Lemma \ref{la:rest})}\\
  &\le
  (1/n)^{\Phi(A|S)} \cdot \max\{\chi_{B|S}(\B_i),\, \chi_{C|S}(\C_i)\}.
\]

\paragraph{Case ($\ddag$): $\Phi(A|S) = \frac12\big(\Phi(D|S) + \Phi(A|S \cup D) + \Delta(A|S)\big)$ for some $D \prec A$}\ \\

We have
\[
  \mu(\A_i) 
  &\le
  \mu(\proj{D|S}{\A_i})
  \max_{z \in [n]^{D|S}}
  \mu(\RHO{A|S \cup D}{\A_i}{z})
  &&\text{(Lemma \ref{la:split})}\\
  &\le
  (1/n)^{\Phi(D|S) + \Phi(A|S \cup D)} \cdot \chi_{D|S}(\proj{D|S}{\A_i})
  \max_{z \in [n]^{D|S}}
  \chi_{A|S \cup D}(\RHO{A|S \cup D}{\A_i}{z})
  &&\text{(induction hyp.)}\\
  &\le
  (1/n)^{\Phi(D|S) + \Phi(A|S \cup D)}  \cdot  \chi_{A|S}(\A_i) \cdot
  \chi_{A|S \cup D}(\RHO{A|S \cup D}{\A_i}{z})
  &&\text{(Proj.\ Lemma \ref{la:proj})}\\
  &\le
  (1/n)^{\Phi(D|S) + \Phi(A|S \cup D)} \cdot \big(\chi_{A|S}(\A_i)\big)^2
  &&\text{(Rest.\ Lemma \ref{la:rest})}.
\]
Since $\A_i \in \P{A|S}$, we also have 
\[
  \mu(\A_i) \le (1/n)^{\Delta(A|S)}.
\] 
Taking the product of the square roots of these two bounds on $\mu(\A_i)$, we conclude
\[
  \mu(\A_i) 
  &\le 
  (1/n)^{\frac12(\Phi(D|S)+\Phi(A|S \cup D)+\Delta(A|S))} \cdot \chi_{A|S}(\A_i)\\
  &\le
  (1/n)^{\Phi(A|S)} \cdot \chi_{A|S}(\B_i \bowtie \C_i)
  &&\text{(since $\A_i \subseteq \B_i \bowtie \C_i$)}\\
  &\le
  (1/n)^{\Phi(A|S)} \cdot \max\{\chi_{B|S}(\B_i),\, \chi_{C|S}(\C_i)\}
  &&\text{(Lemma \ref{la:chi-props})}.
\]
Having established (\ref{eq:to-show}) in all three cases, the proof is complete.
\end{proof}

Finally, we define the graph invariant $\tau(G)$. Since we  no longer fix a particular threshold weighting $\theta$ and positive integer $n$, we include these as subscripts to the pathset complexity function by writing $\chi_{\theta,n,A|S}(\A)$ for relations $\A \subseteq [n]^{V(G) \setminus S}$.

\begin{df}[The parameter $\tau(G)$ of Theorem \ref{thm:main-lb}]\label{df:tau}
For a graph $G$, let $\tau(G) \in \R_{\ge 0}$ be the minimum real number such that
\[
  \chi_{\theta,n,A|S}(\A) \ge n^{\tau(G)} \cdot \mu(\A)
\]
for every threshold weighting $\theta$ on $G$, join-tree $A$ with graph $G$, positive integer $n$, and relation $\A \subseteq [n]^{V(G)}$.
\end{df}

It is evident from this definition that Theorem \ref{thm:tau2} is an immediate corollary of Theorem \ref{thm:general-chiPhi}.

\section{Lower bound $\tau(P_k) \ge \log_{\sqrt 5 + 5}(k)$}\label{sec:Pk}

Throughout this section, we fix the infinite pattern graph $P_\infty$ and threshold weighting $\theta : E(P_\infty) \to \{1\}$.  Let $F$ range over finite subgraphs of $P_\infty$, and let $S$ range over finite subsets of $V(P_\infty)$ ($=\Z$).  We suppress $\theta$, writing $\Delta(F|S)$ instead of $\Delta_\theta(F|S)$. 

For integers $i < j$, let $P_{i,j} \subseteq P_\infty$ be the path from $i$ to $j$ (with edges $\{i,i+1\},\{i+1,i+2\},\dots,\{j-1,j\}$).  For $k \ge 0$, let $P_k \defeq P_{0,k}$.

\begin{df}[Open/half-open/closed components of $F|S$]
For integers $i < j$ such that $P_{i,j} \subseteq F$, we say that
\begin{itemize}
\item
$(i,j)$ is an {\em open component} of $F|S$ if $V(P_{i,j}) \cap S = \{i,j\}$,
\item
$[i,j)$ is a {\em half-open component} of $F|S$ if $V(P_{i,j}) \cap S = \{i\}$ and $\{j,j+1\} \notin E(G)$,
\item
$(i,j]$ is a {\em half-open component} of $F|S$ if $V(P_{i,j}) \cap S = \{j\}$ and $\{i-1,i\} \notin E(F)$,
\item
$[i,j]$ is a {\em closed component} of $F|S$ if $V(P_{i,j}) \cap S = \emptyset$ and $\{i,i-1\} \notin E(F)$ and $\{j,j+1\} \notin E(F)$.
\end{itemize}
We shall use the term `interval' and component interchangeably. In each of the four cases above, we define the length of the interval to be $j-i$ and refer to $i$ and $j$ as the left and right `end-points' of that interval, respectively. We shall also refer to the length of an interval $I$ by $|I|$. We treat a join-tree $A$ as its graph $\Gr{A}$ when speaking of the {\em open/half-open/closed components} of $A|S$.

As a matter of notation, let $V((i,j)) = \{i+1,\dots,j-1\}$ and $V([i,j)) = \{i+1,\dots,j\}$ and $V((i,j]) = \{i+1,\dots,j\}$ and $V([i,j]) = \{i,\dots,j\}$.
\end{df}

\begin{figure}[H]
\centering
\includegraphics[scale=0.60]{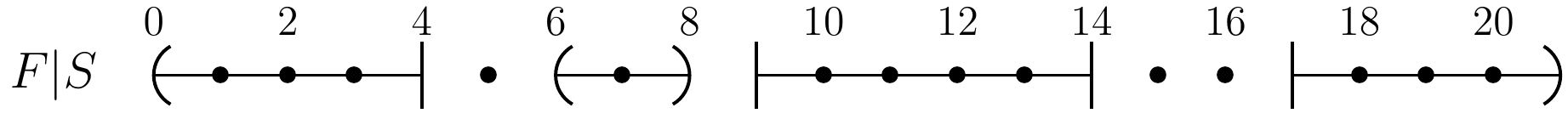}
\caption{\small An example where $F$ is the graph $P_{0,4} \cup P_{6,8} \cup P_{9,14} \cup P_{17,21}$ and $S = \{0,6,8, 21\}$. Then $F|S$ consists of four components: $(0,4], (6,8),[9,14],$ and $[17,21)$. We follow the same format for the rest of the figures in this section: `$($', `$)$' indicate that the corresponding (left or right, respectively) end-point of that component is in $S$, while `$|$' indicates that it is not. We shall skip marking the internal vertices of an interval.}
\label{fig:P_inf}
\end{figure}

\begin{la}\label{lem:Delta=closed}
$\Delta(F|S)$ equals the number of closed components of $F|S$.
\end{la}

\begin{proof}
We have
$
  \Delta(F|S) = \Delta(F \ominus S) = |V(F \ominus S)| - |E(F \ominus S)| = \#\{$connected components of $F \ominus S\}
  = \#\{$closed components of $F|S\}.
$
\end{proof}
Recall that for any any graph $F$ and set $T$, we denote by $F[T]$ the subgraph induced by the vertices of $V(F)\cap T$. We have the following observation that follows immediately from Lemma~\ref{lem:Delta=closed}:

\begin{obs}\label{obs:DeltaSplit}
For a graph $F$, set $S$ and a set $T$ such that for every open/half-open/closed component $K$ of $F|S$, either $T \cap V(K) = \emptyset$ or $V(K) \subseteq T$,
\[
\Delta(F|S) = \Delta(F[T]|S) + \Delta (F|S\cup T)
\]
where note that $\Delta(F[T]|S)$ equals the number of closed components of $F|S$ that $T$ contains.
\end{obs}

The following lemma will allow us to ``zoom'' into any component of $A|S$ when bounding $\Phi(A|S)$ and gain $1$ for each closed component that we discard. 

\begin{la}\label{la:zoom}
Let $A$ be a join-tree and let $S,T$ be subsets of $V(A)$ such that, for every open/half-open/closed component $K$ of $A|S$, either $T \cap V(K) = \emptyset$ or $V(K) \subseteq T$.  Then 
\[
  \Phi(A|S) \ge \Phi(A|S \cup T) + \Delta(A[T]|S).
\]

\end{la}

\begin{proof}
We argue by induction on join-trees $A$ and by a backward induction on the set $S$.  The lemma is trivial when $A$ is atomic, or when $S = V(A)$.  So assume $A = \sq{B,C}$. Note that the given condition means that $A[T]|S$ is a union of components of $A|S$ and since $B|S, C|S$ are subgraphs, $B[T]|S$ is a union of components of $B|S$ (and similarly for $C|S$ and $C|S\cup B$). We start with the observation that 
\begin{equation}\label{eq:DeltaCB}
\Delta(B[T]|S) + \Delta(C[T]|S\cup B) \geq \Delta(A[T]|S).
\end{equation}
To see this, note that as the graph $A[T]|S$ is a union of the graphs $B[T]|S$ and $C[T]|S$, each closed component of $A[T]|S$ either contains at least one closed component of $B[T]|S$, or it does not. In the latter case, it is then clear that it must also be a component of $C[T]|S\cup B$. See Figure~\ref{fig:zoom} for an illustration.

\begin{figure}[H]
\centering
\includegraphics[scale=0.60]{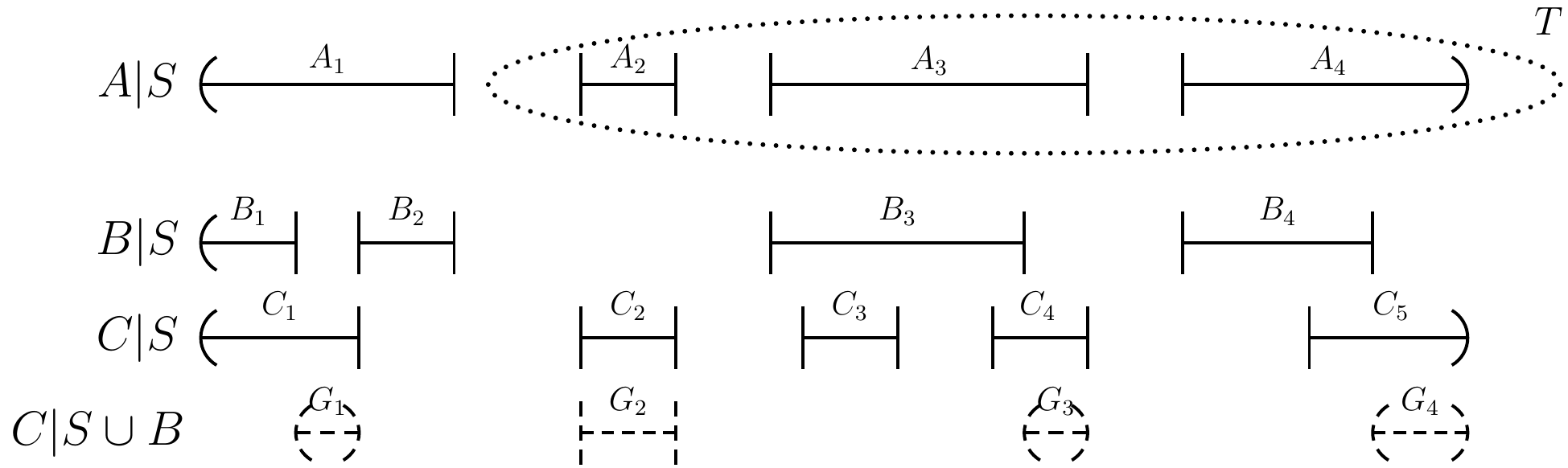}
\caption{\small An example where the set $T$ and the graphs $A|S, B|S,C|S$ are as depicted. Then $B[T]|S = B_3\cup B_4$ and $C[T]|S\cup B = G_2\cup G_3\cup G_4$ with $\Delta(A|S) = \Delta(A[T]|S) = 2$, $\Delta(B[T]|S) = 2$ and $\Delta(C[T]|S\cup B)=1$.}
\label{fig:zoom}
\end{figure}

We will now consider three cases according to whether ($\dag$), ($\gad$) or ($\ddag$) is tight for $\Phi(A|S \cup T)$.  

First, assume $\Phi(A|S \cup T) = \Phi(B|S \cup T) + \Delta(C|S \cup T \cup B)$.  We have
\[
  \Phi(A|S) &\stackrel{(\dag)}{\geq} \Phi(B|S) + \Delta(C|S \cup B))&&\\
  &\ge \Phi(B|S \cup T) + \Delta(B[T]|S) + \Delta(C|S \cup B)
  &&\text{(by induction hypothesis)}\\
  &= \Phi(B|S \cup T) + \Delta(B[T]|S) + \Delta(C[T]|S\cup B) + \Delta(C|S \cup B\cup T) &&\text{(by observation }\ref{obs:DeltaSplit} \text{ for } C|S \cup B\text{)}\\
  &\geq \Phi(A|S \cup T) + \Delta(A[T]|S)&&\text{(by assumption and by eq. (}\ref{eq:DeltaCB} \text{))}
\]

Next, assume $\Phi(A|S \cup T) = \Delta(B|S \cup T) + \Phi(C|S  \cup T \cup B)$. We have
\[
  \Phi(A|S) &\stackrel{(\gad)}{\geq} \Delta(B|S) + \Phi(C|S \cup B) &&\\
  &= \Delta(B|S \cup T) + \Delta(B[T]|S) + \Phi(C|S \cup B)&&\text{(by observation }\ref{obs:DeltaSplit} \text{ for } B|S\text{)}\\  
  &\ge \Delta(B|S \cup T) + \Delta(B[T]|S) + \Phi(C|S \cup B\cup T) + \Delta(C[T]|S\cup B) &&\text{(by induction hypothesis)}\\
  &\geq \Phi(A|S \cup T) + \Delta(A[T]|S)&&\text{(by assumption and by eq. (}\ref{eq:DeltaCB} \text{))}
\]

Finally, assume $\Phi(A|S \cup T) = \frac12\big(\Phi(D|S \cup T) + \Phi(A|S \cup T \cup D) + \Delta(A|S \cup T)\big)$ for some $D \prec A$. We have
\[
\Phi(A|S) &\stackrel{(\ddag)}{\geq} \frac12\Big(\Phi(D|S) + \Phi(A|S \cup D) + \Delta(A|S)\Big)\\
&\ge \frac12\Big(\Phi(D|S\cup T) + \Phi(A|S \cup D\cup T) + \Delta(A|S\cup T) + \Delta(D[T]|S) + \Delta(A[T]|S\cup D) + \Delta(A[T]|S)\Big)
\]
where the last inequality follows from the induction hypothesis on $D|S$ and $A|S \cup D$ (this is where we use the backward induction for sets). It thus suffices to check that 
\[
\Delta(D[T]|S) + \Delta(A[T]|S\cup D) \geq \Delta(A[T]|S).
\]
But this is straightforward as each closed component of $A[T]|S$ either contains at least one closed component of $D[T]|S$, or it does not. In the latter case, it is then clear that it must also be a component of $A[T]|S\cup D$.
\end{proof}

\begin{thm}\label{thm:PhiPk-lb}
For every join-tree $A$ and set $S$, and a component $K$ of $A|S$ of length $k$,
\[  
  \Phi(A|S) &\ge  \log_{c}(\eps\delta k) + \Delta(A|S \cup K)
  &&\text{if }K\text{  is open},\\
  \Phi(A|S) &\ge \log_{c}(\delta k) + \Delta (A|S \cup K)
  &&\text{if }K\text{  is half-open},\\
  \Phi(A|S) &\ge \log_{c}(k) + \Delta (A|S \cup K)
  &&\text{if }K\text{  is closed}.
\]
where $c = \sqrt{5} + 5$, $\delta = (c-3)/c$ and $\eps = 1/2$. 
\end{thm}

\begin{proof}
We prove this by a structural induction on join-trees as well as a backward induction on the set $S$. Note that we may assume without loss of generality that $S\subseteq V(A)$. Assume $A|S$ is non-empty as the statement follows immediately, otherwise. If $A$ is atomic, then as $k=1$ in each case, the statement is trivial. So let us assume that $A$ is non-atomic with $A = \sq{B,C}$ and that the theorem statement holds for any proper sub-join-tree $D\prec A$ and set $T\subseteq \mbb{Z}$ with the given parameter settings of $c,\delta,\eps$. Moreover, we shall also assume so for the given join-tree $A$ and for every $S'$ such that $S\subsetneqq S' \subseteq V_A$. 

Fix a component $K$ of $A|S$. Let $T$ be the union of the vertex sets of all components of $A|S$ excluding $K$. Then note that $\Delta(A[T]|S) = \Delta(A|S\cup K)$. Therefore, by Lemma~\ref{la:zoom}, it suffices to show that 
\[  
  &&&&&&\Phi(A|S\cup T) &\ge  \log_{c}(\eps\delta k)
  &&\text{if }K\text{  is open},\\
  &&&&&&\Phi(A|S\cup T) &\ge \log_{c}(\delta k) 
  &&\text{if }K\text{  is half-open},\\
  &&&&&&\Phi(A|S\cup T) &\ge \log_{c}(k)
  &&\text{if }K\text{  is closed},&&&&&&
\]
where we know that the graph of $A|S\cup T$ is simply $K$. Hence, we may now assume without loss of generality that $A|S$ is connected and has length $k$.

Henceforth, we shall think of $\eps,\delta,$ and $c$ as indeterminates, imposing constraints on them as we move along the proof. We will eventually verify that the parameter settings specified in the theorem statement indeed satisfy these constraints. We shall proceed by considering the three cases one by one: that $A|S$ is (i) open, (ii) half-open, or (iii) closed. The strategy in each case is to suitably apply one of the three rules, namely ($\dag$), ($\gad$), or ($\ddag$) in order to obtain a lower bound on $\Phi(A|S)$. We remark that cases (i) and (ii) are similar in terms of the ideas involved to arrive at the lower bound: in particular, neither uses the ($\ddag$) rule, which is solely used in case (iii). 
Roughly speaking, we shall see that case (i) determines the value of $\eps$, while case (ii) that of $\delta$ and finally, case (iii) that of $c$. 
\vspace{2ex}

\textbf{Case (i): $A|S$ is open.}
Suppose $A|S = (0,k)$. 
Let $A = \sq{B,C}$ and suppose $B|S = B_1\sqcup\cdots\sqcup B_{n(B)}$ (respectively $C|S = C_1\sqcup\cdots\sqcup C_{n(C)}$) where $B_i$'s (respectively $C_i$'s) are the components of $B|S$ (respectively $C|S$), sorted in the increasing order of left end-points. 
Then with the possible exception of $B_1$ and $B_{n(B)}$, every other $B_i$ is a closed interval (similarly for $C$). Also note that for any $i$ and $j$, $B_i \neq C_j$.
We define 
\[
I(B) = \{B_i :\nexists j \text{ such that } C_j \supseteq B_i\}\text{ and } I(C) = \{C_i: \nexists j \text{ such that } B_j \supseteq C_i\}
\]
It follows that $I(B)\cup I(C)$ is a covering of the interval $(0,k)$ and moreover, each of the end-points $0$ and $k$ lie in a \emph{unique} (and half-open) interval among the members of $I(B)\cup I(C)$, which we denote by $I_0$ and $I_k$ respectively. Observe that if we arrange the intervals in $I(B)\cup I(C)$ in the increasing order of left end-points (starting with $I_0$ obviously), then they alternate membership between $I(B)$ and $I(C)$. Finally, we denote by $\mc{C}(A) \coloneqq I(B)\cup I(C)\setminus \{I_0,I_k\}$, the set of all \emph{closed} intervals in the covering $I(B)\cup I(C)$ and define $t(A) \coloneqq |\mc{C}(A)|$. See Figure~\ref{fig:O_ex} for an example that illustrates these notions.

\begin{figure}[H]
\centering
\includegraphics[scale=0.15]{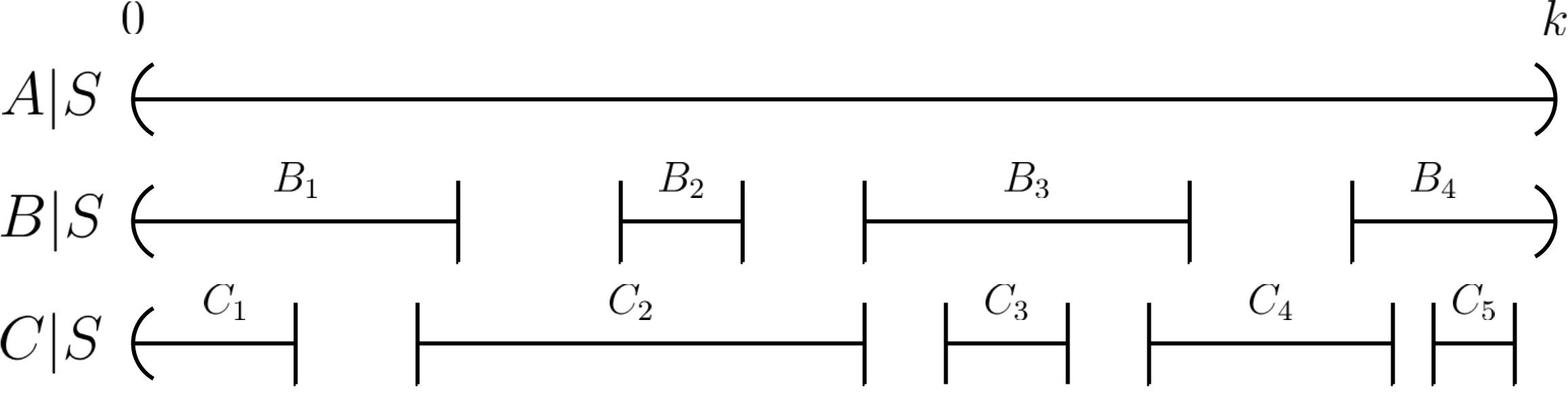}
\caption{\small An example where the graphs $A|S$, $B|S$ and $C|S$ are as depicted above. Here $I(B) = \{B_1,B_3,B_4\}$, $I(C) = \{C_2,C_4\}$, $I_0 = B_1$ and $I_k = B_4$. Thus, $\mc{C}(A) = \{C_2,B_3,C_4\}$ and $t(A) = 3$. Also, notice that $I(B) \cup I(C) = \{B_1,C_2,B_3,C_4,B_4\}$ and thus its members alternate membership between $I(B)$ and $I(C)$ when arranged in increasing order of left end-points.}
\label{fig:O_ex}
\end{figure}

We now consider four sub-cases according to whether $t(A) = 0$, $t(A) = 1$, $t(A) \geq 2$ and $t(A)$ is even, or $t(A) \geq 3$ and $t(A)$ is odd. They all involve very similar ideas except for the $t(A) = 1$ sub-case, which is slightly more subtle and requires the application of ($\gad$) unlike the other sub-cases. The reader should bear in mind that the variables $a,b,x,y$ are taken to hold a distinct meaning in each of these sub-cases.

\paragraph{Open sub-case $t(A) = 0$:}
Either $B$ or $C$ then has a half-open component of size at least $k/2$ and without loss of generality, suppose that is $B$. See Figure~\ref{fig:O_t=0} for an example. 
\begin{figure}[H]
\centering
\includegraphics[scale=0.15]{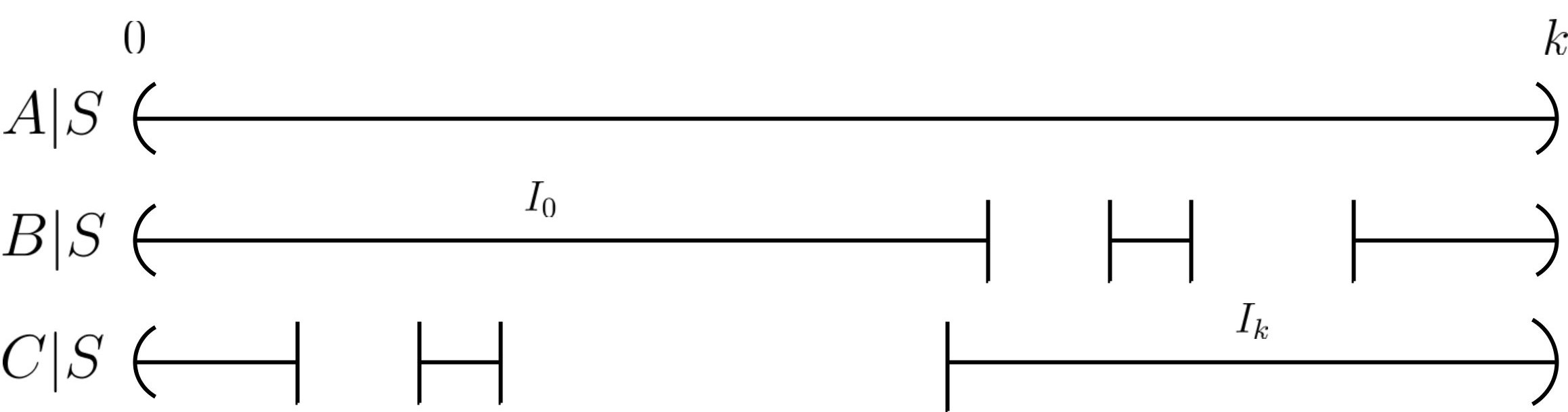}
\caption{\small An example for $t(A) = 0$ where $I_0 \in I(B)$ and $I_k\in I(C)$, with $|I_0| \geq k/2$.  Note that both $B|S$ and $C|S$ each \emph{may} have additional components, such as the unlabelled intervals depicted in this example.}
\label{fig:O_t=0}
\end{figure}

We have
\[
  \Phi(A|S) &\stackrel{(\dag)}{\geq} \Phi(B|S) + \Delta(C|S\cup B)\ge \Phi(B|S)\ge \log_{c}(\delta k/2)\geq \log_c(\eps \delta k)
\] 
as long as $\eps \leq 1/2$. Note that this sub-case is tight for our setting of $\eps$ but works for any $c>1, 0 < \delta \leq 1$.

\paragraph{Open sub-case $t(A) = 1$:}
Both $I_0$ and $I_k$ belong to the same collection i.e., either $I(B)$ or $I(C)$. Without loss of generality, assume that it is the former. Then we know that $|I(C)| = 1$, so let $I = [i,j]$ be that closed interval in $I(C)$, for some $0<i<j<k$. Note that $\max\{i,k-j\} \geq (k - |I|)/2$ so assume without loss of generality that $i\geq (k - |I|)/2$ (in words, that there is more `space' to the left of $I$ than to its right). 
Next, let $J_1,\ldots,J_s$ be the closed intervals in $C|S$ (labelled in the increasing order of left end-points) that appear `before' $I$ i.e., their right end-points are (strictly) less than $i$ (of course it may be the case that $s=0$, when there is no such interval) and let $y\coloneqq \frac{1}{k}(\max_{J\in\{I,J_1\ldots,J_s\}}|J|)$. 
Let $J_0$ be the component (if it exists) in $C|S$ containing $0$, and let $b\coloneqq |J_0|/k$ (define $b=0$ if it does not exist).

Let $G_1\sqcup \cdots\sqcup G_{s+1} = (0,i]\ |\ (J_0\cup \cdots\cup J_s \cup S) \subseteq B|(S\cup C)$ be the `gap' intervals (labelled in the increasing order of left end-points) as shown in Figure~\ref{fig:O_t=1}. Now with the possible exception of $G_1$ which may be half-open or open (depending or whether $J_0$ is empty or not), all other $G_i$ are open. Let $gk$ be the length of the longest gap interval. Then it follows that $(s+1)g + b + sy \geq \frac{i}{k}\geq \frac{1-y}{2}.$

\begin{figure}[H]
\centering
\includegraphics[scale=0.15]{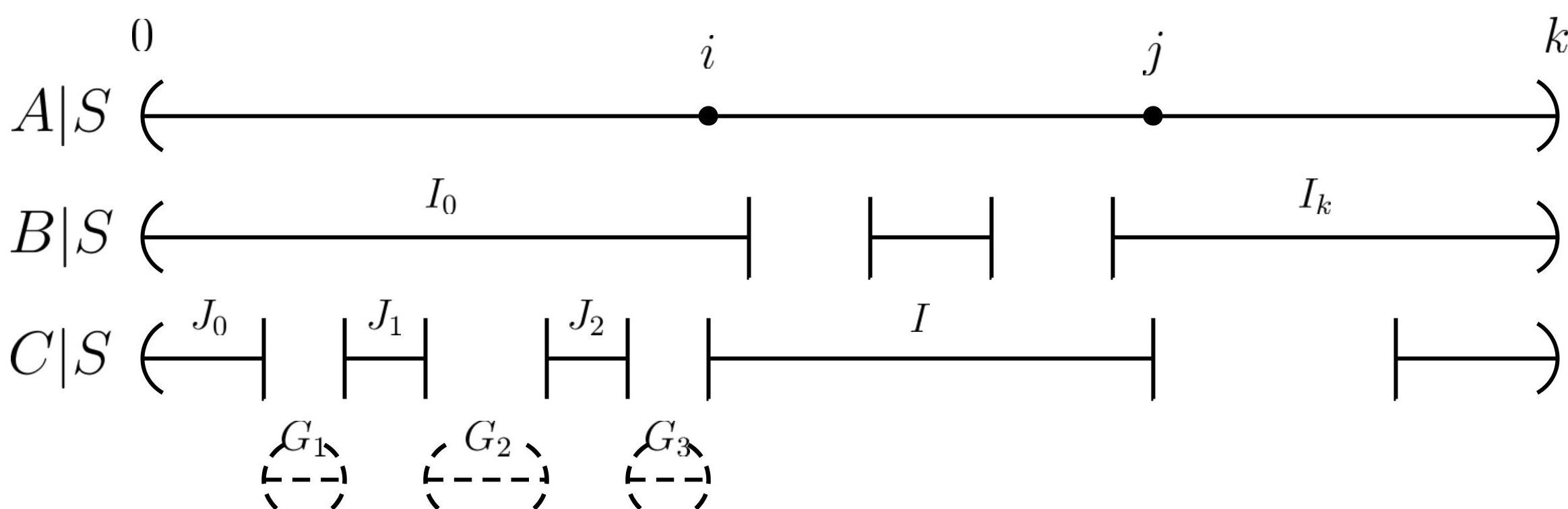}
\caption{\small An example for $t(A) = 1$ where $s = 2$. The gap intervals $G_i$ are shown in dashed lines.  Both $B|S$ and $C|S$ each \emph{may} have additional components, such as the unlabelled intervals depicted in this example.}
\label{fig:O_t=1}
\end{figure}

Now if either of the following occurs, we are immediately done.

\begin{itemize}
\item If $a \geq \eps$, then simply use the induction hypothesis on $B$. We have 
\[
\Phi(A|S)\stackrel{(\dag)}{\geq} \Phi(B|S) + \Delta(C|S\cup B) \geq \Phi(B|S) \geq \log_c(\delta a k) \geq \log_c(\eps \delta k).
\]
\item If $b\geq \eps/c^{s+1}$, then as $C|S$ contains at least $s+1$ other closed components, we have 
\[
\Phi(A|S)\stackrel{(\dag)}{\geq} \Phi(C|S) + \Delta(B|S\cup C) \geq \Phi(C|S) \geq \log_c (\delta bk) + (s+1) \geq \log_c(\eps \delta k).
\]
\item If $y\geq (\eps\delta )/c^s$, then we have 
\[
\Phi(A|S)\stackrel{(\dag)}{\geq} \Phi(C|S) + \Delta(B|S\cup C) \geq \Phi(C|S)\geq  \log_c (yk) + s \geq \log_c(\eps \delta k).
\]
\end{itemize}

So we may assume that $a < \eps$, $b< \eps/c^{s+1}$, and $y< (\eps\delta )/c^s$, which together imply that 
\[
g > \frac{1}{s+1}\left(\frac{1}{2}-\frac{\eps}{c^{s+1}}-\left(s+\frac12\right)\cdot \frac{\eps\delta}{c^s}\right).
\]
We have (this time by ($\gad$))
\[
\Phi(A|S)
  &\stackrel{(\smallgad)}{\geq} \Delta(C|S) + \Phi(B|S\cup C)\ge s+1 + \log_c\left(\frac{\eps\delta k}{s+1}\left(\frac{1}{2}-\frac{\eps}{c^{s+1}}-\left(s+\frac12\right)\cdot \frac{\eps\delta}{c^s}\right)\right)
\]
which is at least $\log_c(\eps\delta k)$ as long as for all $s\geq 0$,
\[
\frac{c^{s+1}-2\eps - 2(s+\frac12)\cdot\eps\delta c}{2(s+1)}\geq 1
\]
which upon plugging in $\eps \leq 1/2$, reduces to showing for all $s\geq 0$ that
\[
c^{s+1}\geq 2s+3 + \left(s+\frac12\right)\cdot\delta c
\]
which is indeed true for $c>6$, $0<\delta \leq 1$. We note that this sub-case is {not} tight for our parameter settings $c = \sqrt{5} + 5$, $\delta = \frac{c-3}{c}$.

\paragraph{Open sub-case $t(A) \geq 2$ and $t(A)$ is even:} If $s = t(A)/2$, then both $B|S$ and $C|S$ contain (exactly) one half-open interval among $I_0,I_k$ and $s\geq 1$ closed intervals each. Similar to the previous case, we define $a\coloneqq \frac{1}{k}(\max\{|I_0|,|I_k|\})$ and $x\coloneqq \frac{1}{k}(\max_{I\in \mc{C}(A)} |I|)$. It follows that $2a + 2sx\geq 1$. 
See Figure~\ref{fig:O_t_even} for an example. 

\begin{figure}[H]
\centering
\includegraphics[scale=0.15]{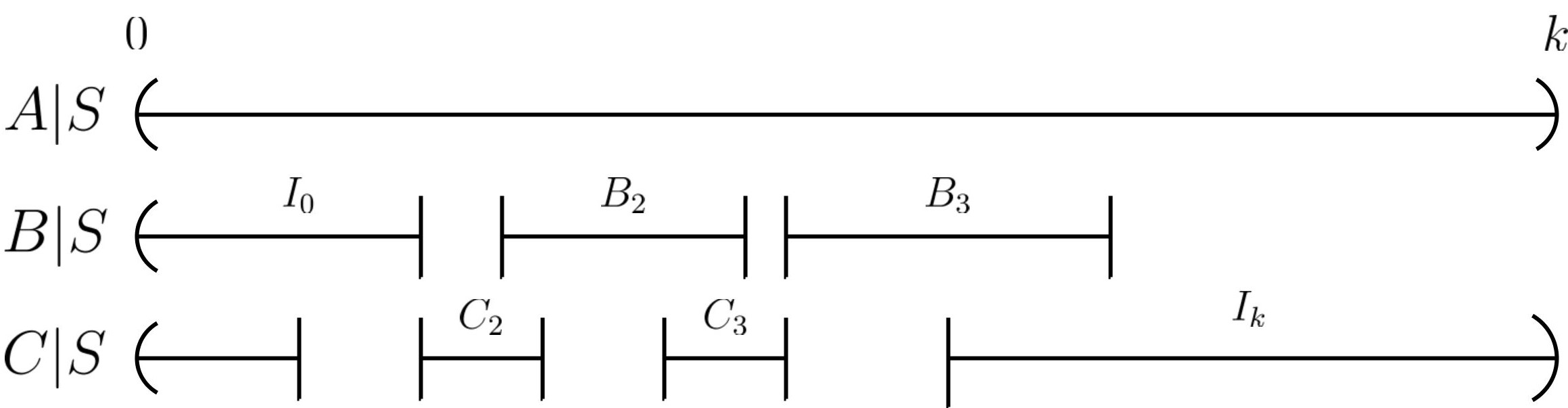}
\caption{\small An example where $t(A) = 4$. Here, $I(B) = \{ I_0, B_2,B_3\}$ and $I(C) = \{C_2,C_3,I_k\}.$ $a$ is therefore defined to be $|I_k|/k$ and $x$ is defined to be $|B_3|/k$.}
\label{fig:O_t_even}
\end{figure}

If $a\geq\eps/c^s$, then as one of $B|S$ or $C|S$ has a half-open interval of length $ak$ along with at least $s$ closed intervals, we have 
\[
\Phi(A|S) \stackrel{(\dag)}{\geq} \max\{\Phi(B|S),\Phi(C|S)\} \geq \log_c (\delta ak) + s \geq \log_c(\eps \delta k)
\]
and we are done. So assume that $a<\eps/c^s$, implying $x>\frac{1}{s}(\frac{1}{2}-\frac{\eps}{c^s})$. Therefore, as one of $B|S$ or $C|S$ has a closed interval of length $xk$ with at least $s-1$ other closed intervals, we have
\[
\Phi(A|S)\stackrel{(\dag)}{\geq}\max\{\Phi(B|S),\Phi(C|S)\} \geq \log_c (xk) + (s-1) >\log_c\left(\frac{c^{s-1}k}{s}\left(\frac{1}{2}-\frac{\eps}{c^s}\right)\right)
\]
and it is enough to check that this last expression is at least $\log_c(\eps \delta k)$, which reduces to showing that for all $s\geq 1$,
\[
c^s-2\eps\geq 2c\eps\delta s
\]
which is clearly true when $c>4$, $0<\delta\leq (c-1)/c$ and $\eps\leq 1/2$.
Therefore, this sub-case is also {not} tight for our parameter settings $c = \sqrt{5} + 5$, $\delta = \frac{c-3}{c}$.

\paragraph{Open sub-case $t(A) \geq 3$ and $t(A)$ is odd:} Both $I_0$ and $I_k$ must then lie in $I(B)$, without loss of generality. Then if $s = (t(A) + 1)/2$, then $I(C)$ has $s\geq 2$ closed components while $I(B)$ has $s-1$. Define $a\coloneqq \frac{1}{k}(\max\{|I_0|,|I_k|\})$, $x\coloneqq \frac{1}{k}(\max_{I\in I(B)\setminus \{I_0,I_k\}} |I|)$, and $y\coloneqq \frac{1}{k}(
\max_{I\in I(C)} |I| )$. It follows that $2a + (s-1) x + sy \geq 1$.

Refer back to the example described in Figure~\ref{fig:O_ex}. For those particular graphs $A|S$, $B|S$, and $C|S$, we would have $a = |B_1|/k$, $x = |B_3|/k$, and $y = |C_2|/k$.

Now if either of the following occurs, we are immediately done.

\begin{itemize}
\item If $a\geq \eps/c^{s-1}$, 
simply use the induction hypothesis on $B$: one of $I_0$ or $I_k$ is a half-open component of length at least $ak$ and $B|S$ has $s-1$ closed components. We have
\[
\Phi(A|S) \stackrel{(\dag)}{\geq} \Phi(B|S)+\Delta(C|S\cup B)\geq \Phi(B|S)\geq \log_c (\delta ak) + (s-1) \geq \log_c (\eps \delta k).
\]
\item Similarly if $x\geq(\eps \delta)/c^{s-2}$, then there is closed component of length at least $xk$ in $B|S$ along with $s-2\geq 0$ other closed components. We have
 \[
\Phi(A|S) \stackrel{(\dag)}{\geq} \Phi(B|S)+\Delta(C|S\cup B)\geq  \Phi(B|S)\geq \log_c (xk) + (s-2) \geq \log_c (\eps \delta k).
\]
\end{itemize}
  
Hence, we may assume that $a<\eps/c^{s-1}$ and $x<(\eps \delta)/c^{s-2}$, which together imply that \[y > \frac{1}{s}\left(1-\frac{2\eps}{c^{s-1}} - \frac{\eps\delta(s-1)}{c^{s-2}}\right).\]
Now $C|S$ has a closed component of length at least $yk$ along with $s-1$ other components. We have
\[\Phi(A|S) \stackrel{(\dag)}{\geq} \Phi (C|S) \geq \log_c(yk) + (s-1) > \log_c\left(\frac{c^{s-1}k}{s}\left(1-\frac{2\eps}{c^{s-1}} - \frac{\eps\delta(s-1)}{c^{s-2}}\right)\right)\]
and it is enough to show that this last expression is at least $\log_c(\eps\delta k)$, which in turn reduces to checking that for all $s\geq 2$,

\[
\frac{c^{s-1}}{\eps\delta s} - \frac{2}{\delta s} - \frac{c(s-1)}{s} \geq 1
\]
which is straightforward to verify for any $4<c, 0<\delta\leq 1$, and  $\eps \leq 1/2$. In particular, this sub-case is also {not} tight for our parameter settings $c = \sqrt{5} + 5$, $\delta = \frac{c-3}{c}$.

\vspace{1ex}
\noindent\textit{Remark:}
The above inequality is {not} true for $s=1$ irrespective of the choice of $c$, which is precisely why we have a separate argument for the sub-case $t(A) = 1$. 
\vspace{2ex}

\textbf{Case (ii): $A|S$ is half-open.}
Suppose without loss of generality that $A|S = [0,k)$. Again, let $A = \sq{B,C}$ and suppose $B|S = B_1\sqcup\cdots\sqcup B_{n(B)}$ (respectively $C|S = C_1\sqcup\cdots\sqcup C_{n(C)}$) where $B_i$s (respectively $C_i$s) are the connected components of $B|S$ (respectively $C|S$). Then with the possible exception of $B_{n(B)}$, every other $B_i$ is a closed interval (similarly for $C$). Also note that for any $i$ and $j$, $B_i \neq C_j$.
We define 
\[
I(B) = \{B_i :\nexists j \text{ such that } C_j \supseteq B_i\}\text{ and } I(C) = \{C_i: \nexists j \text{ such that } B_j \supseteq C_i\}
\]
It follows that $I(B)\cup I(C)$ is a covering of the interval $[0,k)$ and moreover, each of the end-points $0$ and $k$ lies in a \emph{unique} (closed and half-open, respectively) interval among the members of $I(B)\cup I(C)$, which we denote by $I_0$ and $I_k$ respectively. Again, observe that if we arrange the intervals in $I(B)\cup I(C)$ in the increasing order of left end-points, then they alternate membership between $I(B)$ and $I(C)$. Finally, we denote by $\mc{C}(A) \coloneqq I(B)\cup I(C)\setminus \{I_k\}$, the set of all \emph{closed} intervals in the covering $I(B)\cup I(C)$ and define $t(A) \coloneqq |\mc{C}(A)|$. 
See Figure~\ref{fig:HO_ex} for an example that illustrates these notions.

\begin{figure}[H]
\centering
\includegraphics[scale=0.15]{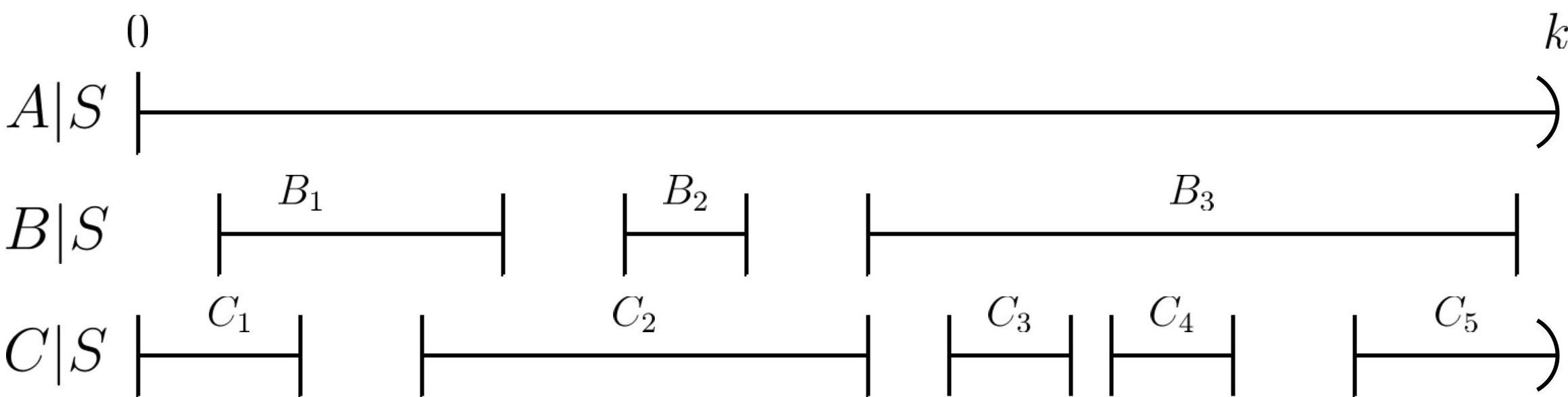}
\caption{\small An example where the graphs $A|S$, $B|S$ and $C|S$ are as depicted above. Here $I(B) = \{B_1,B_2,B_3\}$, $I(C) = \{C_1,C_2,C_5\}$, $I_0 = C_1$ and $I_k = C_5$. Thus, $\mc{C}(A) = \{C_1,B_1,C_2,B_3\}$ and $t(A) = 4$. Also, notice that $I(B) \cup I(C) = \{C_1,B_1,C_2,B_3,C_5\}$ and thus its members alternate membership between $I(B)$ and $I(C)$ when arranged in increasing order of left end-points.}
\label{fig:HO_ex}
\end{figure}

We again consider four sub-cases according to whether $t(A) = 1$, $t(A) = 2$, $t(A) \geq 3$ and $t(A)$ is odd, or $t(A) \geq 4$ and $t(A)$ is even. The first two sub-cases are slightly more subtle and require the application of ($\gad$) unlike the rest. Again, the variables $a,b,x,y$ are taken to hold a distinct meaning in each of these sub-cases.

\paragraph{Half-open sub-case $t(A) = 1$:}
This means that $I_0$ is the unique interval in $\mc{C}(A)$, also implying that $I_k$ and $I_0$ overlap. So suppose without loss of generality that $I(B) = \{I_k\}$, $I(C) = \{I_0\}$, and let $J_1,\ldots J_s$ be the closed intervals that appear in $C|S$ `after' $I_0$. Further, let $J$ be the half-open interval in $C|S$ that contains $k$ and let $b \coloneqq |J|/k$ where $b\coloneqq 0$ if such an interval does not exist. We define $y \coloneqq \frac{1}{k}(\max_{I\in\{|I_0,J_1,\ldots,J_s\}} |I|)$. See Figure~\ref{fig:HO_t=1} for an example.

\begin{figure}[H]
\centering
\includegraphics[scale=0.15]{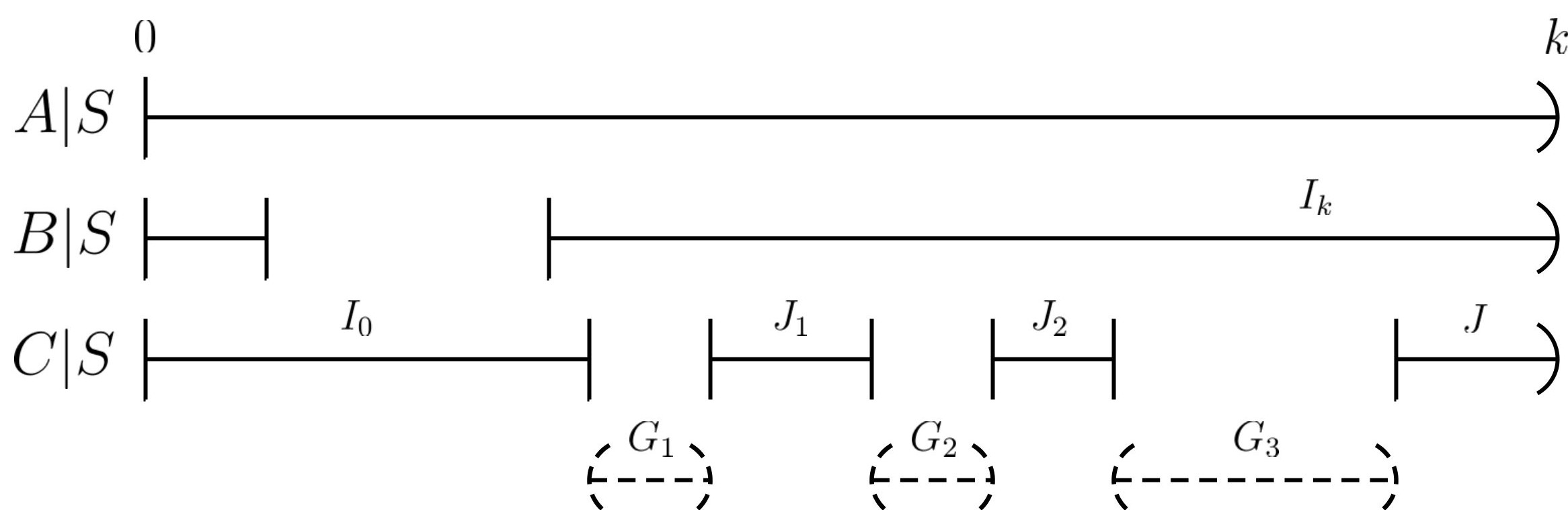}
\caption{\small An example for $t(A) = 1$ where $s = 2$. The gap intervals $G_i$ are shown in dashed lines.  $B|S$ \emph{may} have additional components, such as the unlabelled interval depicted in this example.}
\label{fig:HO_t=1}
\end{figure}

Now if $y\geq \delta/c^s$, then
\[
\Phi(A|S)\stackrel{(\dag)}{\geq} \Phi (C|S) +\Delta(B|S\cup C) \geq \Phi(C|S)\geq \log_c(yk) + s \geq \log_c(\delta k)
\]
and we are done. So assume that $y<\delta/c^s$. Further, if $b\geq 1/c^{s+1}$, then 
\[
\Phi(A|S)\stackrel{(\dag)}{\geq} \Phi (C|S) +\Delta(B|S\cup C) \geq \Phi(C|S)\geq \log_c(bk) + s + 1 \geq \log_c(\delta k)
\]
and so we assume that $b<1/c^{s+1}$. Next, suppose $G_1\sqcup \cdots\sqcup G_{s+1} = I_k\ |\ (I_0\cup J_1\cup \cdots\cup J_s \cup S) \subseteq B|(C\cup S)$ are the `gap' intervals as shown in Figure~\ref{fig:HO_t=1}. Let $gk$ be the length of the longest gap interval. Then it follows that $(s+1)g  + sy + b \geq 1-y$. We have (this time by ($\gad$))
\[
\Phi(A|S) \stackrel{(\smallgad)}{\geq} \Delta(C|S) + \Phi(B|S\cup C) \geq s+1 +\log_c(\eps\delta gk)\geq s+1 + \log_c\left(\frac{\eps\delta k}{s+1}\left(1-\frac{\delta(s+1)}{c^s}-\frac{1}{c^{s+1}}\right)\right). 
\]
The task of showing that this expression is at least $\log_c (\delta k)$ reduces to showing that for all $s\geq 0$,
\[
\eps\left(\frac{c^{s+1}-1}{s+1}-\delta c\right)\geq 1
\]
which is clearly true when $c>4$, $\delta \leq (c-3)/c$, and $\eps = 1/2$.
We note that this sub-case is indeed tight for our parameter settings $\delta = \frac{c-3}{c}$ and $\eps = 1/2$, however it works for any $c>4$.

\paragraph{Half-Open sub-case $t(A) = 2$:} Let $I(C) = \{I_0, I_k\}$, $I(B) = \{J\}$ and $a \coloneqq |I_k| /k$. 
If $a\geq1/c$, then we have ($\dag$)
\[
\Phi(A|S)\geq\Phi(C|S) \geq \log_c (\delta ak) + 1 \geq \log_c( \delta k)
\]
and we are done.
So we may assume that $a< 1/c$. Let $J_1,\ldots,J_s$ ($s$ may be zero) be the closed intervals in $C|S$ that appear `before' $J$ as shown in Figure~\ref{fig:HO_t=2} and suppose $y\coloneqq \frac{1}{k}(\max_{I\in\{J,J_1,\ldots,J_s\}} |I|) $.
If $y\geq \delta/c^s$, then
\[
\Phi(A|S)\stackrel{(\dag)}{\geq} \Phi (B|S) \geq \log_c(yk) + s \geq \log_c(\delta k)
\]
and we are done. So assume that $y<\delta/c^s$. Now let $i$ be the left end-point of $J$ and suppose $G_1\sqcup \cdots\sqcup G_{s+1} = (0,i]\ |\ (J\cup J_1\cup \cdots\cup J_s \cup S) \subseteq C|(S\cup B)$ are the `gap' intervals as shown in the Figure~\ref{fig:HO_t=2}. Let $gk$ be the length of the longest gap interval. Then it follows that $(s+1)g  + sy \geq \frac{i}{k}\geq 1-a-y \geq 1 - \frac{\delta}{c^s}-\frac{1}{c}.$ 

\begin{figure}[H]
\centering
\includegraphics[scale=0.15]{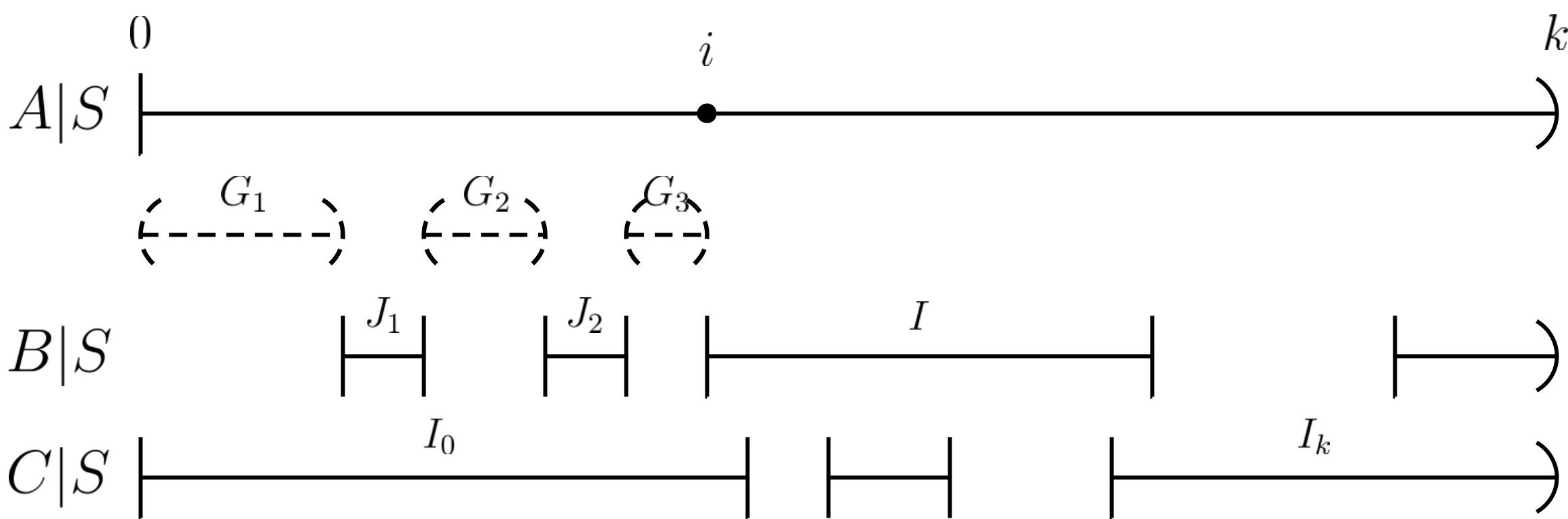}
\caption{\small An example for $t(A) = 2$ where $s = 2$. The gap intervals $G_i$ are shown in dashed lines.  Both $B|S$ and $C|S$ each \emph{may} have additional components, such as the unlabelled intervals depicted in this example.}
\label{fig:HO_t=2}
\end{figure}

We have (by ($\gad$))
\[
\Phi(A|S)\stackrel{(\smallgad)}{\geq} \Delta(B|S) + \Phi(C|S\cup B) \geq s+1 +\log_c(\eps\delta gk)\geq s+1 + \log_c\left(\frac{\eps\delta k}{s+1}\left(1-\frac{\delta(s+1)}{c^s}-\frac{1}{c}\right)\right) 
\]
The task of showing that this expression is at least $\log_c (\delta k)$ reduces to showing that for all $s\geq 0$,
\[
\eps\left(\frac{c^{s+1}-c^s}{s+1}-\delta c\right)\geq 1
\]
which is clearly true when $c>4$, $\delta \leq (c-3)/c$, and $\eps = 1/2$.
This sub-case is also tight for our parameter settings $\delta = \frac{c-3}{c}$ and $\eps = 1/2$, however it works for any $c>4$.

\paragraph{Half-open sub-case $t(A) \geq 3$ and $t(A)$ is odd:} Suppose $I_k$ is in $I(B)$, without loss of generality. Then if $s = (t(A) + 1)/2$, then $I(C)$ has $s\geq 2$ closed components while $I(B)$ has $s-1$. Let $a\coloneqq |I_k| /k$, $x  \coloneqq \frac{1}{k}( \max_{I\in I(B)\setminus \{I_k\}} |I|),$ and $
y \coloneqq \frac{1}{k}( \max_{I\in I(C)} |I|)$ (see Figure \ref{fig:HO_t_odd} for an example). It follows that $a + (s-1) x + sy \geq 1$.

\begin{figure}[H]
\centering
\includegraphics[scale=0.15]{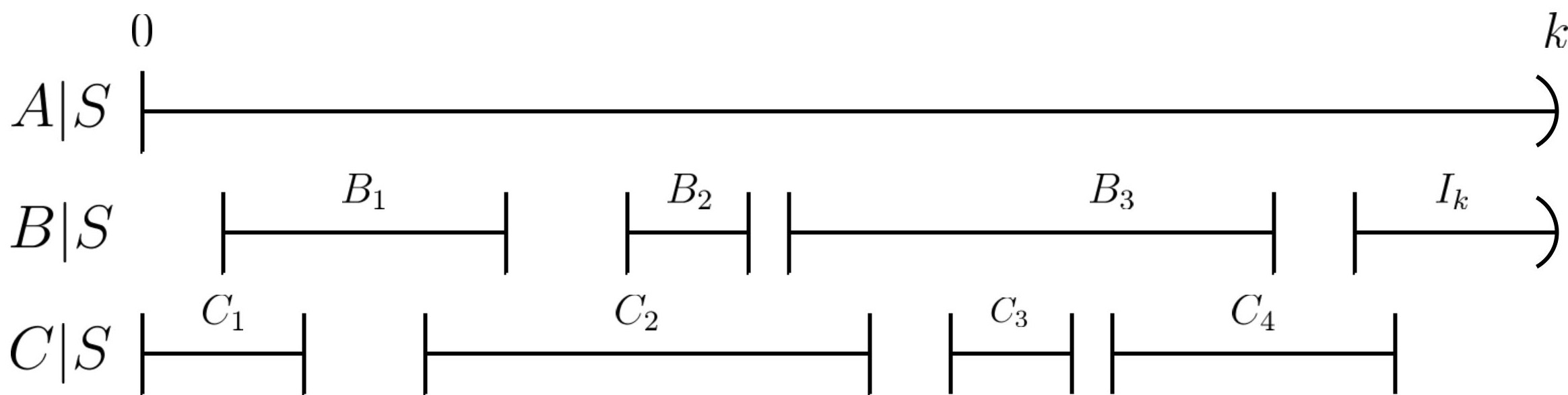}
\caption{\small An example for $t(A) = 5$. Here, we would have $x = |B_3|/k$ and $y = |C_2|/k$.}
\label{fig:HO_t_odd}
\end{figure}

Now if either of the following occurs, we are immediately done.

\begin{itemize}
\item If $a\geq 1/c^{s-1}$, then 
\[
\Phi(A|S) \stackrel{(\dag)}{\geq} \Phi(B|S)\geq \log_c (\delta ak) + (s-1) \geq \log_c ( \delta k).
\]
\item Similarly if $x\geq \delta/c^{s-2}$, then there is closed component of length at least $xk$ in $B|S$ along with $s-2\geq 0$ other closed components. We have 
 \[
\Phi(A|S) \stackrel{(\dag)}{\geq} \Phi(B|S)\geq \log_c (xk) + (s-2) \geq \log_c ( \delta k).
\]
\end{itemize}
  
Hence, we may assume that $a<1/c^{s-1}$ and $x< \delta/c^{s-2}$, which together imply that \[y > \frac{1}{s}\left(1-\frac{1}{c^{s-1}} - \frac{\delta(s-1)}{c^{s-2}}\right).\] Now $C|S$ has a closed component of length at least $yk$ along with $s-1$ other components. Thus, 
\[\Phi(A|S) \stackrel{(\dag)}{\geq} \Phi (C|S) \geq \log_c(yk) + (s-1) > \log_c\left(\frac{c^{s-1}k}{s}\left(1-\frac{1}{c^{s-1}} - \frac{\delta(s-1)}{c^{s-2}}\right)\right)\]
and it is enough to show that this last expression is at least $\log_c(\delta k)$, which in turn reduces to checking that for all $s\geq 2$,

\[
\frac{c^{s-1}}{\delta s} - \frac{1}{\delta s} - \frac{c(s-1)}{s} \geq 1
\]
which is straightforward to verify for $4<c$ and $0<\delta\leq (c-1)/(c+2)$. In particular, this is true for our parameter settings $c = \sqrt{5} + 5$, $\delta = \frac{c-3}{c}$ and it follows that this sub-case is {not tight}.

\paragraph{Half-open sub-case $t(A) \geq 4$ and $t(A)$ is even:} Suppose $a \coloneqq |I_k|/k$ without loss of generality, $I_k\in I(C)$. Then if $s = t(A)/2$, then both $B|S$ and $C|S$ contain at least $s\geq 2$ closed intervals each. We define $x\coloneqq \frac{1}{k}(\max_{I\in \mc{C}(A)} |I| )$. It follows that $a + 2sx\geq 1$. 

Refer back to the example depicted in Figure~\ref{fig:HO_ex}. For that particular example, we would then have $a = |C_5|/k$ and $x = |B_3|/k$, as $B_3$ is the longest interval in $\mc{C}(A)$. 

If $a\geq1/c^s$, then 
\[
\Phi(A|S)\stackrel{(\dag)}{\geq}\Phi(C|S) \geq \log_c (\delta ak) + s \geq \log_c( \delta k)
\]
and we are done. So assume that $a<1/c^s$, implying $x>\frac{1}{2s}(1-\frac{1}{c^s})$. Therefore, 
\[
\Phi(A|S)\stackrel{(\dag)}{\geq}
\max\{\Phi(B|S),\Phi(C|S)\} \geq \log_c (xk) + (s-1) >\log_c\left(\frac{c^{s-1}k}{2s}\left(1-\frac{1}{c^s}\right)\right)
\]
and it is enough to check that this last expression is at least $\log_c( \delta k)$, which reduces to showing that for all $s\geq 2$,
\[
c^s-1\geq 2c\delta s
\]
which is clearly true when $c>4$ and $0<\delta\leq 1$. This sub-case is also {not} tight for our parameter settings $c = \sqrt{5} + 5$, $\delta = \frac{c-3}{c}$.
\vspace{2ex}

\textbf{Case (iii): $A|S$ is closed.}
Finally, suppose $A|S = [0,k]$. 
First, consider the sub-case that there exists $D \prec A$ such that $D = [0,j]$ or $D = [j,k]$ with 
\[
  \frac12(1-\sqrt x)k \le j \le \frac12(1+\sqrt x)k
\]
for $x = \frac{1}{(\sqrt{5}+2)^2}$, as shown in Figure~\ref{fig:C_int}. Then 
$
  j(k-j) \ge \frac{1-x}{4}k^2.
$
We have (this time by ($\ddag$))
\[
  \Phi(A|S)
  &\stackrel{(\ddag)}{\geq} \frac12\Big(\Phi(D|S) + \Phi(A|S\cup D) + \Delta(A|S)\Big)\\
  &\ge \frac12\Big(\log_{c}(k-j) + \log_{c}(\delta j) + 1\Big) \quad (\text{induction hypothesis on } A|S\cup D)\\
  &= \frac12\log_{c}(c\delta j(k-j))
  \ge \frac12\log_{c}\left(\frac{c\delta(1-x)}{4}k^2\right)
  = \log_c(k) + \frac12\log_c\left(\frac{c\delta(1-x)}{4}\right).
\]
Note that for our choice of $x$, and the given parameter settings $c = \sqrt{5} + 5$ and $\delta = \frac{c-3}{c}$, $\tsfrac{c\delta(1-x)}{4} = 1$, 
thereby establishing the claim if such a $D\prec A$ exists. We also note that our parameter settings are indeed tight in this case.

\begin{figure}[H]
\centering
\includegraphics[scale=0.15]{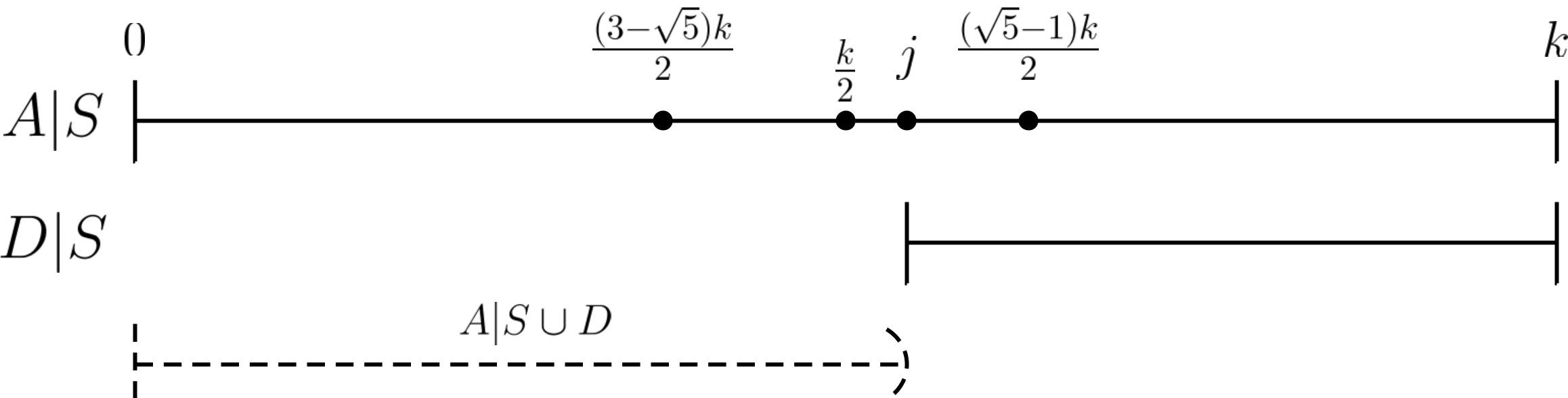}
\caption{\small An example where $A|S$ is closed and there exists a connected $D\prec A$ containing one of the end-points of $A|S$ and having `intermediate' length. We apply the ($\ddag$) rule in this situation.}
\label{fig:C_int}
\end{figure}

Therefore, now assume that no such $D$ exists.

Furthermore, we may assume that if there exists $D \prec A$ such that $D$ contains a component of length at least $\frac{1}{c}k$, then $D$ is connected (call this assumption ($\ast$)). Because otherwise, we have the following:
\[  
  \Phi(A|S)
  \stackrel{(\dag)}{\geq}
  \Phi(D|S)
  \ge
  1 + \log_c(k/c)
  = \log_c(k)
\]
Note that $\frac{1}{c} < \frac12(1-\sqrt{x})$.

It follows from the application of ($\ddag$) above and ($\ast$) that the only case left now is when there exists a (connected) $D \prec A$ and a child $E$ of $D$ such that $D = [0,j]$ and the component $J_0 = [0,i]$ of $E$ containing $0$ satisfy 
\[
  i \le \frac12\left(1-\sqrt{x}\right)k \quad\text{ and }\quad \frac12\left(1+\sqrt{x}\right)k \le j. 
\]
Let $E$ have $s$ other components $J_1,\ldots,J_s$ apart from $J_0$. We now consider sub-cases according to whether $s$ is zero or not.

\paragraph{Closed sub-case I: $s=0$.} It follows that $E$ is connected. Then note that $ D|S \cup E = G = (i,j]$ is as shown in Figure~\ref{fig:C_s=0}, and we have
\[
  \Phi(A|S)
  \stackrel{(\dag)}{\geq} 
  \Phi(D|S)
  &\stackrel{(\smallgad)}{\geq}
  \Delta(E|S) + \Phi(D|S \cup E)
  \ge
  1 + \log_c(\ts\delta(j-i))
  \ge
  1 + \log_c(\delta \sqrt x k)
  =
  \log_c(k) + \log_c(c\delta\sqrt x).
\]
So it only remains to check that $x \ge 1/(c\delta)^2$, which is indeed true for our parameter settings. We note that this sub-case is  tight for our parameter settings.

\begin{figure}[H]
\centering
\includegraphics[scale=0.15]{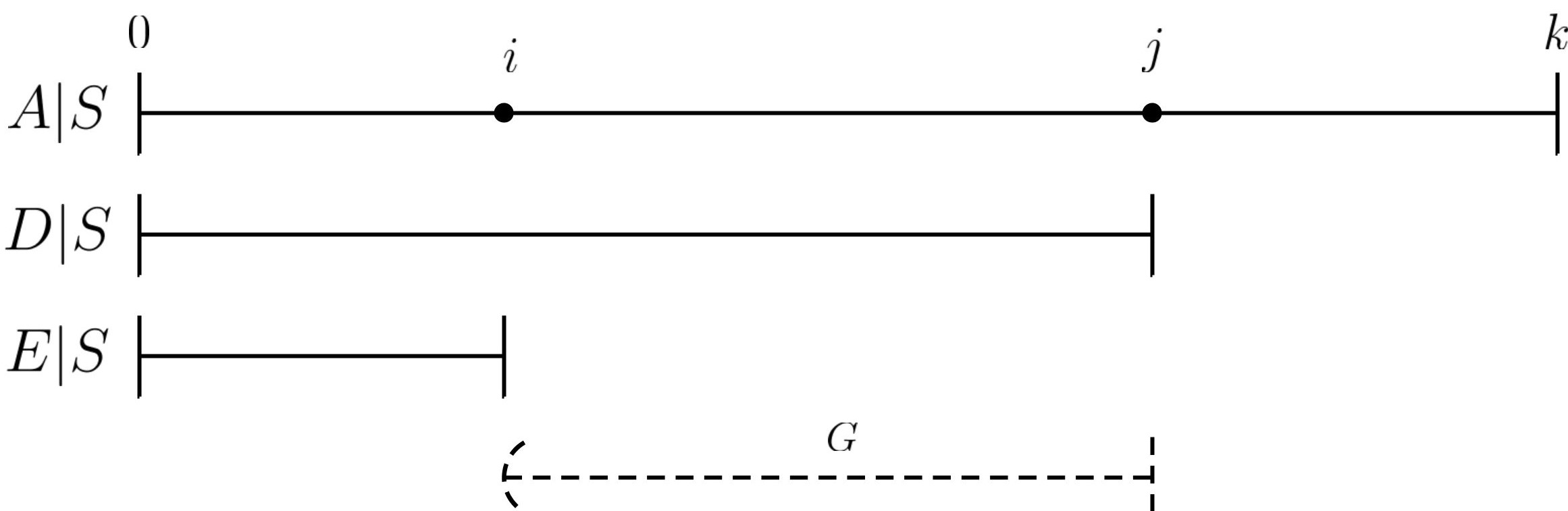}
\caption{\small An example where $A|S$ is closed and $E\prec D$ is connected.}
\label{fig:C_s=0}
\end{figure}

\paragraph{Closed sub-case II: $s\geq 1$.} We still apply ($\gad$) as in the previous sub-case, but the analysis is different. As $E$ has $s+1$ components in all, we may assume that each $J_i$ has length at most $\frac{1}{c^{s}}k$ as otherwise, we immediately obtain $\Phi(A|S) \geq \Phi(D|S) \geq \log_c(k)$ from the ($\dag$) rule. As a consequence, if  $G_1\sqcup \cdots \sqcup G_{s+1} = D|S\cup E$ are the `gap' intervals as shown in Figure~\ref{fig:C_s>0}, some $G_i$ must be an open (or half-open) interval of length at least $\frac{1}{s+1} (j - \frac{s+1}{c^{s}})$. 

\begin{figure}[H]
\centering
\includegraphics[scale=0.15]{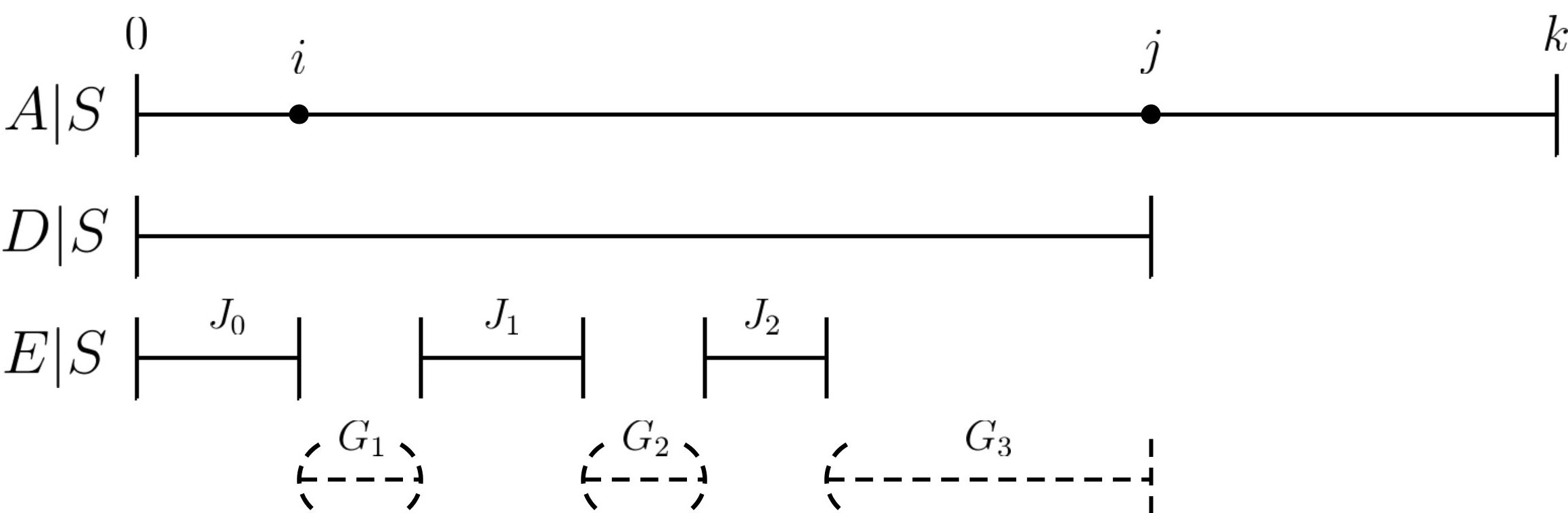}
\caption{\small An example where $A|S$ is closed and $s = 2$.}
\label{fig:C_s>0}
\end{figure}

We have 
\[
\Phi(A|S)
  \stackrel{(\dag)}{\geq} 
  \Phi(D|S)
  &\stackrel{(\smallgad)}{\geq} 
  \Delta(E|S) + \Phi(D|S \cup E)
  \ge
  s+1 + \log_c\left(\frac{\eps\delta }{s+1} \left(j - \frac{s+1}{c^{s}}k\right)\right)\\  
\]
which is at least $
  \log_c(k)$ as long as for all $s\geq 1$,
\[
\frac{\eps\delta c^{s+1}}{s+1} \left(\frac{j}{k} - \frac{s+1}{c^{s}}\right) - 1 \geq 0
\]
By assumption, $\frac{j}{k}\geq \frac12\left(1+\sqrt{x}\right) = \frac{\sqrt{5}-1}{2}$. Thus, we plug in $\eps = \frac12$ and $c\delta = \sqrt{5} + 2$ and see that it is enough to verify that for all $s\geq 1$,
\[
\frac{\sqrt{5}+2}{2(s+1)}\left(\frac{c^s(\sqrt{5}-1)}{2} - (s+1)\right) - 1 \geq 0
\]
which is easily seen to be true for $c = \sqrt{5} + 5$. 
\end{proof}

Lemma~\ref{la:PhiPk2} and hence Theorem~\ref{thm:Pklower} ($\tau(P_k) \ge \log_{\sqrt 5 + 5}(k) - 1$) now follow as a corollary to Theorem~\ref{thm:PhiPk-lb} and Lemma~\ref{la:infinite} (in the latter, we simply plug in $G^\ast = P_\infty$, $G = P_k$, $\theta$ and $\theta^\ast$ to be the constant $1 + \frac{1}{k}$ and $1$ threshold weightings respectively).

\section{Randomized $\ACzero$ formulas computing the product of $k$ permutations}\label{sec:Pkupper}

In this section we define a broad class of randomized $\ACzero$ formulas for computing the product of $k$ permutations.  The size of these formulas corresponds to a complexity measure related to pathset complexity, but much simpler and easier to analyze.

\begin{df}
For integers $k \ge 1$, let $\mc P(k)$ be the set of sequences $\vec a = (a_0,\dots,a_k) \in [0,1]^{\{0,\dots,k\}}$ such that $\|\vec a\| \defeq a_0 + \dots + a_k \ge 1$. {We denote by $\vec a \le \vec b$ that each $a_i\leq b_i$.}

We define a complexity measure $\chi : \bigcup_{k \ge 1} \mc P(k) \to \R_{\ge 0}$ by the following induction:
\begin{itemize}
\item
In the base case $k=1$, let
$
  \chi(a_0,a_1) \defeq 0.
$
\item
For $k \ge 2$, let
\[
  \chi(a_0,\dots,a_k) \defeq 
  \min_{\substack{0 < i \le j < k,\\
  \vec b \in \mc P(k)\,:\\ 
  \vec a \le \vec b,\\
  (b_0,\dots,b_j) \in \mc P(j),\\
  (b_i,\dots,b_k) \in \mc P(k-i)
  }} 
  \|\vec b - \vec a\| + \max\{\chi(b_0,\dots,b_j),\, \chi(b_i,\dots,b_k)\}.\qquad
\]
\end{itemize}
\end{df}

Note the following properties of $\chi$:
\begin{enumerate}
   \item
    For all $\vec a \in \mc P(k)$ and $1 < i \le j < k$, if $(a_0,\dots,a_j) \in \mc P(j)$ and $(a_i,\dots,a_k) \in \mc P(k-i)$, then $\chi(\vec a) \le \max\{\chi(a_0,\dots,a_j),\,\chi(a_i,\dots,a_k)\}$.
	\item
    For all $\vec a,\vec b \in \mc P(k)$, if $\vec a \le \vec b$, then $\chi(\vec a) \le \|\vec b - \vec a\| + \chi(\vec b)$.  
  \item
    For all $k \ge 1$, we have
    $\ts\chi(\underbrace{\ts\frac12,\dots,\frac12}_{k+1\text{ times}}) = 0$. 
\end{enumerate}

The complexity measure $\chi(\vec a)$ is a simplified version of pathset complexity $\chi_A(\A)$.  In fact, $\chi(\vec a)$ provides an upper bound on pathset complexity!

\begin{rmk}\label{rmk:pathset-connection}
Consider the infinite pattern graph $P_\infty$ under the constant $1$ threshold weighting. 
For join-trees $A$ over $P_\infty$, we will write $\P{A}$ for $\P{A|\emptyset}$ and $\chi_A(\cdot)$ for $\chi_{A|\emptyset}(\cdot)$.

Each $\vec a \in \mc P(k)$ corresponds to a $P_k$-pathset
\[
  \A_{\vec a} &\defeq \{x \in [n]^{\{0,\dots,k\}} : 
  x_h \in S_h \text{ for all } h \in {\{0,\dots,k\}}\}
\]
where $S_0,\dots,S_k$ are arbitrary subsets of $[n]$ of size $|S_h| \le n^{1-a_h}$.
Then there exists a join-tree $A$ with graph $P_k$ such that
\[
  \chi_{A}(\A_{\vec a}) \le n^{\chi(\vec a)+o(1)}.
\]
This join-tree arises from the optimal $1 < i \le j < k$ and $\vec b$ in the definition of $\chi(\vec a)$: namely, $A = \sq{B,C}$ where $B$ is the join-tree for $P_{0,j}$ associated with $(b_0,\dots,b_j)$ and $C$ is the join-tree for $P_{i,k}$ associated with $(b_i,\dots,b_k)$.  (Note that $A$ has the property that $\Gr{D}$ is a path for each $D \preceq A$; so, not all join-trees with graph $P_k$ arise in this way.)

The above bound on $\chi_{A}(\A)$ {(now dropping the subscript as $\vec a $ is fixed)} is justified as follows. Letting $m \defeq n^{\|\vec b - \vec a\| + o(1)}$, there exist sets $T_{\ell,h} \subseteq S_{h}$ of size $|T_{\ell,h}| \le n^{1-b_h}$, indexed over $\ell \in [m]$ and $h \in \{0,\dots,k\}$, such that $\bigcup_{\ell \in [m]} (T_{\ell,0} \times \dots \times T_{\ell,k}) = S_0 \times \dots \times S_k$.
We then have
$
  \A = 
  \bigcup_{\ell \in [m]} 
  \B_\ell \bowtie \C_\ell
$
where 
\[
  \B_\ell &\defeq \{y \in [n]^{\{0,\dots,j\}} : 
  y_h \in T_{\ell,h} \text{ for all } 0 \le h \le j\},\\
  \C_\ell &\defeq \{z \in [n]^{\{i,\dots,k\}} : 
  z_h \in T_{\ell,h} \text{ for all } i \le h \le k \}.
\]
Arguing by induction on proper subsequences $(b_0,\dots,b_j)$ and $(b_i,\dots,b_k)$ (note that the base case $k = 1$ is trivial as $\chi(\vec a) = 0$ and $\chi_{A}(\A) = 1$ as $\A$ itself is a pathset), it follows that
\[
  \chi_{A}(\A)
  &\le
  \sum_{\ell \in [m]}
  \max\{\chi_{B}(\B_\ell),\, \chi_{C}(\C_\ell)\}\\
  &\le
  \sum_{\ell \in [m]}
  \max\{n^{\chi(b_0,\dots,b_j)+o(1)},\, n^{\chi(b_i,\dots,b_k)+o(1)}\}
  \le
  m \cdot n^{\chi(\vec a) - \|\vec b - \vec a\| +o(1)}
  =
  n^{\chi(\vec a) + o(1)}.
\]

As a consequence of these observations, we see that $\chi(\vec a)$ is lower-bounded by $\log_n(\chi_A(\A))$.  Since $\chi_A(\A) \le n^{\Phi(A)} \cdot \mu(\A) \le n^{\Phi(A)} \cdot n^{-\|\vec a\| + o(1)}$ (by Theorem \ref{thm:general-chiPhi}), it follows that 
\[
  \chi(\vec a) \ge \max_{\substack{\vphantom{t^t}\text{join-trees $A$ with graph $P_k$ s.t.}\\ 
  \text{$\Gr{D}$ is connected for all $D \preceq A$}}} \Phi(A) - \|\vec a\|.
\]
In particular, our lower bound of Section \ref{sec:Pk} implies that $\chi(\vec a) \ge \log_{\sqrt 5+5}(k)-\|\vec a\|-1$ for all $\vec a \in \mc P(k)$.

Finally, note that by covering the complete relation $[n]^{\{0,\dots,k\}}$ by $n^{\|\vec a\|+o(1)}$ shifted copies of rectangles $S_0 \times \dots \times S_k$, we get an upper bound 
\[
\chi_{A}([n]^{\{0,\dots,k\}}) \le n^{\chi(\vec a) + \|\vec a\|+o(1)}.
\]
\end{rmk}

By a similar construction, we will show that $n^{\chi(\vec a) + \|\vec a\| + o(1)}$ is an upper bound on the randomized $\ACzero$ formula size of computing the product of $k$ permutations.

\begin{df}
Let $\vec\pi = (\pi_1,\dots,\pi_k)$ be a sequence of permutations $[n] \stackrel\cong\to [n]$.
For a sequence $\vec x = (x_0,\dots,x_k) \in [n]^{\{0,\dots,k\}}$, we say that $\vec x$ is a {\em $\vec\pi$-path} if $\pi_h(x_{h-1}) = x_h$ for all $h \in \{1,\dots,k\}$.

If $\vec x$ is a $\vec\pi$-path and $\vec S = (S_0,\dots,S_k)$ is a sequence of sets $S_0,\dots,S_k \subseteq [n]$, we will say that {\em $\vec S$ isolates $\vec x$} if $\vec x \in S_0 \times \dots \times S_k$ and $\vec x$ is the {\em only} $\vec\pi$-path in $S_0 \times \dots \times S_k$.
\end{df}

\begin{df}
For a set $U$ and $p \in [0,1]$, notation $\mb S \subseteq_p U$ denotes that $\mb S$ is a random subset of $U$ that contains each element independently with probability $p$.

Given $\vec a = (a_0,\dots,a_k) \in \mc P(k)$, we will denote by $\vec{\mb S} = (\mb S_0,\dots,\mb S_k)$ the sequence of independent random sets $\mb S_h \subseteq_{n^{-a_h}} [n]$.
\end{df}

We now state the key lemma for our construction.

\begin{la}\label{la:formulas}
For every $\vec a \in \mc P(k)$ and sequence $\vec S = (S_0,\dots,S_k)$ of sets $S_h \subseteq [n]$, there exist randomized $\ACzero$ formulas 
\[
\mb f_{\vec a,\vec S} \text{ and }
\vec{\mb g}_{\vec a,\vec S} = 
\{\mb g_{\vec a,\vec S}^{(h,t)}\}_{\substack{r \in \{0,\dots,k\},\ 
t \in \{1,\dots,\lceil\log(n+1)\rceil\}}}
\]
each of depth $O(k)$ and size $n^{\chi(\vec a) + o(1)}$ and taking as input a sequence $\vec\pi = (\pi_1,\dots,\pi_k)$ 
of permutations $[n] \stackrel\cong\to [n]$, such that 
on every input $\vec\pi$ then with probability $1 - n^{-\omega(1)}
$ (with respect to both $\vec{\mb S}$ and the randomness of $\mb f_{\vec a,\vec{\mb S}}$ and $\vec{\mb g}_{\vec a,\vec{\mb S}}$):
\begin{enumerate}
  \item
    $\vec{\mb f}_{\vec a,\vec{\mb S}}(\vec\pi)$ outputs $1$ if, and only if, $\vec{\mb S}$ isolates some $\vec\pi$-path.
  \item
    If $\vec{\mb S}$ isolates a (necessarily unique) $\vec\pi$-path $\vec x = (x_0,\dots,x_k)$, then 
    formulas $\vec{\mb g}_{\vec a,\vec{\mb S}}(\vec\pi)$ output the binary representation of integers $x_0,\dots,x_k \in [n]$.
\end{enumerate}
\end{la}

\begin{proof}
The construction mimics the pathset complexity upper bound in Remark \ref{rmk:pathset-connection}. 
In the base case $k=1$, 
we have sets $S_0,S_1 \subseteq [n]$ and need to determine if a permutation $\pi : [n] \stackrel\cong\to [n]$ satisfies $\pi(x)=y$ for a unique pair $(x,y) \in S_0 \times S_1$. 
This is accomplished by the following $\ACzero$ formula (writing $1_{\pi(x)=y}$ for the input variable that is $1$ if and only if $\pi(x)=y$):
\[ 
  \mb f&{}_{\vec a,\vec S}(\pi) 
  \defeq
  \bigvee_{(x,y) \in S_0\times S_1} 1_{\pi(x)=y}\\
  &\wedge \bigwedge_{t \in \{1,\dots,\lceil\log(n+1)\rceil\}}
  \neg\bigg(
  \bigg(
\bigvee_{(x,y) \in S_0\times S_1 \,:\, \text{the $t^{\text{th}}$ bit of $x$ is $0$}} 1_{\pi(x)=y}\bigg)
  \wedge
  \bigg(
\bigvee_{(x,y) \in S_0\times S_1 \,:\, \text{the $t^{\text{th}}$ bit of $x$ is $1$}} 1_{\pi(x)=y}\bigg)
  \bigg)\\
  &\wedge \bigwedge_{t \in \{1,\dots,\lceil\log(n+1)\rceil\}}
  \neg\bigg(
  \bigg(
\bigvee_{(x,y) \in S_0\times S_1 \,:\, \text{the $t^{\text{th}}$ bit of $y$ is $0$}} 1_{\pi(x)=y}\bigg)
  \wedge
  \bigg(
\bigvee_{(x,y) \in S_0\times S_1 \,:\, \text{the $t^{\text{th}}$ bit of $y$ is $1$}} 1_{\pi(x)=y}\bigg)
  \bigg). 
\]
This formula has depth $O(1)$ and size $O(\log n)$ (as measured by number of gates). Since $\chi(\vec a) = 0$, this size bound is $n^{\chi(\vec a)+o(1)}$ as required. Formulas $\vec{\mb g}_{\vec a,\vec S}$ giving the binary representation of $x$ and $y$ (whenever $(x,y)$ uniquely exists) have just a single OR gate:
\[
  \mb g^{(0,t)}_{\vec a,\vec S}(\pi) &\defeq \bigvee_{(x,y) \in S_0\times S_1 \,:\, \text{the $t^{\text{th}}$ bit of $x$ is $1$}} 1_{\pi(x)=y},\\  
  \mb g^{(1,t)}_{\vec a,\vec S}(\pi) &\defeq \bigvee_{(x,y) \in S_0\times S_1 \,:\, \text{the $t^{\text{th}}$ bit of $y$ is $1$}} 1_{\pi(x)=y}.
\]

Onto the induction step where $k \ge 2$. Fix $0 < i \le j < k$ and $\vec b \in \mc P(k)$ with $\vec a \le \vec b$ and
\[
  \vec b' &\defeq (b_0,\dots,b_j) \in \mc P(j),\\
  \vec b'' &\defeq (b_i,\dots,b_k) \in \mc P(k-i),\\
  \chi(\vec a)\! &\phantom{:}= \|\vec b - \vec a\| + \max\{\chi(\vec b'),\,\chi(\vec b'')\}.
\] 
Letting $m \defeq n^{\|\vec b-\vec a\|+o(1)}$, we sample independent random sequences of sets $\vec{\mb T}_1,\dots,\vec{\mb T}_m$ where for each $\ell \in [m]$, we have $\vec{\mb T}_\ell = (\mb T_{\ell,0},\dots,\mb T_{\ell,k})$ with $\mb T_{\ell,h} \subseteq_{n^{-b_h+a_h}} S_h$.\footnote{A minor technicality arises when $a_h = b_h$: in this case, we should instead sample $\mb T_{\ell,h} \subseteq_{1/2} S_h$.  This case can be also avoided by approximating $\chi(\vec a)$ to an arbitrary additive constant $\eps > 0$: if we instead consider $\vec c$ defined by $c_h \defeq b_h+(\eps/k)$, then we have $\max\{\chi(c_0,\dots,c_j),\, \chi(c_i,\dots,c_k)\}
\le \max\{\chi(b_0,\dots,b_j) + (j+1)\eps,\, \chi(b_i,\dots,b_k) + (k-i+1)\eps\} \le \chi(\vec a) + \eps$.}

Writing $\vec{\mb T}_\ell'$ for $(\mb T_{\ell,0},\dots,\mb T_{\ell,j})$ and $\vec{\mb T}_\ell''$ for $(\mb T_{\ell,i},\dots,\mb T_{\ell,k})$ and $\vec\pi'$ for $(\pi_1,\dots,\pi_j)$ and $\vec\pi''$ for $(\pi_{i+1},\dots,\pi_k)$, we now introduce auxiliary randomized formulas $\mb{join}_1,\dots,\mb{join}_m$ defined by
\[
  \mb{join}_\ell(\vec\pi)
  &\defeq
  \mb f_{\vec b',\vec{\mb T}_\ell'}(\vec\pi') \wedge 
  \mb f_{\vec b'',\vec{\mb T}_\ell''}(\vec\pi'') \wedge 
  \bigwedge_{\substack{h \in \{i,\dots,j\}\\ t \in \{1,\dots,\lceil\log(n+1)\rceil\}}}
  \Big(\mb g^{(h,t)}_{\vec b',\vec{\mb T}_\ell'}(\vec\pi') 
  \leftrightarrow \mb g^{(h-i,t)}_{\vec b'',\vec{\mb T}_\ell''}(\vec\pi'')\Big).
\]
(Here $P \leftrightarrow Q$ abbreviates the formula $(P \wedge Q) \vee (\neg P \wedge \neg Q)$.)
If we consider the random sequence $\vec{\mb S}$ (in place of the arbitrary fixed sequence $\vec S$), then for every input $\vec\pi$, with high probability, the formula $\mb{join}_\ell(\vec\pi)$ outputs $1$ if, and only if, there exists a $\vec \pi$-path $(x_0,\dots,x_k)$ such that $\vec{\mb T}_\ell'$ isolates the $\vec\pi'$-path $(x_0,\dots,x_j)$ and $\vec{\mb T}_\ell''$ isolates the $\vec\pi''$-path $(x_i,\dots,x_k)$.

Note that the number of $\vec\pi$-paths in $\vec{\mb S}$ has expectation $n^{1-\|\vec a\|}$ ($\le 1$); it is easily shown that this number is at most $n^{o(1)}$ with high probability. For each $\vec\pi$-path $\vec x$ and $\ell \in [m]$, we have (by independence)
\[
  \Pr\Big[\ \vec x \in \mb T_{\ell,0}\times \dots \times \mb T_{\ell,k}
  \ \Big|\ \vec x \in \mb S_0 \times \dots \times \mb S_k\ \Big] 
  =
  n^{- \|\vec b - \vec a\|}.
\]
A further argument\footnote{This is a straightforward union bound.  Here is where we use the assumption that $b_h > a_h$ for all $h \in \{0,\dots,k\}$.} shows that
\[
  \Pr\Big[\ \vec{\mb T}_\ell' \text{ isolates } (x_0,\dots,x_j)
  \text{ and } \vec{\mb T}_\ell'' \text{ isolates } (x_i,\dots,x_k)
  \ \Big|\ \vec x \in \mb T_{\ell,0}\times \dots \times \mb T_{\ell,k}\ \Big] 
  =
  1 - o(1).
\]
By independence of $\vec{\mb T}_1,\dots,\vec{\mb T}_m$, we next have
\[
  \Pr\left[\ 
  \bigvee_{\ell \in [m]}
  \left(\begin{aligned}
  &\vec{\mb T}_\ell' \text{ isolates } (x_0,\dots,x_j) \text{ and\,}\\
  &\vec{\mb T}_\ell'' \text{ isolates } (x_i,\dots,x_k)
  \end{aligned}\right)
  \ \middle|\ \vec x \in \mb S_0 \times \dots \times \mb S_k\ \right]
  &\le
  1 - \Big(1 - \Omega(n^{-\|\vec b -\vec a\|})\Big)^m\\
  &\le
  1 - \exp(-\Omega(n^{-\|\vec b -\vec a\|} m)).
\]
Recalling that $m = n^{\|\vec b -\vec a\|+o(1)}$, the above bound will be $1 - n^{-\omega(1)}$ (i.e., ``with high probability'') for a suitable choice of $o(1)$ in the exponent of $m$ (for instance, if we set $m = n^{\|\vec b - \vec a\|} (\log n)^c$ for any constant $c > 1$).

We have shown that, with high probability, for every $\vec\pi$-path $(x_0,\dots,x_k)$ in $\mb S_0 \times \dots \times \mb S_k$, there exists $\ell \in [m]$ such that $\mb{join}_\ell(\vec\pi)$ outputs $1$. This justifies defining
\[
  \mb f_{\vec a,\vec S}(\vec\pi) 
  \defeq
  \mbox{}&
  \bigvee_{\ell \in [m]}\mb{join}_\ell(\vec\pi) 
  \\
  &
  \wedge \bigwedge_{\substack{h \in \{0,\dots,j\}\\ t \in \{1,\dots,\lceil\log(n+1)\rceil\}}} 
  \neg\bigg(
  \bigg(
  \bigvee_{\ell \in [m]}
  \mb{join}_\ell(\vec\pi) \wedge \mb g^{(h,t)}_{\vec b',\vec{\mb T}_\ell'}(\vec\pi')
  \bigg)
  \wedge
  \bigg(
  \bigvee_{\ell \in [m]}
  \mb{join}_\ell(\vec\pi) \wedge \neg\mb g^{(h,t)}_{\vec b',\vec{\mb T}_\ell'}(\vec\pi')
  \bigg)
  \bigg)\\
  & 
  \wedge \bigwedge_{\substack{h \in \{j+1,\dots,k\}\\ t \in \{1,\dots,\lceil\log(n+1)\rceil\}}} 
  \neg\bigg(
  \bigg(
  \bigvee_{\ell \in [m]}
  \mb{join}_\ell(\vec\pi) \wedge \mb g^{(h-i,t)}_{\vec b'',\vec{\mb T}_\ell''}(\vec\pi'')
  \bigg)
  \wedge
  \bigg(
  \bigvee_{\ell \in [m]}
  \mb{join}_\ell(\vec\pi) \wedge \neg\mb g^{(h-i,t)}_{\vec b'',\vec{\mb T}_\ell''}(\vec\pi'')
  \bigg)
  \bigg).
\]
In light of the above discussion, for random $\vec{\mb S}$, the subformula $\bigvee_{\ell \in [m]} \mb{join}_\ell(\vec\pi)$ with high probability asserts the existence of a $\vec\pi$-path in $\vec{\mb S}$, while the remainder of $\mb f_{\vec a,\vec{\mb S}}(\vec \pi)$ asserts uniqueness.  Formulas $\vec{\mb g}_{\vec a,\vec S}$ are defined by
\[
  \mb g^{(h,t)}_{\vec a,\vec S}(\vec\pi)
  \defeq 
  \begin{cases}
  \ds
  \bigvee_{\ell \in [m]}
  \mb{join}_\ell(\vec\pi) \wedge 
  \mb g^{(h,t)}_{\vec b',\vec{\mb T}_\ell'}(\vec\pi')
  &\text{if } h \in \{0,\dots,j\},\\
  \ds
  \bigvee_{\ell \in [m]}
  \mb{join}_\ell(\vec\pi) \wedge 
  \mb g^{(h-i,t)}_{\vec b'',\vec{\mb T}_\ell''}(\vec\pi'')
  &\text{if } h \in \{j+1,\dots,k\}.
  \end{cases}
\]

If formulas $f_{\vec b',\vec{\mb T}_\ell'}$ and $\vec{\mb g}_{\vec b',\vec{\mb T}_\ell'}$ (respectively, $f_{\vec b'',\vec{\mb T}_\ell''}$) and $\vec{\mb g}_{\vec b'',\vec{\mb T}_\ell''}$ have depth at most $d'$ and size at most $z'$ (respectively, $d''$ and $z''$),
then it is readily seen that formulas $\mb f_{\vec a,\vec S}$ and $\vec{\mb g}_{\vec a,\vec S}$ have depth $d$ and size $z$ where 
\[
d = \max\{d',d''\} + O(1),\qquad z = O\big(m \cdot (k \log n)^2 \cdot (z' + z'')\big) = n^{\|\vec b - \vec a\| + o(1)} \cdot (z' + z'').
\]
This recurrence justifies the bounds $d = O(k)$ and $z = n^{\chi(\vec a)+o(1)}$.

As for the error probability of Properties (1) and (2),  
it should be clear that every usage of ``with high probability'' in this argument can be made to be 
$1 - n^{-\omega(1)}$ by setting $m = n^{\|\vec b - \vec a\|} (\log n)^{c_k}$ for suitable constants $c_k > 1$.
\end{proof}

Lemma \ref{la:formulas} has the following corollary, which for each $\vec a \in \mc P(k)$, gives a collection of randomized $\ACzero$ formulas of aggregate size $n^{\chi(\vec a) + \|\vec a\| + o(1)}$  that compute the product of $k$ permutations.  Moreover, these formulas, on input $\vec \pi$, produce a list of all paths $\vec x = (x_0,\dots,x_k) \in [n]^{\{0,\dots,k\}}$ such that $\pi_h(x_{h-1})=x_h$ for all $h \in \{1,\dots,k\}$.

\begin{cor}\label{cor:formulas}
For every $\vec a \in \mc P(k)$, there exists a matrix of randomized $\ACzero$ formulas 
\[
  \vec{\mb h}_{\vec a} = \{\mb h_{\vec a}^{(\ell,t)}\}_{\substack{\ell \in \{1,\dots,n^{\|a\|+o(1)}\},\ 
  t \in \{1,\dots,(k+1)\log\lceil(n+1)\rceil\}}}
\]
each of depth $O(k)$ and size $n^{\chi(\vec a) + o(1)}$ and taking a sequence $\vec\pi = (\pi_1,\dots,\pi_k)$ 
of permutations $[n] \stackrel\cong\to [n]$
as input, such that 
the following properties hold 
with probability $1 - n^{-\omega(1)}$:
\begin{enumerate}
  \item
    Each row in $\vec{\mb h}_{\vec a}(\vec\pi)$ is either the all-0 string or contains the binary representation of integers $x_0,\dots,x_k$ for some $\vec\pi$-path $\vec x \in [n]^{\{0,\dots,k\}}$.
  \item
    For every 
    $\vec\pi$-path $\vec x \in [n]^{\{0,\dots,k\}}$, the binary representation of integers $x_0,\dots,x_k \in [n]$ is given by at least one row of $\vec{\mb h}_{\vec a}(\vec\pi)$. 
\end{enumerate}
\end{cor}

Similar to the bound $\chi_{A}([n]^{\{0,\dots,k\}}) \le n^{\chi(\vec a) + \|\vec a\|+o(1)}$ in Remark \ref{rmk:pathset-connection},
Corollary \ref{cor:formulas} is obtained from Lemma \ref{la:formulas} by covering $[n]^{\{0,\dots,k\}}$ with $m \defeq n^{\|\vec a\|+o(1)}$ random rectangles $\mb S_{\ell,0}\times \dots \times \mb S_{\ell,k}$ where each $\vec{\mb S}_\ell = (\mb S_{\ell,0},\dots,\mb S_{\ell,k})$ has the same distribution as $\vec{\mb S}$ (i.e.,\ $\mb S_{\ell,h} \subseteq_{n^{-a_h}} [n]$).  The rows of $\vec{\mb h}_{\vec a}(\vec\pi)$ are then given by the conjunction of $\mb f_{\vec a,\vec{\mb S}_\ell}(\vec\pi)$ with formulas $\vec{\mb g}_{\vec a,\vec{\mb S}_\ell}(\vec\pi)$.  Property (1) is immediate from Lemma \ref{la:formulas}, while Property (2) follows by noting that, with high probability, every $\vec\pi$-path in $[n]^{\{0,\dots,k\}}$ is isolated by a rectangle $\mb S_{\ell,0}\times \dots \times \mb S_{\ell,k}$ for some $\ell \in [m]$.

\subsection{Upper bounds on $\chi(\vec a) + \|\vec a\|$}

We describe a few different constructions giving upper bounds on $\chi(\vec a) + \|\vec a\|$ for sequences $\vec a \in \mc P(k)$.  Thanks to Corollary \ref{cor:formulas}, each of these constructions corresponds to randomized $\ACzero$ formulas of size $n^{\chi(\vec a) + \|\vec a\| + o(1)}$ computing the product of $k$ permutations.
Our best bound, $\frac13\log_{(\sqrt 5+1)/2}(k)+O(1)$, is obtained via a construction we call ``Fibonacci overlapping''.

\subsubsection{Recursive doubling}

For $k \ge 2$, let
\[
  \vec a_k \defeq 
  \ts(\frac12,\underbrace{0,\dots,0}_{\lceil k/2 \rceil - 1},\frac12,\underbrace{0,\dots,0}_{\lfloor k/2 \rfloor - 1},\frac12)
\]
Then $\chi(\vec a_2) = \chi(\frac12,\frac12,\frac12) = \max\{\chi(\frac12,\frac12),\,\chi(\frac12,\frac12)\} = 0$ and for $k \ge 3$,
\[
  \chi(\vec a_k)
  &\ts\le
  \max\{\chi(\frac12,\underbrace{0,\dots,0}_{\lceil k/2 \rceil - 1},\frac12),\, \chi(\frac12,\underbrace{0,\dots,0}_{\lfloor k/2 \rfloor - 1},\frac12)\}
  \le
  \max\{\frac12 + \chi(\vec a_{\lceil k/2 \rceil}),\, 
  \frac12 + \chi(\vec a_{\lceil k/2 \rceil})\}
  \le
  \frac12 \lceil\log_2(k)\rceil.
\]
This construction achieves
\[
  \ts\chi(\vec a_k) + \|\vec a_k\| \le \frac12 \lceil\log_2(k)\rceil + 1.
\]

\subsubsection{Maximally overlapping joins}

For $k \ge 2$, let 
\[
  \vec a_k \defeq \ts(\underbrace{\ts\frac{1}{k},\dots,\frac{1}{k}}_{k+1}) \in \mc P(k).
\]
Then $\chi(\vec a_2) = \chi(\frac12,\frac12,\frac12) = 0$ and for $k \ge 3$,
\[
  \chi(\vec a_k) 
  &\ts\le
  \max\{\chi(\underbrace{\ts\frac{1}{k},\dots,\frac{1}{k}}_{k}),\,
  \chi(\underbrace{\ts\frac{1}{k},\dots,\frac{1}{k}}_{k})\}
  \le
  \chi(\vec a_{k-1}) + k\big(\frac1{k-1} - \frac1k\big)
  =
  \chi(\vec a_{k-1}) + \frac{1}{k-1}
  \le
  \frac{1}{2} + \dots + \frac{1}{k-1}.
\]
This construction achieves
\[
  \chi(\vec a_k) + \|\vec a_k\|
  &\ts\le
  1 + \frac{1}{2} + \dots + \frac{1}{k-1} + \frac{1}{k}
  =
  \ln(k) + O(1).
\]
Since $\ln(k) \approx 0.69\log_2(k)$, this upper bound is worse that the one from recursive doubling.  

It turns out that a $\frac12\log_2(k) + O(1)$ upper bound is achievable via a different construction via the maximally overlapping join-tree.  If $k = 2^\ell + t$ where $\ell \ge 0$ and $t \in \{0,\dots,2^\ell-1\}$, we instead define
\[
  \vec a_k \defeq 
  (\underbrace{\ts\frac{1}{2^{\ell+1}},\dots,\frac{1}{2^{\ell+1}}}_{t},
  \underbrace{\ts\frac{1}{2^{\ell}},\dots,\frac{1}{2^{\ell}}}_{2^\ell - t + 1},
  \underbrace{\ts\frac{1}{2^{\ell+1}},\dots,\frac{1}{2^{\ell+1}}}_{t}).
\]
(Note that $\vec a \in \mc P(k)$ since $\|\vec a\| = 
(2^\ell - t + 1)\frac{1}{2^\ell} + 2t\frac{1}{2^{\ell+1}} = 
1 + \frac{1}{2^\ell} > 1
$.)
In the base case $k = 1$ (i.e.,\ $\ell = t = 0$), we have 
$\chi(\vec a_2) = \ts\chi(1,1) = 0$.
When $\ell \ge 1$ and $t \ge 1$, we have
\[
  \chi(\vec a_k)
  &\le
  \max\{
  \chi(\underbrace{\ts\frac{1}{2^{\ell+1}},\dots,\frac{1}{2^{\ell+1}}}_{t},
  \underbrace{\ts\frac{1}{2^{\ell}},\dots,\frac{1}{2^{\ell}}}_{2^\ell - t + 1},
  \underbrace{\ts\frac{1}{2^{\ell+1}},\dots,\frac{1}{2^{\ell+1}}}_{t-1}),\,
  \chi(\underbrace{\ts\frac{1}{2^{\ell+1}},\dots,\frac{1}{2^{\ell+1}}}_{t-1},
  \underbrace{\ts\frac{1}{2^{\ell}},\dots,\frac{1}{2^{\ell}}}_{2^\ell - t + 1},
  \underbrace{\ts\frac{1}{2^{\ell+1}},\dots,\frac{1}{2^{\ell+1}}}_{t})\}\hspace{-3in}\\
  &=
  \chi(\underbrace{\ts\frac{1}{2^{\ell+1}},\dots,\frac{1}{2^{\ell+1}}}_{t},
  \underbrace{\ts\frac{1}{2^{\ell}},\dots,\frac{1}{2^{\ell}}}_{2^\ell - t + 1},
  \underbrace{\ts\frac{1}{2^{\ell+1}},\dots,\frac{1}{2^{\ell+1}}}_{t-1}) &&\text{(by symmetry)}\\
  &=
  \chi(\underbrace{\ts\frac{1}{2^{\ell+1}},\dots,\frac{1}{2^{\ell+1}}}_{t-1},
  \underbrace{\ts\frac{1}{2^{\ell}},\dots,\frac{1}{2^{\ell}}}_{2^\ell - t + 2},
  \underbrace{\ts\frac{1}{2^{\ell+1}},\dots,\frac{1}{2^{\ell+1}}}_{t-1}) + \ts\frac{1}{2^{\ell}} - \frac{1}{2^{\ell+1}}
  &&\text{($t^{\text{th}}$ coordinate increases from $\tsfrac{1}{2^{\ell+1}}$ to $\tsfrac{1}{2^{\ell}}$)}\\
  &=
  \chi(\vec a_{k-1}) + \ts\frac{1}{2^{\ell+1}}.
\]
When $\ell \ge 1$ and $t = 0$ (i.e.,\ $k = 2^\ell$), we have
\[
  \chi(\vec a_k)
  = \chi(\underbrace{\ts\frac{1}{2^\ell},\dots,\frac{1}{2^\ell}}_{2^\ell + 1})
  \le
  \chi(\underbrace{\ts\frac{1}{2^\ell},\dots,\frac{1}{2^\ell}}_{2^\ell})
  &\le
  \chi(
  \underbrace{\ts\frac{1}{2^\ell},\dots,\frac{1}{2^\ell}}_{2^{\ell-1}-1},
  \tsfrac{1}{2^{\ell-1}},\tsfrac{1}{2^{\ell-1}},
  \underbrace{\ts\frac{1}{2^\ell},\dots,\frac{1}{2^\ell}}_{2^{\ell-1}-1}
  )
  + 2(\tsfrac{1}{2^{\ell-1}} - \tsfrac{1}{2^{\ell}})\\
  &=
  \chi(\vec a_{k-1}) + \tsfrac{1}{2^{\ell-1}}.
\]
For any $k = 2^\ell + t$, this recurrence shows
\[
  \chi(\vec a_k) 
  &\le
  t\cdot
  \tsfrac{1}{2^{\ell+1}} + \chi(\vec a_{2^\ell})\\
  &\le
  t\cdot
  \tsfrac{1}{2^{\ell+1}} + \tsfrac{1}{2^{\ell-1}} + \chi(\vec a_{2^\ell-1})\\
  &\le
  t\cdot\tsfrac{1}{2^{\ell+1}} + \big(\tsfrac{1}{2^{\ell-1}} + (2^{\ell-1}-1)\cdot \tsfrac{1}{2^\ell}\big) + 
  \chi(\vec a_{2^{(\ell-1)}})\\
  &=
  t\cdot
  \tsfrac{1}{2^{\ell+1}} + \big(\tsfrac12 + \tsfrac{1}{2^\ell}\big) + 
  \chi(\vec a_{2^{(\ell-1)}})\\
  &\le  
  t\cdot \tsfrac{1}{2^{\ell+1}} + 
  \ts\sum_{j=1}^\ell
   \big(\tsfrac12 + \tsfrac{1}{2^j}\big)\\
  &=
  \ts\frac{1}{2}(\ell+\frac{t}{2^\ell}) + 
  1 - \frac{1}{2^\ell}.
\]
This second construction thus achieves
\[
\chi(\vec a) + \|\vec a\| \le
\ts\frac{1}{2}(\ell+\frac{t}{2^\ell}) + 2
= \tsfrac12\log_2(k)+O(1).
\]

\subsubsection{Fibonacci overlapping joins}

Let $\Fib(1)=\Fib(2)=1$ and for $\ell \ge 3$, let $\Fib(\ell) \defeq \Fib(\ell-1)+\Fib(\ell-2)$. 
For $\ell \ge 4$, let
\[
  \chi(\vec a_{\Fib(\ell)}) 
  &\defeq
  \ts(\frac13,\underbrace{0,\dots,0}_{\Fib(\ell-2)-1},\frac13,\underbrace{0,\dots,0}_{\Fib(\ell-3)-1},\frac13,\underbrace{0,\dots,0}_{\Fib(\ell-2)-1},\frac13)
  \in
  \mc P(\Fib(\ell))
\]
We have $\Fib(4) = 3$ and
\[
  \chi(\vec a_{\Fib(4)}) 
  =
  \ts\chi(\frac13,\frac13,\frac13,\frac13) 
  \le
  \chi(\frac13,\frac13,\frac13)
  \le
  \frac13 + \max\{\chi(\frac13,\frac23),\,\chi(\frac23,\frac13)\}
  =
  \frac13.
\]
For $\ell \ge 5$, we have
\[
  \chi(\vec a_{\Fib(\ell)}) 
  &\le
  \ts\max\{\chi(
  \frac13,\underbrace{0,\dots,0}_{\Fib(\ell-2)-1},\frac13,\underbrace{0,\dots,0}_{\Fib(\ell-3)-1},\frac13),\,
  \chi(\frac13,\underbrace{0,\dots,0}_{\Fib(\ell-3)-1},\frac13,\underbrace{0,\dots,0}_{\Fib(\ell-2)-1},\frac13)\}\\
  &\ts=
  \chi(\frac13,\underbrace{0,\dots,0}_{\Fib(\ell-2)-1},\frac13,\underbrace{0,\dots,0}_{\Fib(\ell-3)-1},\frac13)
  \quad\text{(by symmetry)}\\
  &\ts\le
  \frac13 + 
  \chi(\frac13,\underbrace{0,\dots,0}_{\Fib(\ell-3)-1},\frac13,\underbrace{0,\dots,0}_{\Fib(\ell-4)-1},\frac13,\underbrace{0,\dots,0}_{\Fib(\ell-3)-1},\frac13) 
  =
  \frac13 + \chi(\vec a_{\Fib(\ell-1)})
  =
  \frac13\ell - 1.
\]
For $k = \Fib(\ell)$ with $\ell \ge 4$, this construction gives $\vec a_k \in \mc P(k)$ with
\[
  \chi(\vec a_k) + \|\vec a_k\| \le \ts\frac{1}{3}(\ell+1).
\]
For $\Fib(\ell-1) < k \le \Fib(\ell)$, the bound $\chi(\vec a_k) + \|\vec a_k\| \le \ts\frac{1}{3}(\ell+1)$ extends to all $\vec a_k \in \mc P(k)$ of the form
\[
  \vec a_k \defeq 
  \ts(\frac13,\underbrace{0,\dots,0}_{\le \Fib(\ell-2)-1},\frac13,\underbrace{0,\dots,0}_{\le \Fib(\ell-3)-1},\frac13,\underbrace{0,\dots,0}_{\le \Fib(\ell-2)-1},\frac13).
\]
This construction proves the following

\begin{thm}\label{thm:Fib}
For all $k \ge 1$, there exists $\vec a \in \mc P(k)$ with
\[
  \ts\chi(\vec a) + \|\vec a\| = \frac13\log_{\varphi}(k) + O(1) 
\]
where $\varphi = (\sqrt 5 + 1)/2$ is the golden ratio.
\end{thm}

Since $\frac13\log_{\varphi}(k) = \log_{\sqrt 5+2}(k) \le 0.49\log_2(k)$, Theorem \ref{thm:Fib} improves the $\frac12\log_2(k) + O(1)$ upper bounds from the recursive doubling and maximally overlapping join-trees described above.  As a corollary of Corollary \ref{cor:formulas} and Theorem \ref{thm:Fib}, we have

\begin{cor}
There exist randomized $\ACzero$ formulas of size $n^{\frac13\log_{\varphi}(k) + O(1)}$ that compute the product of $k$ $(n\times n)$-permutation matrices.
\end{cor}

\subsection{Tightness of upper bounds}

We say that a join-tree $A$ (over $P_\infty$) is {\em connected} if $\Gr{D}$ is connected for all $D \preceq A$. For every connected join-tree $A$ with graph $P_k$, we can consider the constrained complexity measure $\chi_A : \mc P(k) \to \R_{\ge 0}$ where parameters $0 < i \le j < k$ in the definition of $\chi(\vec a)$ are fixed according to $A$.  As described in Remark \ref{rmk:pathset-connection}, the potential function $\Phi(A)$ implies a lower bound $\chi_A(\vec a) \ge \Phi(A) - \|\vec a\|$.

Let $\mr{RD}_k$, $\mr{MO}_k$ and $\mr{FO}_k$ be the ``recursive doubling'', ``maximally overlapping'' and ``Fibonacci overlapping'' connected join-trees recursively defined by $\mr{RD}_1 = \mr{MO}_1 = \mr{FO}_1 = \sq{\{0,1\}}$ and
\[
  &&&&&&&&\mr{RD}_k = \mr{RD}_{0,k} &\defeq \sq{\mr{RD}_{0,\lceil k/2 \rceil},\mr{RD}_{\lceil k/2 \rceil,k}} 
  &&\text{for }k \ge 2,&&&&&&&&\\
  &&&&&&&&\mr{MO}_k = \mr{MO}_{0,k} &\defeq \sq{\mr{MO}_{0,k-1},\mr{MO}_{1,k}} 
  &&\text{for }k \ge 2,\\
  &&&&&&&&\mr{FO}_k = \mr{FO}_{0,k} &\defeq \sq{\mr{FO}_{0,\Fib(\ell-1)},\mr{FO}_{\Fib(\ell-1),k}} 
  &&\text{for }k = \mr{Fib}(\ell),\, \ell \ge 3.
\]
The upper bounds on $\chi(\cdot)$ in the previous subsection respectively apply to constrained complexity measures $\chi_{\mr{RD}_k}(\cdot)$, $\chi_{\mr{MO}_k}(\cdot)$ and $\chi_{\mr{FO}_k}(\cdot)$.

With respect to these particular join-trees, the upper bounds of previous subsection are in fact tight!  By ($\dag$) we have
\[
  &&&&&&&&\Phi(\mr{RD}_k) &\ge \Phi(\mr{RD}_{\lceil \lceil k/2 \rceil / 2\rceil}) + 1 
  &&\text{for } k \ge 4&&&&&&&&,\\
  &&&&&&&&\Phi(\mr{FO}_{\Fib(\ell)}) &\ge \Phi(\mr{FO}_{\Fib(\ell-3)}) + 1
  &&\text{for } k = \Fib(\ell),\, \ell \ge 5.
\]
It follows that $\Phi(\mr{RD}_k) \ge \frac12\log_2(k)$ and $\Phi(\mr{FO}_k) \ge \frac13\log_\varphi(k) - O(1)$.
We get a lower bound on $\Phi(\mr{MO}_k) \ge \frac12\log_2(k) - O(1)$ via ($\ddag$):
\[
  &&&&&&&&\Phi(\mr{MO}_k) &\ge 
  \tsfrac12\big(
  \Phi(\mr{MO}_{0,\lfloor (k-1)/2 \rfloor})
  +
  \Phi(\mr{MO}_{\lceil (k-1)/2 \rceil,k})
  +
  1
  \big)
  &&\text{for } k \ge 3.&&&&&&&&
\]
Therefore, for all $\vec a \in \mc P(k)$, we have
\[
  \chi_{\mr{RD}_k}(\vec a) + \|\vec a\| &\ge \tsfrac12\log_2(k),\\
  \chi_{\mr{MO}_k}(\vec a) + \|\vec a\| &\ge \tsfrac12\log_2(k) - O(1),\\
  \chi_{\mr{FO}_k}(\vec a) + \|\vec a\| &\ge \tsfrac13\log_\varphi(k) - O(1).
\]
This establishes the tightness of our upper bounds for these specific join-trees.  It is open whether a different connected join-tree achieves a better bound.  (The best lower bound on $\chi(\vec a) + \|\vec a\|$ that we could determine is $\log_{\sqrt 5 + 4}(k) - 1$ via a strengthening of the argument in Section \ref{sec:Pk} in the case of connected join-trees.)

\paragraph{Experimental results.}
For any connected join-tree $A$, the value
$\min_{\vec a \in \mc P(k)}\, \chi_A(\vec a)$
is computable by a linear program with $O(\sum_{D \preceq A} |V(D)|)$ variables and $O(\sum_{D \preceq A} |V(D)|)$ constraints.
In fact, our second upper bound for the maximally overlapping pattern (achieving $\chi_{\mr{MO}_k}(\vec a_k)+\|\vec a_k\| = \frac12\log_2(k)+O(1)$) was found with the help of this linear programs!

We experimentally searched for connected join-trees $A$ that beat the $\frac13\log_\varphi(k)+O(1)$ upper bound via Fibonacci overlapping by evaluating 
$\min_{\vec a \in \mc P(k)}\, \chi_A(\vec a) + \|\vec a\|$ on various examples, both structured and randomly generated. We could not find any better upper bound. (In particular, join-trees $\mr{FO}_k$ appears optimal among a broad class of ``recursively overlapping'' join-trees.)
It is tempting to conjecture that $\mr{FO}_k$ in fact gives the optimal bound, that is, $\chi_A(\vec a) + \|\vec a\| \le \frac13\log_\varphi(k)+O(1)$ for all $\vec a \in \mc P(k)$. 
We leave this as an intriguing open question.

\section{Open problems}\label{sec:open-problems}

We conclude by mentioning some open questions raised by this work.

\begin{problem}\label{prob1}
Prove that $\tau(G) = \Omega(\td(G))$ for all graphs $G$.  (An $\wt\Omega(\td(G))$ would also be interesting.)
\end{problem}

Problem \ref{prob1} is unlikely to follow from any excluded-minor approximation of tree-depth along the lines of Theorem \ref{thm:excluded-minor}.  A first step to resolving this problem is to identify, for each graph $G$, a particular threshold weighting $\theta$ such that $\mb X_{\theta,n}$ is a ``hard'' input distribution with respect to the average-case $\ACzero$ formula size of $\SUB(G)$.  (The paper \cite{li2017ac} does precisely this with respect to $\ACzero$ circuit size.)

Problem \ref{prob1} should be easier to tackle in the special case of trees.  (We remark that our lower bounds for $P_k$ and $T_k$, combined with results in \cite{czerwinski2019improved,KR18}, imply that $\tau(T) = \Omega(\sqrt{\td(T)})$ for all trees $T$.)

\begin{problem}\label{prob2}
Prove that $\tau(T) = \Omega(\td(T))$ for all trees $T$.
\end{problem}

A solution to Problem \ref{prob2} could perhaps be shown by a common generalization of our lower bounds for $P_k$ and $T_k$.

A third open problem is to nail down the exact average-case $\ACzero$ formula size of $\SUB(P_k)$ (or the related problem of multiplying $k$ permutations).

\begin{problem}\label{prob3}
Prove that $\tau(P_k) = \frac13\log_\varphi(k)$ or find an upper bound improving Theorem \ref{thm:Pkupper}.
\end{problem}

Finally and most ambitiously:

\begin{problem}\label{prob4}
Prove that $n^{\tau(G)-o(1)}$ is a lower bound the {\em unrestricted} formula size $\SUB(G)$.
\end{problem}

This of course would imply $\NCone \ne \cc{NL}$.  An $n^{\Omega(\log k)}$ lower bound for $\SUB(P_k)$ in the average-case (or for the problem of multiplying $k$ permutations) would moreover imply $\NCone \ne \cc{L}$.  Although Problem \ref{prob3} lies beyond current techniques, the applicability of the pathset framework in establishing $n^{\tau(G)-o(1)}$ lower bounds in the disparate $\ACzero$ and monotone settings is possibly reason for optimism.

\appendix{}

\section{Appendix: Lower bound $\tau(P_k) \ge \frac12\log_{\sqrt{13}+1}(k)$ from \cite{rossman2018formulas}}\label{sec:appendix}

\newcommand{\PHI}[1]{\Phi(#1)}
\renewcommand{\c}[1]{\Delta(#1)}
\newcommand{\AB}{\sq{A,B}}
\newcommand{\parent}[1]{#1^{\uparrow}}
\newcommand{\sib}[1]{#1^{\sim}}
\newcommand{\edge}[2]{\mr{edge}(#1,#2)}
\newcommand{\LLL}[1]{\lambda(#1)}

This appendix gives the proof of Lemma \ref{la:PhiPk} from \cite{rossman2018formulas}. 
As in Section \ref{sec:Pk}, we consider infinite pattern graph $P_\infty$ with the constant threshold weighting $E(P_\infty) \to \{1\}$.

\begin{df}
For a join-tree $A$, let $\LLL{A}$ denote the length of the longest path in $A$ (i.e.,\ the number of edges in the largest connected component of $A$).
\end{df}

We omit the proof of the following lemma, which is similar to Lemma \ref{la:zoom}.

\begin{la}\label{la:A-S}
For every join-tree $A$ and set $S$, if $S$ intersects $t$ distinct connected components of $\Gr{A}$, then
\[
  \PHI{A} \ge \PHI{A \ominus S} + t.
\]
\end{la}

We now present the result from \cite{rossman2018formulas} that implies Lemma \ref{la:PhiPk}.
(The precise value of $c$ in Lemma \ref{la:PHI-lb}, below, is thanks to an optimization suggested by an anonymous referee of the journal paper \cite{rossman2018formulas}.)


\begin{la}\label{la:PHI-lb}
For every join-tree $A$, $\PHI{A} \ge \tsfrac1c\log(\LLL{A}) + \c{A}$ where $c = 2\log(\sqrt{13}+1)$.
\end{la}


\begin{proof}
Here $c$ is chosen such that $\frac12 - \frac{1}{2^{c/2}} - \frac{1}{2^{c-1}} = \frac{1}{2^{c-2}}$.

We argue by induction on join-trees. The base case where $A$ is atomic is trivial. For the induction step, let $A$ be a non-atomic join-tree and assume the lemma holds for all smaller join-trees. We will consider a sequence of cases, which will be summarized at the end of the proof.

First, consider the case that $G_A$ is disconnected. Let $t = \c{A}$ ($\ge 2$). Let $S$ be the set of all vertices of $G_A$, except those in the largest connected component of $G_A$.  We have
\begin{align*}
  &&&&\PHI{A} 
    &\ge \PHI{A{\uhr}S} + t - 1
    &&\text{(Lemma \ref{la:A-S})}&&&&\\
  &&&&&\ge \tsfrac1c\log(\LLL{A \ominus S}) + \c{A{\uhr}S} + t - 1
    &&\text{(induction hypothesis)}\\
  &&&&&= \tsfrac1c\log(\LLL{A}) + \c{A}.
\end{align*}
This proves the lemma in the case where $G_A$ is disconnected.

Therefore, we proceed under the assumption that $G_A$ is connected (i.e.\ $\c{A} = 1$). Without loss of generality, we assume that $G_A = P_k$ (i.e.\ $\LLL{A} = k$). Our goal is to show that
\[
  \PHI{A} \ge \tsfrac1c\log(k)+1.
\]

Consider the case that there exists a sub-join-tree $A' \preceq A$ such that $|E_{A'}| \ge \tsfrac{1}{2^{c-1}}k$ and $\c{A'} \ge 2$. Note that $\LLL{A'} \ge |E_{A'}|/\c{A'}$ (i.e.\ the number of edges in the largest component of $G_{A'}$ is at least the number of edges in $G_{A'}$ divided by the number of components in $G_{A'}$). We have
\begin{align*}
  \PHI{A} &\ge \PHI{A'}\\
  &\ge \tsfrac1c\log(\LLL{A'}) + \c{A'}
  &&\text{(induction hypothesis)}\\
  &\ge \tsfrac1c\log(k) - \tsfrac{c-1}{c} - \tsfrac1c\log(\c{A'}) + \c{A'}
  &&\text{($\LLL{A'} \ge |E_{A'}|/\c{A'} \ge \tsfrac{1}{2^{c-1}}k\c{A'}$)}\\
  &\ge \tsfrac1c\log(k) - \tsfrac{c-1}{c} - \tsfrac1c\log(2) + 2
  &&\text{($\c{A'} \ge 2$ and $x - \tsfrac1c\log x$ increasing)}\\
  &= \tsfrac1c\log(k) + 1.
\end{align*}
This proves the lemma in this case.

Therefore, we proceed under the following assumption:
\begin{equation}
\tag{$\hspace{-.5pt}{\circledast}\hspace{-.5pt}$}
\label{eq:ast}
 \text{for all $A' \preceq A$, if 
 $|E_{A'}| \ge \tsfrac{1}{2^{c-1}}k$
 then }\c{A'} = 1.
\end{equation}
Going forward, the following notation will be convenient: for a proper sub-join-tree $B \prec A$, let $\parent{B}$ denote the parent of $B$ in $A$, and let $\sib{B}$ denote the sibling of $B$ in $A$. Note that $\parent{B} = \{B,\sib{B}\} \preceq A$.

By walking down the join-tree $A$, we can proper sub-join-trees $B,Z \prec A$ such that 
\[
  v_0 \in V_B,\qquad
  v_k \in V_Z,\qquad
  |E_B|,|E_Z| < \tsfrac{1}{2^{c/2}}k,\qquad
  |E_{\parent{B}}|,|E_{\parent{Z}}| \ge \tsfrac{1}{2^{c/2}}k.
\]
Fix any choice of such $B$ and $Z$. Note that $G_{\parent{B}}$ and $G_{\parent{Z}}$ are connected by (\ref{eq:ast}). In particular, $G_{\parent{B}}$ is a path of length $|E_{\parent{B}}|$ with initial endpoint $v_0$, and $G_{\parent{Z}}$ is a path of length $|E_{\parent{Z}}|$ with final endpoint $v_k$.

Consider the case that $\parent{B}$ and $\parent{Z}$ are vertex-disjoint. Note that $\tsfrac{1}{2^{c/2}}k \ge \tsfrac{1}{2^{c-1}}k$, so the assumption (\ref{eq:ast}) implies that $\parent{B}$ and $\parent{Z}$ are connected and $\LLL{\parent{B}},\LLL{\parent{Z}} \ge \tsfrac{1}{2^{c/2}}k$.
Let $Y$ denote the least common ancestor of $\parent{B}$ and $\parent{Z}$ in $A$. We have
\begin{align*}
  \PHI{A} &\ge \PHI{Y}\\
  &\ge \tsfrac12\big(\PHI{\parent{B}} + \PHI{\parent{Z} \ominus \parent{B}} + \c{Y} + \c{Y \ominus \{\parent{B},\parent{Z}\}}\big)
  &&\text{(by $(\ddag)$)}\\
  &\ge \tsfrac12\big(\PHI{\parent{B}} + \PHI{\parent{Z}}\big) + \tsfrac12
  &&\text{($\c{Y} \ge 1$ and $\parent{Z} \ominus \parent{B} = \parent{Z}$)}\\
  &\ge \tsfrac12\big(\tsfrac1c\log(\LLL{\parent{B}}) + \c{\parent{B}} + \tsfrac1c\log(\LLL{\parent{Z}}) + \c{\parent{Z}}\big) + \tsfrac12
  &&\text{(induction hypothesis)}\\
  &\ge \tsfrac12\big(\tsfrac1c\log(\tsfrac{1}{2^{c/2}}k) + 1 + \tsfrac1c\log(\tsfrac{1}{2^{c/2}}k) + 1\big) + \tsfrac12\\
  &= \tsfrac1c\log(k) + 1.
\end{align*}

Therefore, we proceed under the assumption that $\parent{B}$ and $\parent{Z}$ are not vertex-disjoint. It follows that $\LLL{\parent{B}} \ge k/2$ or $\LLL{\parent{Z}} \ge k/2$. Without loss of generality, we assume that $\LLL{\parent{B}} \ge k/2$. (We now forget about $Z$ and $\parent{Z}$.)

Before continuing, let's take stock of the assumptions we have made so far:
\[
  G_A = P_k,\quad\
  \text{(\ref{eq:ast})},\quad\
  B \preceq A,\quad\
  v_0 \in V_B,\quad\
  |E_B| < \tsfrac{1}{2^{c/2}}k,\quad\
  |E_{\parent{B}}| = \LLL{\parent{B}} \ge k/2.
\]
Going forward, we will define vertices $v_r,v_s,v_t$ where $0 < r < s < t \le k$. The following illustration will be helpful for what follows:

\begin{figure}[H]
\centering
\includegraphics[scale=.75]{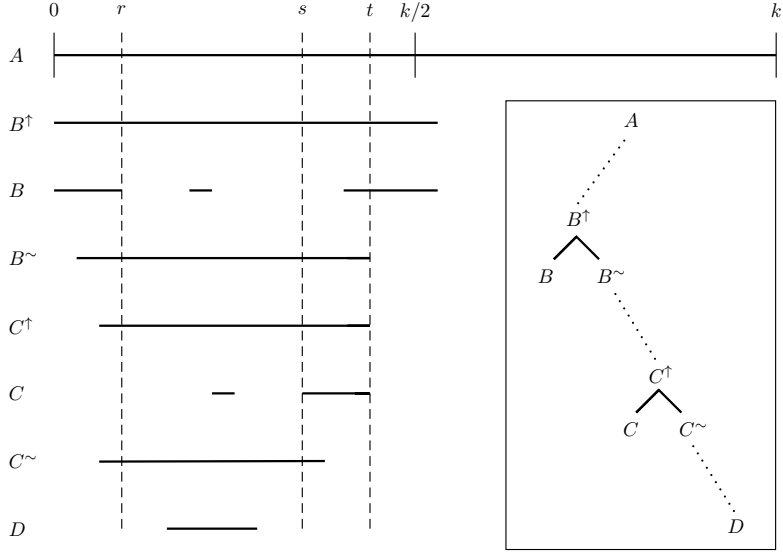}
\caption{\small Schematic of sub-join-trees in the argument.}
\end{figure}

We first define $v_r \in B$ and $v_t \in \sib{B}$ as follows: Let $\{v_0,\dots,v_r\}$ be the component of $G_B$ containing $v_0$. (That is, the component of $v_0$ in $G_B$ is a path whose initial vertex is $v_0$; let $v_r$ be the final vertex in this path.) Let $v_t$ be the vertex in $V_{\sib{B}}$ with maximal index $t$ (i.e.\ farthest away from $v_0$).

Note that $E_B$ contains edges $\edge{v_i}{v_{i+1}}$ for all $i \in \{0,\dots,r-1\} \cup \{t,\dots,\lceil k/2 \rceil-1\}$. (In the event that $t < k/2$, since $G_{\parent{B}} = G_B \cup G_{\sib{B}}$ is a path of length $\ge k/2$ and $G_{\sib{B}}$ does not contain vertices $v_{t+1},\dots,v_{\lceil k/2 \rceil}$, it follows that $G_B$ contains all edges between $v_t$ and $v_{\lceil k/2 \rceil}$.) Therefore, $r + (k/2) - t \le |E_B| < \tsfrac{1}{2^{c/2}}k$. It follows that
\[
  t-r > (\tsfrac12 - \tsfrac1{2^{c/2}})k.
\]

Next, note that $|E_{\sib{B}}| \ge |E_{\parent{B}}| - |E_B| \ge (\tsfrac12 - \tsfrac1{2^{c/2}})k > \frac{1}{2^{c-1}}k$.
We now walk down $\sib{B}$ to find a proper sub-join-tree $C \prec \sib{B}$ such that
\[
  v_t \in V_C,\qquad
  |E_C| < \tsfrac{1}{2^{c-1}}k,\qquad 
  |E_{\parent{C}}| \ge \tsfrac{1}{2^{c-1}}k.
\]
Fix any choice of such $C$. Note that $G_{\parent{C}}$ is connected by (\ref{eq:ast}).

Consider the case that $|E_{\parent{C}}| < (\tsfrac12 - \tsfrac1{2^{c/2}})k$. Since $G_{\parent{C}}$ is connected and $v_t \in V_{\parent{C}}$ and $t - r > (\tsfrac12 - \tsfrac1{2^{c/2}})k$, it follows that $V_{\parent{C}} \cap \{v_0,\dots,v_r\} = \emptyset$ and hence $\c{B \ominus \parent{C}} \ge 1$. We have
\begin{align*}
  &&\PHI{A} \ge \PHI{\parent{B}}
  &\ge \PHI{\parent{C}} + \c{B \ominus \parent{C}} + \c{\parent{B} \ominus \{B,\parent{C}\}}
    &&\text{(by $(\dag)$)}\\
  &&&\ge \PHI{\parent{C}} + 1\\
  &&&\ge \tsfrac1c\log(\LLL{\parent{C}}) + \c{\parent{C}} + 1
    &&\text{(induction hypothesis)}\\
  &&&\ge \tsfrac1c\log(\tsfrac{1}{2^{c-1}}k) + 2\\
  &&&> \tsfrac1c\log(k) + 1.
\end{align*}

Therefore, we proceed under the assumption that $|E_{\parent{C}}| \ge (\tsfrac12 - \tsfrac1{2^{c/2}})k$. Since $E_{\parent{C}} = E_C \cup E_{\sib{C}}$, we have
\[
  |E_{\sib{C}}| 
  \ge |E_{\parent{C}}| - |E_C| 
  > (\tsfrac12 - \tsfrac1{2^{c/2}} - \tsfrac{1}{2^{c-1}})k
  = \tsfrac{1}{2^{c-2}}k.
\]
We now define vertex $v_s \in V_C$. Since $v_t$ is the vertex of $G_{\sib{B}}$ with maximal index, it follows that $\edge{v_t}{v_{t+1}} \notin E_{\sib{B}}$ and hence $\edge{v_t}{v_{t+1}} \notin E_C$ (since $C \prec \sib{B}$). Therefore, the component of $G_C$ containing $v_t$ is a path with final vertex $v_t$; let $v_s$ be the initial vertex in this path. That is, $\{v_s,\dots,v_t\}$ is the component of $G_C$ which contains $v_t$.
Recall that $t-r > (\tsfrac12 - \tsfrac1{2^{c/2}})k$ and note that $t-s \le |E_C| < \tsfrac{1}{2^{c-1}}k$. Therefore,
\[
  s - r = (t - r) - (t - s) > (\tsfrac12 - \tsfrac1{2^{c/2}} - \tsfrac{1}{2^{c-1}})k = \tsfrac{1}{2^{c-2}}k.
\]

We now claim that there exists a proper sub-join-tree $D \prec \sib{C}$ such that 
\[
  \tsfrac{1}{2^{c-1}}k \le |E_D| < \tsfrac{1}{2^{c-2}}k.
\]
To see this, note that there exists a chain of sub-join-trees $\sib{C} = D_0 \succ D_1 \succ \dots \succ D_j$ such that $D_j$ is atomic and $D_{i-1} = \parent{D_i}$ and $|E_{D_i}| \ge |E_{\sib{D}_i}|$ for all $i \in \{1,\dots,j\}$. Since $|E_{D_0}| > \tsfrac{1}{2^{c-2}}k$ and $|E_{D_j}| = 1$ and $|E_{D_{i-1}}| \le |E_{D_i}| + |E_{\sib{D}_i}| \le 2|E_{D_i}|$, it must be the case that there exists $i \in \{1,\dots,j\}$ such that $\tsfrac{1}{2^{c-1}}k \le |E_{D_i}| < \tsfrac{1}{2^{c-2}}k$.

Since $|E_D| \ge \tsfrac{1}{2^{c-1}}k$, (\ref{eq:ast}) implies that $G_D$ is connected. Since $|E_D| < \tsfrac{1}{2^{c-2}}k$ and $s-r > \tsfrac{1}{2^{c-2}}k$, it follows that $V_D$ cannot contain both $v_r$ and $v_s$. We are now down to our final two cases: either $v_r \notin V_D$ or $v_s \notin V_D$.

First, suppose that $v_r \notin V_D$. We have $\c{B \ominus D} \ge 1$ and hence
\begin{align*}
  &&\PHI{A} \ge \PHI{\parent{B}}
    &\ge \PHI{D} + \c{B \ominus D} + \c{\parent{B} \ominus \{B,D\}}
    &&\text{(by $(\dag)$)}\\
  &&&\ge \PHI{D} + 1\\
  &&&\ge \tsfrac1c\log(\LLL{D}) + \c{D} + 1
    &&\text{(induction hypothesis)}\\
  &&&\ge \tsfrac1c\log(\tsfrac{1}{2^{c-1}}k) + 2\\
  &&&> \tsfrac1c\log(k) + 1.
\end{align*}
Finally, we are left with the alternative that $v_s \notin V_D$. In this case $\c{C \ominus D} \ge 1$ and hence (substituting $C$ for $B$ in the above), we have
\begin{align*}
  \PHI{A} \ge \PHI{\parent{C}}
  \ge \PHI{D} + \c{C \ominus D} + \c{\parent{C} \ominus \{C,D\}}
  \ge \PHI{D} + 1
  > \tsfrac1c\log(k) + 1.
\end{align*}

We have now covered all cases. In summary, we considered cases in the following sequence:
\begin{enumerate}[\quad\ \,1.\ ]
\setlength{\itemsep}{0pt}
  \item
    \makebox[3in]{$\c{A} \ge 2$\hfill}
    otherwise assume $G_A = P_k$ w.l.o.g.,    
 \item
    \makebox[3in]{$\exists A' \prec A$ with $\c{A'} \ge 2$ and $\LLL{A'} \ge \frac{1}{2^{c-1}}k$\hfill} 
    otherwise assume (\ref{eq:ast}),
  \item
    \makebox[3in]{
    $\parent{B}$ and $\parent{Z}$ are vertex-disjoint
    \hfill} 
    otherwise assume $|E_{\parent{B}}| \ge k/2$ w.l.o.g.,
  \item
    \makebox[3in]{$|E_{\parent{C}}| < (\tsfrac12 - \tsfrac1{2^{c/2}})k$\hfill}
    otherwise assume $|E_{\parent{C}}| \ge (\tsfrac12 - \tsfrac1{2^{c/2}})k$,
  \item
    $v_r \notin V_D$ or $v_s \notin V_D$.
\end{enumerate}
Since $\PHI{A} \ge \tsfrac1c\log(\LLL{A}) + \c{A}$ in each case, the proof is complete.
\end{proof}

\bibliographystyle{plain}
\bibliography{phi-bib}

\end{document}